\pdfoutput=1
\newif\iflncs\lncsfalse
\newif\ifmai\maifalse
\newif\iffull\fulltrue
\documentclass[figuresright]{lmcs}
\usepackage{iftex}

  \usepackage{underscore}         
  \usepackage[T1]{fontenc}        

\usepackage{wrapfig}
\usepackage{graphicx}
\usepackage{subcaption}
\usepackage[T1]{fontenc}
\usepackage[utf8]{inputenc}
\usepackage{listings}
\usepackage{graphicx}
\usepackage{algorithm}
\usepackage{tikz}
\usetikzlibrary{decorations.pathreplacing,calc}
\newcommand{\tikzmark}[1]{\tikz[overlay,remember picture] \node (#1) {};}
\usepackage{algpseudocode}
\usepackage{listings}
\usepackage{ifthen}
\usepackage{verbatimbox}
\usepackage{svg}
\usepackage{chngcntr}
\usepackage{apptools}
\usepackage{booktabs}
\usepackage{multirow}
\usepackage{graphicx}
\usepackage{pgf}
\usepackage{tikz}
\usepackage{wrapfig}
\usepackage{subcaption}
\usepackage{url}
\usepackage{listings}
\usepackage{moresize}
\usepackage{lipsum}
\usepackage{algorithm}
\usepackage{algpseudocode}
\usepackage{tikz}

\newcommand*{\AddNote}[4]{%
	\begin{tikzpicture}[overlay, remember picture]
		\draw [decoration={brace,amplitude=0.5em},decorate,very thick]
		($(#3)!(#1.north)!($(#3)-(0,1)$)$) --
		($(#3)!(#2.south)!($(#3)-(0,1)$)$)
		node [align=center,  pos=0.5, anchor=west] {#4}; 
	\end{tikzpicture}
}%

\newcommand{\reals}{{\mathbb{R}}}

\newcommand{\T}{\ensuremath{\Theta}}

\newcommand{\supp}{\mathrm{supp}}

\newcommand{\es}{\emptyset}
\usepackage{color}
\definecolor{legacyblue}{rgb}{0.12,0.31,0.65}
\definecolor{legacyred}{rgb}{0.72,0.13,0.18}
\definecolor{legacygreen}{rgb}{0.08,0.42,0.25}
\definecolor{legacyorange}{rgb}{0.69,0.34,0.04}
\definecolor{currentviolet}{rgb}{0.52,0.16,0.64}

\newcommand{\highlight}[1]{#1}
\newcommand{\minorhighlight}[1]{#1}
\newcommand{\greenhighlight}[1]{#1}
\newcommand{\orangehighlight}[1]{#1}
\newcommand{\revhighlight}[1]{#1}
\newenvironment{proofof}[1]{ {\it\noindent Proof of }
{#1}.}{ \hfill $\Box$  \vskip .3mm}

\newcommand{\X}{\mathcal{X}}

\newcommand{\St}{\mathcal{P}}

\newcommand{\F}{\mathcal{F}}

\newcommand{\sem}[1]{[\![#1]\!]}
\newcommand{\ereals}{\overline\reals}
\newcommand{\erealspl}{\overline\reals^+}
\newcommand{\expc}{\mathrm{E}}
\newcommand{\extF}{{\F}}

\newcommand{\PP}{\mathcal{P}}
\newcommand{\K}{\kappa}

\newcommand{\nil}{\text{\rm\textsf{nil}}}
\newcommand{\ppg}{\mathbf{G}}
\newcommand{\psco}{\mbox{$\mathbf{sc}$}}
\newcommand{\iomega}{\tilde\omega}
\newcommand{\Cyl}{\mathcal{C}}
\newcommand{\cy}[1]{\mathfrak{c}(#1)}

\newcommand{\term}{\text{\rm\textsf{T}}}

\newcommand{\scoz}{\mathrm{sc}}
\newcommand{\wt}{\mathrm{w}}
\newcommand{\lift}[1]{\check{#1}}
\newcommand{\FC}{\mathrm{FK}}
\newcommand{\mufk}{\phi}
\newcommand{\TSIpf}{VPF}
\newcommand{\vsp}{\vspace*{.2cm}}
\newcommand{\mik}[1]{\marginpar{ \textbf{Mic:} {\footnotesize #1}}}

\def\th@plain{%
\thm@notefont{}
\itshape 
}
\def\th@definition{%
\thm@notefont{}
\normalfont 
}
\makeatother

\iflncs{}\else \newtheorem{definition}{Definition}\fi

\iflncs{}\else \fi

\iflncs{}\else \newtheorem{theorem}{Theorem}\fi

\iflncs{}\else  \fi

\iflncs{}\else \newtheorem{lemma}{Lemma}\fi

\iflncs{}\else \newtheorem{remark}{Remark}\fi
\iflncs{}\else \newtheorem{example}{Example}\fi

\usepackage{array}
\newcommand{\fini}{\mathrm{f}}
\newcommand{\pr}{\mathrm{pr}}
\usepackage{chngcntr}
\usepackage{apptools}
\AtAppendix{\counterwithin{lemma}{section}}
\AtAppendix{\counterwithin{theorem}{section}}
\AtAppendix{\counterwithin{corollary}{section}}
\AtAppendix{\counterwithin{proposition}{section}}
\AtAppendix{\counterwithin{theorem}{section}}
\AtAppendix{\counterwithin{definition}{section}}

\usepackage{amssymb,amsmath}

\title{Parallelizable Feynman--Kac Models for Universal Probabilistic Programming}

\author{
	Michele Boreale
	\and
	Luisa Collodi
}

\address{
	Università degli Studi di Firenze, Italy.\newline
	Dipartimento di Statistica, Informatica, Applicazioni ``G. Parenti''.\footnote{Work partially supported by the project SERICS (PE00000014), EU funded NRRP MUR program -- NextGenerationEU.}
}

\email{
	michele.boreale@unifi.it,
	luisa.collodi@unifi.it
}

\begin{document}
\maketitle

\begin{abstract}
We study   provably correct and efficient instantiations of  Sequential Monte Carlo (SMC) inference  in the context of formal operational semantics of Probabilistic Programs (PPs). We focus on \emph{universal} PPs   featuring  sampling from arbitrary measures and conditioning/reweighting in unbounded loops. We first equip Probabilistic Program Graphs (PPGs), an automata-theoretic description format of PPs, with an expectation-based semantics over infinite execution traces, which also incorporates trace weights. We then prove a finite approximation theorem that provides bounds to this semantics based on expectations taken over   finite, fixed-length traces. This enables us to frame our semantics within a Feynman-Kac (FK) model, and ensures the consistency of the Particle Filtering (PF) algorithm, an instance of SMC, with respect to our semantics. Building on these results, we introduce VPF, a vectorized version of the PF algorithm tailored to PPGs and our semantics. Experiments conducted with a proof-of-concept implementation of VPF   show  very promising results compared to state-of-the-art PP inference tools.
\end{abstract}
\section{Introduction}\label{sec:intro}
%
Probabilistic Programming Languages (PPLs) \cite{introGA,introKSS} offer a   systematic approach to define arbitrarily complicated probabilistic models. One is typically interested in   performing \emph{inference} on these models,    given observed data; for example,    finding the posterior distribution of the program's output conditioned on the observed data. Here,  in the context of formal operational semantics of Probabilistic Programs, we study   provably correct and parallelizable  instantiations of  Sequential Monte Carlo (SMC) inference.

In terms of  formal semantics of PPLs, the denotational approach introduced by Kozen \cite{Kozen} offers a solid mathematical foundation. However, when it comes to practical algorithms for PPL-based inference, the landscape appears somewhat fragmented. On one hand, \emph{symbolic} and \emph{static analysis} techniques, see e.g. \cite{Sankaranarayanan,psi-solver,Hakaru,BartocciEtAl,OngGuarantees,TribastonePOPL},   yield results with correctness guarantees firmly grounded in the semantics of PPLs but often struggle with scalability.
On the other hand, practical languages and inference algorithms predominantly leverage Monte Carlo (MC) \emph{sampling} techniques (MCMC, SMC), which are more scalable but often lack a clear connection to   formal semantics   \cite{webppl,Stan,Pyro}. Notable exceptions to this situation include works such as \cite{R2,WuEtAl,LundenBorgstrom,VMCAI24,lew}, which are discussed in the related work section.

Establishing the consistency of an inference algorithm with respect to a PPL's formal semantics is not merely a theoretical pursuit. In the context of \emph{universal}  PPLs \cite{GoodmanEtal}, integration of unbounded loops and conditioning with MC sampling, which requires truncating computations at a finite time, presents significant challenges \cite{OngGuarantees}. Additionally, the interplay between continuous and discrete distributions in these PPLs can lead to complications, potentially causing existing sampling-based algorithms to yield incorrect results \cite{WuEtAl}.
In the present work, we establish a precise connection between \emph{Probabilistic Program Graphs (PPGs)}, a   general  automata-theoretic    description format of PPs, and \emph{Feynman-Kac (FK)} models, a  formalism for   state-based probabilistic processes and observations defined over a finite time horizon \cite[Ch.5]{SMC}.
This connection enables us to prove the consistency  for PPGs of the \emph{Particle Filtering (PF)} inference algorithm, one of the incarnations of Sequential Monte Carlo approach \cite[Ch.10]{SMC}. In establishing this connection, we adopt a decisively
operational perspective, as explained below.

In a PPG (Section \ref{sec:PP}), computation (essentially, {sampling}) progresses in successive stages specified by the direct edges of a graph (transitions), with nodes serving as {\emph{checkpoints}} between stages for \emph{conditioning} on observed data or more generally updating computation weights.
%
%
%
The operational semantics of PPGs is formalized in terms of   Markov kernels and score functions.
Building on this, we introduce a measure-theoretic, infinite-trace semantics (Section \ref{sec:obs}, with the necessary measure theory reviewed in Section \ref{sec:measure}). A finite approximation theorem (Section \ref{sec:FA}) then allows us to relate this trace semantics precisely to a finite-time horizon {FK model} (Section \ref{sec:MC}). PF is known to be \emph{consistent} for FK models asymptotically: as the number $N$ of simulated instances (\emph{particles}) tends to infinity, the distribution of these particles converges to the measure defined by the FK model \cite[Ch.11]{SMC}. Therefore,   consistency of PF for PPGs will automatically follow.


\orangehighlight{Our approach yields additional insights. First, the finite approximation theorem applies to a natural class of \emph{termination-determined functions} on infinite traces, whose value   depends only on finite prefixes that lead to termination. This includes standard quantities read at termination, such as termination indicators and functions of the terminal store. The theorem shows that expectations of such observables under the infinite-trace semantics can be bounded using expectations over finite traces and, for  terminal-state observables, reformulated in terms of a suitable FK model.} Second,   the automata-theoretic operational semantics  of PPGs    translates into   a \emph{vectorized}  implementation of PF,  leveraging the fine-grained,  \orangehighlight{Single Instruction Multiple Data} (SIMD) parallelism  existing at the level of particles.
Specifically, the PPG's transition function and  the score functions are applied simultaneously to the entire vector of $N$ simulated particles at each step.
This is practically significant, as modern CPUs and programming languages offer extensive support for vectorization, that may lead to dramatic speedups.
We demonstrate this aspect with a prototype vectorized implementation of a PPG-based PF algorithm using TensorFlow \cite{TF}, called VPF. Experiments comparing VPF with state-of-the-art PPLs  on challenging examples from the literature show very promising results (Section \ref{sec:experiments}).
Concluding remarks are provided in the final section (Section \ref{sec:concl}). \iffull Proofs and additional technical material have been confined to separate appendices (\ref{app:proofs}, \ref{app:PF}, \ref{app:air}, \ref{app:models}).\else Omitted proofs and additional technical material have been reported in an extended version available online \cite{BC25}.\fi

\paragraph{Main contributions}
In summary, our main contributions are as follows.
\begin{enumerate}
\item A clean operational semantics for PPGs, based on expectation taken over infinite-trace,  incorporating conditioning/reweighting.
\item A finite approximation theorem linking this semantics to finite traces and FK models, thereby establishing the consistency of PF for PPGs.
\item A vectorized version of the PF algorithm based on PPGs, and experimental evidence of its practical scalability and competitiveness.
\end{enumerate}


\paragraph{Related work}

With few notable exceptions, most work on the semantics of PPL follows  the denotational approach initiated by Kozen \cite{Kozen}; see \cite{introKSS,GB,borgstrom,introGA,gorinova1,gorinova2,scibior,Staton2,Staton3,joseph}
and references therein for representative works in this area.  In this context, a general goal orthogonal to ours, is devising methods to combine and reason on densities. Note that we
do not require that a PP induces a density on the probability space of infinite traces.


Relevant to our approach is a series of works by Lunden  et al.  on SMC inference applied to PPLs. In \cite{LundenBorgstrom}, for a lambda-calculus enriched with an explicit \textsf{resample} primitive,  consistency of PF is shown to hold, under certain restrictions,   independently of the placements of the \textsf{resample}s in the code. Operationally, their functional approach is very different from our  automata-theoretic one. In particular, they handle suspension and resumption of particles  in correspondence of resampling   via an implicit use of \emph{continuations}, in the style of webPPL \cite{webppl} and other   PPLs. The combination of functional style and continuations does not naturally lend itself to vectorization. For instance, ensuring that all particles are \emph{aligned}, that is are at a  {\textsf{resample}} point of their execution, is an issue that can impact negatively on performance or accuracy.
On the contrary, in  our automata-theoretic model, placement of \textsf{resample}s    and alignment are not issues: resampling always  happens after each  (vectorized) transition step, so all particles are automatically aligned. Note that in PPGs a transition can group together complicated, conditioning free      computations; in any case, consistency of PF is guaranteed. In a subsequent   work \cite{CorePPL,LundenAlign}, Lunden et al. study   concrete implementation issues of SMC. In  \cite{CorePPL}, they  consider \emph{PPL Control-Flow Graphs} (PCFGs), a structure intended as a target for the compilation of high-level PPLs, such as their CorePPL. The PCFG model is very similar in spirit to   PPGs, however, it lacks a formal semantics. Lunden et al. also offer an implementation of this framework,   designed to take advantage of the  potential parallelism existing at the level of particles. We compare our implementation with theirs in Section \ref{sec:experiments}.


Aditya et al. prove consistency of Markov Chain Monte Carlo (MCMC) for their PPL R2
\cite{R2}, which is based on a big-step    sampling
semantics that considers finite execution paths. No
approximation results bridging  finite and infinite traces, and hence unbounded loops, is
provided. It is also  unclear if  a big-step semantics would effectively
translate  into a SIMD-parallel algorithm. Wu et al. \cite{WuEtAl} provide the PPL Blog with a rigorous measure-theoretic semantics, formulated   in terms of Bayesian Networks, and a very efficient implementation of the PF algorithm   tailored to such networks. Again, they do not offer results for unbounded loops. In our previous work \cite{VMCAI24},
we have considered  a measure theoretic semantics for a PPL with unbounded loops, and provided a finite approximation result   and a SIMD-parallel implementation, with guarantees,  of what is in effect a \emph{rejection sampling} algorithm. Rejection may be effective for limited forms of conditioning; but it rapidly becomes wasteful and ineffective as   conditioning becomes more demanding, so to speak: e.g. when it is repeated in a loop, or the observed data have a low likelihood in the model. Finally, SMCP3 \cite{lew} provides a  rich   measure-theoretic framework for extending  the practical Gen language \cite{cusumano} with expressive proposal distributions.

A rich area in the field of PPL focuses on symbolic, exact techniques \cite{Sankaranarayanan,psi-solver,Hakaru,BartocciEtAl,OngGuarantees,TribastonePOPL,hirata} aiming to obtain  termination certificates, or certified bounds on termination probability of PPs, or even exact representations of the posterior distribution; see also  \cite{1,2,3,4,5,Wang24} for some recent works in this direction. Our goal and methodology, as already stressed, are rather different, as we focus on scalable inference
via sampling
and the ensuing consistency issues. This is part of a broader research agenda, aimed at developing flexible and scalable formal methods applicable  across   diverse, probability-related domains,  including: dynamical  systems with safety-related aspects \cite{HSCC18,SOFSEM18,LMCS19,MFCS19,IC22,VMCAI23}, information leakage and security \cite{ISC14,QIFForte14}, distributed  systems with notions of failure and recovery \cite{CaspisMSCS},   randomized model counting and  testing \cite{BG19,OutgenIEEE}.

The present work is an extended and revised version of  the conference paper \cite{BCGandalf25}. The new material here includes: additional examples (Examples \ref{ex:RW2} and \ref{ex:lift}),  a revised and expanded experimental validation (Section \ref{sec:experiments}, Appendices \ref{app:air}, \ref{app:models}),     detailed proofs (Appendix \ref{app:proofs}), and a more complete description of the standard Particle Filtering algorithm (Appendix \ref{app:PF}).



\section{Preliminaries on measure theory}\label{sec:measure}
We review a few basic concepts from measure theory following closely the presentation in the first two chapters of  \cite{Ash}, which is a reference for whatever is not explicitly described below.
Given a nonempty set $\Omega$, a \emph{sigma-field} $\F$ on $\Omega$ is a collection of subsets of $\Omega$ that contains $\Omega$, and is closed under complement and under countable disjoint union. The pair $(\Omega,\F)$ is called a \emph{measurable space}. A (total) function $f:\Omega_1\rightarrow \Omega_2$ is \emph{measurable} w.r.t. the sigma-fields $(\Omega_1,\F_1)$ and $(\Omega_2,\F_2)$ if whenever $A\in \F_2$ then $f^{-1}(A)\in \F_1$.
We let $\overline\reals=\reals\cup\{-\infty,+\infty\}$ be the set of extended reals, assuming the standard arithmetic for $\pm\infty$ (cf. \cite[Sect.1.5.2]{Ash}), and $\ereals^+$ the set of nonnegative reals including $+\infty$. The \emph{Borel sigma-field} $\F$ on $\Omega=\ereals^m$ is the minimal sigma-field that contains all {rectangles} of the form $[a_1,b_1]\times\cdots \times [a_n,b_n]$, with $a_i,b_i\in \ereals$.  An important case of measurable spaces $(\Omega,\F)$ is when $\Omega=\ereals^m$  for some $m\geq 1$  and  $\F$ is the Borel sigma-field over $\Omega$. Throughout the paper, {\emph{“measurable” means “Borel measurable”}}, both for sets and for functions. On functions, Borel measurability is preserved by composition and other elementary operations on functions; continuous real functions are Borel measurable. We will let $\F_k$ denote   the Borel sigma-field over $\overline\reals^{k}$ ($k\geq 1$) when we want to be specific about the dimension of the space.

A \emph{measure} over a measurable space $(\Omega,\F)$ is a function $\mu:\F\rightarrow \ereals^+$ that is countably additive, that is $\mu(\cup_{j\geq 1}A_j)=\sum_{j\geq 1}\mu(A_j)$ whenever $A_j$'s are pairwise disjoint sets in $\F$.
%
The \emph{Lebesgue integral} of a Borel measurable function $f$  w.r.t. a measure $\mu$ \cite[Ch.1.5]{Ash}, both  defined over a measure space  $(\Omega,\F)$, is denoted by $\int_\Omega \mu(d\omega)f(\omega)$,
with the subscript $\Omega$ omitted when clear from the context. When $\mu$ is the standard Lebesgue measure, we  may omit $\mu$ and write the integral   as   $\int_\Omega d\omega f(\omega)$.  For $A\in \F$, $\int_A \mu(d\omega)f(\omega)$ denotes $\int_\Omega \mu(d\omega)f(\omega)1_A(\omega)$, where $1_A(\cdot)$ is the indicator function of the set $A$. We let $\delta_v$ denote  Dirac's measure concentrated on $v$: for each  set $A$ in an appropriate sigma-field,  $\delta_v(A)=1$ if $v\in A$,   $\delta_v(A)=0$ otherwise. Otherwise said, $\delta_v(A)=1_A(v)$. Another measure that arises (in connection with discrete distributions) is the counting measure,
\orangehighlight{$\mu_C(A):=|A|$ if $A$ is finite,
and $\mu_C(A):=+\infty$ otherwise}. In particular, for a nonnegative $f$, we have the equality $\int_A \mu_C(d\omega)f(\omega)=\sum_{\omega\in A}f(\omega)$. \revhighlight{For a possibly infinite $A$, the sum means the supremum of its finite partial sums; the discrete densities used below have finite or countable support.}
%
%
%
A \emph{probability measure} is a   measure $\mu$ defined on $\F$ such that $\int \mu(du)=1$.
For a given nonnegative measurable function $f$ defined over $\Omega$, its \emph{expectation} w.r.t. a probability measure $\nu$  is just its integral: $\expc_\nu[f]=\int\nu(d \omega)f(\omega)$.
The following definition is central.
\begin{definition}[Markov kernel]\label{def:MK} Let $(\Omega_1,\F_1)$ and $(\Omega_2,\F_2)$ be   measurable spaces. A function $K: \Omega_1\times \F_2 \longrightarrow \erealspl$ is a \emph{Markov kernel}  from $\Omega_1$ to $\Omega_2$  if  it satisfies the following properties:
	\begin{enumerate}
		\item for each $\omega\in \Omega_1$, the function $K(\omega,\cdot):\F_2\rightarrow \erealspl$ is a probability measure on $(\Omega_2,\F_2)$;
		\item for each $A\in \F_2$, the function $K(\cdot,A): \Omega_1 \rightarrow \erealspl$ is   measurable.
	\end{enumerate}
\end{definition}

Notationally, we will most often write $K(\omega,A)$ as $K(\omega)(A)$.
The following is a standard result about  the  construction  of finite product of measures  over a  product space\footnote{We shall freely identify language-theoretic \emph{words} with \emph{tuples}.}
  $\Omega^t=\Omega\times \cdots \times \Omega$ ($t$ times) for $t\geq 1$ an integer \iffull  (see Theorem \ref{th:prode} in the Appendix for a more detailed statement). \else. \fi It is customary to denote the measure $\mu^t$ defined by  the theorem also as $\mu^1\otimes K_2\otimes\cdots \otimes K_t$.

\begin{theorem}[product of measures, \cite{Ash},Th.2.6.7]\label{th:prod}
	Let  $t\geq 1$ be an integer. Let $\mu^1$ be a probability measure on $\Omega$ and $K_2,...,K_t$ be $t-1$   (not necessarily distinct) Markov kernels from $\Omega$ to $\Omega$. Then there is a unique probability measure $\mu^t$ defined on $(\Omega^t,\F^t)$ such that for every $A_1\times \cdots\times A_t\in \F^t$ we have:
	$\mu^t(A_1\times \cdots\times A_t) = \int_{A_1} \mu^1(d\omega_1)\int_{A_2}K_2(\omega_1)(d\omega_2)\cdots \int_{A_t}K_t(\omega_{t-1})(d\omega_t)
	$.
\end{theorem}

\section{Probabilistic programs}\label{sec:PP}
We first  introduce a general formalism for specifying programs, in the form of certain graphs  that can be regarded as symbolic finite automata. For this formalism, we introduce then an operational semantics  in terms of  Markov kernels. 
\paragraph{Probabilistic Program Graphs}\label{sec:PPG}
In defining probabilistic programs, we will   rely on  a repertoire of basic distributions: 
%
continuous, discrete and mixed distributions will be allowed. A crucial point for   expressiveness  is that a measure may depend on  \emph{parameters}, whose   value at runtime is determined by the state of the program. To ensure that the resulting programs  define measurable functions (on a suitable space), it is important that the dependence between the parameters and the measure be in turn of measurable type.  We will formalize this in terms of Markov kernels. Additionally, we will   consider score functions, a generalization of 0/1-valued predicates.
%
Formally, we will consider the     two families of functions defined below. In the definitions, we will let  $m\geq 1$ denote a fixed integer, representing the number of \emph{variables} in the program,   conventionally referred  to as $x_1,...,x_m$. We will let $v$ range over $\ereals^m$,  the content of the program variables in a given state, or \emph{store}.
\begin{itemize}
	
	\item \emph{Score functions}:    measurable functions  $\gamma:\ereals^m\rightarrow [0,1]$.
	A \emph{predicate} is a special case of a score function $\varphi:\ereals^m\rightarrow \{0,1\}$. An Iverson bracket style notation will be often employed, e.g.: $[x_1\geq 1]$ is the   predicate  that on input $v$ yields 1 if $v_1\geq 1$, 0 otherwise.  
	\item \emph{Parametric measures}:   Markov kernels       $\zeta:\ereals^m \times \F_m \rightarrow [0,1] $.
\end{itemize}
For a parametric measure $\zeta$ and a store $v\in \ereals^m$,  $\zeta(v)$  is a distribution, that can be used to  sample a new store  $v'\in \ereals^m$   depending on the current program store $v$. Analytically, $\zeta$   may be expressed by,  for instance,   chaining together sampling of individual components of the store. This can be done by relying on \emph{parametric densities}:  measurable functions $\rho:\ereals^m\times \ereals \rightarrow \ereals^+$ such that, for a designated measure $\mu_\rho$, the function
$(v,A)\mapsto \int_A\mu_\rho(dr)\,\rho(v,r)$ ($A\in \minorhighlight{\F_1}$) is a Markov kernel from $\ereals^m$ to $\ereals$. This is explained via the following example.

\begin{example}{\label{ex:parametricD}
		Fix $m=2$. Consider the Markov kernel defined as follows, for each $x_1,x_2\in \ereals$ and $A\in \F_2$
		\vspace{-0.2cm}
		\begin{equation}\label{eq:zeta}
			\begin{array}{c}
				\zeta(x_1,x_2)(A):=\int \mu_1(d r_1)\, \left(\rho_1(x_1,x_2,r_1)\cdot \int\mu_2(d r_2)\rho_2(r_1,x_2,r_2)1_A(r_1,r_2)\right)
			\end{array}
		\end{equation}
		\noindent
		where:
\begin{itemize}
\item $\mu_1=\mu_C$ is the counting measure;
\item $\rho_1(x_1,x_2,r)=\frac 1 2 1_{\{x_1\}}(r)+\frac 1 2 1_{\{x_2\}}(r)$ is the density of a discrete distribution on $\{x_1,x_2\}$;
\item $\mu_2=\mu_L$ is the ordinary Lebesgue measure;
\item $\rho_2(x_1,x_2,r)=N(x_1, x_2 ,r):=\frac 1 { |x_2|\sqrt{2\pi} }\exp(-\frac 1 2 (\frac{ r-x_1}{|x_2|} )^2 )$ is the  density of the Normal distribution of mean $x_1$ and standard deviation\footnote{With the proviso that,  when $x_2=0$ or $|x_1|,|x_2|=+\infty$, $N(x_1,x_2,r)$ denotes an arbitrarily fixed, default probability density.} $|x_2|$.
\end{itemize}
The function $\zeta$ is a parametric measure: concretely, it corresponds to  first sampling uniformly $r_1$ from the set $\{x_1,x_2\}$,   then sampling $r_2$ from the Normal distribution of mean $r_1$ and s.d. $|x_2|$ (if $|x_2|$ is positive and finite, otherwise from a default distribution). Rather than via \eqref{eq:zeta}, we will describe $\zeta$   via the following more handy notation: 
		$$
		\begin{array}{l}
			r_1\sim \rho_1(x_1,x_2); \\
           \,r_2\sim \rho_2(r_1,x_2);\\
			\zeta(x_1,x_2)(A):=\int \mu_1(d r_1)\, \left(\rho_1(x_1,x_2,r_1)\cdot \int\mu_2(d r_2)\rho_2(r_1,x_2,r_2)1_A(r_1,r_2)\right)
		\end{array}.
		$$

Note that the sampling order from top to bottom is relevant here.
}\end{example}

In fact, as far as the formal framework of PPGs introduced below is concerned, how the parametric measures $\zeta$'s are analytically described is irrelevant. From the practical point of view, it is important   we know how to (efficiently) sample from the measure $\zeta(v)$,  for any $v$, in order for  the inference algorithms to be actually implemented (see Section \ref{sec:MC}). In concrete terms,    $\zeta(v)$ might represent  the (possibly unknown) distribution of the outputs in $\ereals^m$ returned by a piece of code, when invoked with input $v$.
Another important special case  of parametric measure is the following. For any   $v=(v_1,...,v_m)\in \overline\reals^m$,  $r\in \ereals$ and $1\leq i \leq m$,   let $v[r@ i]:=(v_1,...,r,...,v_m)$ denote  the tuple where $v_i$ has been replaced by $r$.  Consider the parametric measure  $\zeta(v) = \delta_{v[g(v)@i]}$, where $g:\ereals^m\rightarrow \ereals$ is a measurable function. In programming terms, this corresponds to the deterministic  \emph{assignment} of the value $g(v)$ to the variable $x_i$. We will describe this $\zeta$ as: $x_i:=g(x_1,...,x_m)$.

In the definition of PPG below, one may think of   the computation (sampling) taking place in successive stages    on the  edges (transitions) of the graph, with   nodes serving as \emph{checkpoints} (a term we have borrowed from \cite{LundenBorgstrom})  between  stages for conditioning on observed data --- or, more generally, re-weighting the score assigned to a computation. The edges also account for the control flow among the different stages  via predicates computed on the store of the source nodes.

\lstset{
basicstyle=\ttfamily,
mathescape
}

\captionsetup{width=\textwidth}

\begin{figure}[t!]
\hspace{-1.9cm}
\begin{minipage}{.27\linewidth}
\vspace{-0.25cm}
{\scriptsize
	\begin{lstlisting}
c$\sim$B(1/2);
if (c==0) skip
else {
    d$\sim$B(1/2);
    observe(d==1);
    skip
}
	\end{lstlisting}
}
\vspace{-0.3cm}
\end{minipage}
\hspace{-1.5cm}
\begin{minipage}{.22\linewidth}
{\small
	\begin{center}
		\tikzset{every picture/.style={line width=0.75pt}} 
		\begin{tikzpicture}[x=0.75pt,y=0.75pt,yscale=-.7,xscale=.7]
			%
			\draw   (185.13,65.54) .. controls (185.13,72.4) and (190.69,77.96) .. (197.54,77.96) .. controls (204.4,77.96) and (209.96,72.4) .. (209.96,65.54) .. controls (209.96,58.68) and (204.4,53.12) .. (197.54,53.12) .. controls (190.69,53.12) and (185.13,58.68) .. (185.13,65.54) -- cycle ;
			\draw   (185.13,125.54) .. controls (185.13,132.4) and (190.69,137.96) .. (197.54,137.96) .. controls (204.4,137.96) and (209.96,132.4) .. (209.96,125.54) .. controls (209.96,118.68) and (204.4,113.12) .. (197.54,113.12) .. controls (190.69,113.12) and (185.13,118.68) .. (185.13,125.54) -- cycle ;
			\draw   (233.13,176.54) .. controls (233.13,183.4) and (238.69,188.96) .. (245.54,188.96) .. controls (252.4,188.96) and (257.96,183.4) .. (257.96,176.54) .. controls (257.96,169.68) and (252.4,164.12) .. (245.54,164.12) .. controls (238.69,164.12) and (233.13,169.68) .. (233.13,176.54) -- cycle ;
			\draw   (140.56,176.85) .. controls (140.56,183.71) and (146.12,189.27) .. (152.98,189.27) .. controls (159.84,189.27) and (165.4,183.71) .. (165.4,176.85) .. controls (165.4,170) and (159.84,164.44) .. (152.98,164.44) .. controls (146.12,164.44) and (140.56,170) .. (140.56,176.85) -- cycle ;
			\draw   (143.25,176.85) .. controls (143.25,182.23) and (147.61,186.58) .. (152.98,186.58) .. controls (158.35,186.58) and (162.71,182.23) .. (162.71,176.85) .. controls (162.71,171.48) and (158.35,167.12) .. (152.98,167.12) .. controls (147.61,167.12) and (143.25,171.48) .. (143.25,176.85) -- cycle ;
			\draw    (196.69,39.23) -- (197.36,50.13) ;
			\draw [shift={(197.54,53.12)}, rotate = 266.49] [fill={rgb, 255:red, 0; green, 0; blue, 0 }  ][line width=0.08]  [draw opacity=0] (8.93,-4.29) -- (0,0) -- (8.93,4.29) -- cycle    ;
			\draw    (197.54,77.96) -- (197.54,110.12) ;
			\draw [shift={(197.54,113.12)}, rotate = 270] [fill={rgb, 255:red, 0; green, 0; blue, 0 }  ][line width=0.08]  [draw opacity=0] (8.93,-4.29) -- (0,0) -- (8.93,4.29) -- cycle    ;
			\draw    (206.5,135.6) -- (235.45,166.41) ;
			\draw [shift={(237.5,168.6)}, rotate = 226.79] [fill={rgb, 255:red, 0; green, 0; blue, 0 }  ][line width=0.08]  [draw opacity=0] (8.93,-4.29) -- (0,0) -- (8.93,4.29) -- cycle    ;
			\draw    (188.5,135.6) -- (163.42,166.79) ;
			\draw [shift={(161.54,169.12)}, rotate = 308.8] [fill={rgb, 255:red, 0; green, 0; blue, 0 }  ][line width=0.08]  [draw opacity=0] (8.93,-4.29) -- (0,0) -- (8.93,4.29) -- cycle    ;
			\draw    (233.13,176.54) -- (168.4,176.84) ;
			\draw [shift={(165.4,176.85)}, rotate = 359.74] [fill={rgb, 255:red, 0; green, 0; blue, 0 }  ][line width=0.08]  [draw opacity=0] (8.93,-4.29) -- (0,0) -- (8.93,4.29) -- cycle    ;
			\draw    (147.5,163.6) .. controls (133.02,137.54) and (112.03,178.54) .. (137.59,177.17) ;
			\draw [shift={(140.56,176.85)}, rotate = 171.03] [fill={rgb, 255:red, 0; green, 0; blue, 0 }  ][line width=0.08]  [draw opacity=0] (8.93,-4.29) -- (0,0) -- (8.93,4.29) -- cycle    ;
			
			\draw (191,58) node [anchor=north west][inner sep=0.75pt]    {$0$};
			\draw (192,119) node [anchor=north west][inner sep=0.75pt]    {$1$};
			\draw (240,170) node [anchor=north west][inner sep=0.75pt]    {$3$};
			\draw (147,170) node [anchor=north west][inner sep=0.75pt]{$2$};
			\draw (210,77) node [anchor=north west][inner sep=0.75pt]  [font=\scriptsize]  {${\textstyle c\sim B\left(\frac{1}{2}\right)}$};
			\draw (225,123) node [anchor=north west][inner sep=0.75pt]  [font=\scriptsize]  {$ \begin{array}{l}
					[{\textstyle c\neq 0}],\\
					{\textstyle d\sim B\left(\frac{1}{2}\right)}
				\end{array}$};
			\draw (140,128) node [anchor=north west][inner sep=0.75pt]  [font=\footnotesize]  {$[{\scriptstyle c=0}]$};
			\draw (258,178) node [anchor=north west][inner sep=0.75pt]  [font=\scriptsize]  {$\gamma =[d=1]$};

		\end{tikzpicture}	
\end{center}}
\end{minipage}\hspace*{.5cm}
\begin{minipage}{.25\linewidth}

{\scriptsize
	\begin{lstlisting}
d$\sim$U(0,2); r$\sim$U(0,1);
x:=-1; y:=1;
while(|x-y|$\geq$.1) {
     x$\sim$N(x,d); y$\sim$N(y,r);
     observe(|x-y|$\leq$3);
}
	\end{lstlisting}
}
\vspace{-0.3cm}
\end{minipage}
\begin{minipage}{.2\linewidth}
{\small
	\begin{center}
		\tikzset{every picture/.style={line width=0.75pt}} 
		\begin{tikzpicture}[x=0.75pt,y=0.75pt,yscale=-.7,xscale=.7]
			\draw   (221.39,117.28) .. controls (221.49,110.04) and (227.23,104.26) .. (234.2,104.37) .. controls (241.17,104.48) and (246.73,110.44) .. (246.62,117.69) .. controls (246.52,124.93) and (240.78,130.71) .. (233.81,130.6) .. controls (226.84,130.49) and (221.28,124.53) .. (221.39,117.28) -- cycle ;
			\draw    (246.62,117.69) -- (301.63,118.26) -- (327.5,118.26) ;
			\draw [shift={(330.5,118.26)}, rotate = 180] [fill={rgb, 255:red, 0; green, 0; blue, 0 }  ][line width=0.08]  [draw opacity=0] (8.93,-4.29) -- (0,0) -- (8.93,4.29) -- cycle    ;
			\draw   (330.5,118.26) .. controls (330.5,110.25) and (336.74,103.76) .. (344.45,103.76) .. controls (352.15,103.76) and (358.4,110.25) .. (358.4,118.26) .. controls (358.4,126.26) and (352.15,132.76) .. (344.45,132.76) .. controls (336.74,132.76) and (330.5,126.26) .. (330.5,118.26) -- cycle ;
			\draw    (358.4,118.26) .. controls (396.51,104.16) and (341.53,67.77) .. (344.13,101.06) ;
			\draw [shift={(344.45,103.76)}, rotate = 261.3] [fill={rgb, 255:red, 0; green, 0; blue, 0 }  ][line width=0.08]  [draw opacity=0] (8.93,-4.29) -- (0,0) -- (8.93,4.29) -- cycle    ;
			\draw    (344.45,132.76) -- (344.91,180.87) ;
			\draw [shift={(344.93,182.87)}, rotate = 269.45] [color={rgb, 255:red, 0; green, 0; blue, 0 }  ][line width=0.75]    (10.93,-3.29) .. controls (6.95,-1.4) and (3.31,-0.3) .. (0,0) .. controls (3.31,0.3) and (6.95,1.4) .. (10.93,3.29)   ;
			\draw   (330.98,197.37) .. controls (330.98,189.36) and (337.22,182.87) .. (344.93,182.87) .. controls (352.63,182.87) and (358.88,189.36) .. (358.88,197.37) .. controls (358.88,205.37) and (352.63,211.87) .. (344.93,211.87) .. controls (337.22,211.87) and (330.98,205.37) .. (330.98,197.37) -- cycle ;
			\draw   (333.18,197.37) .. controls (333.18,190.92) and (338.44,185.7) .. (344.93,185.7) .. controls (351.42,185.7) and (356.68,190.92) .. (356.68,197.37) .. controls (356.68,203.81) and (351.42,209.04) .. (344.93,209.04) .. controls (338.44,209.04) and (333.18,203.81) .. (333.18,197.37) -- cycle ;
			\draw    (358.88,197.37) .. controls (390.13,196.49) and (368.8,160.14) .. (355.82,185.51) ;
			\draw [shift={(354.64,188.04)}, rotate = 292.96] [fill={rgb, 255:red, 0; green, 0; blue, 0 }  ][line width=0.08]  [draw opacity=0] (8.93,-4.29) -- (0,0) -- (8.93,4.29) -- cycle    ;
			
			\draw (259.34,43.82) node [anchor=north west][inner sep=0.75pt]  [font=\scriptsize,rotate=-359.58,xslant=0]  {$ \begin{array}{l}
					d\sim U( 0,2) ;\\
					r\sim U( 0,1) ;\\
					x:=-1;\\
					y:=1
				\end{array}$};
			\draw (375.69,60.58) node [anchor=north west][inner sep=0.75pt]  [font=\scriptsize,rotate=-359.84,xslant=0]  {$ \begin{array}{l}
					[|x-y|\geq 0.1] ,\\
					x\sim N( x,d) ;\\
					y\sim N( y,r)
				\end{array}$};
			\draw (245.62,150.52) node [anchor=north west][inner sep=0.75pt]  [font=\scriptsize,rotate=-359.84,xslant=0]  {$[|x-y|< 0.1]$};
			\draw (229.23,109.42) node [anchor=north west][inner sep=0.75pt]    {$0$};
			\draw (339.55,110.29) node [anchor=north west][inner sep=0.75pt]    {$1$};
			\draw (339.55,189.62) node [anchor=north west][inner sep=0.75pt]    {$2$};
			\draw (374.76,174.6) node [anchor=north west][inner sep=0.75pt]  [font=\footnotesize]  {};
			\draw (354.99,122.04) node [anchor=north west][inner sep=0.75pt]  [font=\scriptsize]  {$\gamma =[ \ |x-y|\leq 3\ ]$};
		\end{tikzpicture}		
\end{center}}
\end{minipage}
\caption{\scriptsize \textbf{Left}. The PPG of Example \ref{ex:simple0} and a corresponding pseudo-code. The    $\nil$ node (2) is distinguished with a double border. Constant 1 predicates and score functions, and the identity function are not displayed in transitions. The score function $\gamma$ decorates node 3, that is $\psco(3)=\gamma$. \textbf{Right}. The PPG for the drunk man and   mouse random walk of Example \ref{ex:dm1} and a corresponding pseudo-code.   The score function $\gamma$ decorates node 1, that is $\psco(1)=\gamma$.}\label{fig:drunk}
\end{figure}


\begin{definition}[PPG]\label{def:PPG}
Fix $m\geq 1$. A \emph{Probabilistic Program Graph (PPG)} on $\ereals^m$ is a 4-tuple $\ppg=(\St,E,\nil,\psco)$ satisfying the following.
\begin{itemize}
\item $\St=\{S_1,...,S_k\}$ is a finite, nonempty set of \emph{program checkpoints} (\emph{programs}, for short).
\item $E$ is a finite, nonempty set of \emph{transitions} of the form $(S,\varphi,\zeta,S')$, where: $S,S'\in \St$ are called the \emph{source} and \emph{target} program checkpoint, respectively; $\varphi:\ereals^m\rightarrow\{0,1\}$ is a predicate;  and  $\zeta:\ereals^m \times \F_m \rightarrow [0,1] $ is a parametric measure.
\item $\nil\in \St$ is a distinguished \emph{terminated} program checkpoint, such that $(\nil,1,\mathrm{id},\nil)$ ($\mathrm{id}$ = identity) is the only transition in $E$ with $\nil$ as source.
\item $\psco$ is a mapping from $\St$ to the set of score functions, s.t. $\psco(\nil)$ is the constant  1.
\end{itemize}
Additionally,   denoting by $E_S$ the set of transitions in $E$ with $S$ as a source checkpoint, the following \emph{consistency} condition is   assumed: for each $S\in \St$, the function $\sum_{(S,\varphi,\zeta,S')\in E_S}\varphi $  is  the constant 1.
\end{definition}

We first illustrate Definition \ref{def:PPG} with a simple example. This will also serve to illustrate the finite approximation theorem in the next section.

\begin{example}\label{ex:simple0}
Consider the PPG in Fig. \ref{fig:drunk}, left. Here we have $m=2$    and $B(p)$ is the Bernoulli distribution with success probability $p$.  On the left, a more conventional pseudo-code      notation for the resulting program.
We will not pursue a systematic formal translation from this program notation to PPGs, though.
\end{example}

The following example illustrates the use of scoring functions inside loops. While a bit contrived, it is close to the structure of more significant scenarios, such as the aircraft tracking example of \cite{WuEtAl}, cf. Section \ref{sec:experiments}.

\begin{example}[Of mice and drunk men]{\label{ex:dm1}
On  a street, a  drunk man and a mouse perform independent  random walks  starting at conventional positions $-1$ and $1$ respectively. Initially, each of them samples   a standard deviation (s.d.) from a uniform distribution.
 Then, at each discrete time step, they independently sample  their own next position from a Normal distribution centered at the current position with  the  s.d. chosen at the beginning. The process is stopped as soon as the man and the mouse meet, which we take to mean the distance between them is $<1/10$.

It has been suggested that  in certain urban areas   a man is never more than 3m away from a   mouse \cite{TheGuardian}. Taking this    information at face value, we incorporate it into   our model  with the score function $\gamma:=[|x-y|\leq 3]$. The resulting PPG is described in Figure \ref{fig:drunk}, right (cf. also pseudo-code).
}\end{example}

\iffull
\begin{remark}{ 
The definition of score function  requires that $\gamma$ takes values on $[0,1]$. This restriction is necessary \orangehighlight{for the weight function to be defined on infinite traces (cf. \eqref{eq:wt})}.
A similar restriction is found in e.g. \cite{borgstrom}. \highlight{In practice, a bounded nonnegative score can be rescaled into $[0,1]$ by dividing it by a suitable positive constant. Such a rescaling does not in general preserve the trace semantics: it is immaterial when every trace under consideration evaluates the rescaled score the same number of times---for example, in certain fixed-horizon models---since only then the same multiplicative factor is introduced in both the numerator and the denominator.} On the other hand, we cannot use unbounded score functions, like say, $\gamma(x_1,x_2)=N(x_1,x_2,x_2)$.
}
\end{remark}
\fi

\paragraph{Operational semantics of PPGs}\label{sec:PPGsem}
For any given    PPG $\ppg$, we will define a Markov kernel $\K_{\ppg}(\cdot,\cdot)$ that describes its operational semantics. From now on, we will consider one arbitrarily  fixed  PPG, $\ppg=(\St,E,\nil,\psco)$ and just drop the subscript ${}_{\ppg}$ from the notation. Let us also remark that the scoring function $\psco(\cdot)$ will play no role in the definition of the Markov kernel --- it will come into play in the trace based semantics of Section \ref{sec:obs}.

Some additional notational shorthand  is in order. First, we   identify $\PP$ with   the finite set of naturals $\{0,...,|\St|-1\}$. With this convention,  we have that  $\ereals^m\times \PP\subseteq  \ereals^{m+1}$. Henceforth, we define   our  state space and sigma-field as follows:
\begin{align*}
\Omega :=\overline\reals^{m+1} \quad \quad\quad
\extF := \text{Borel sigma-field over $ \ereals^{m+1}$.}
\end{align*}
We keep the symbol $\F_k$ for the Borel sigma-field over $\overline\reals^{k}$, for any $k\geq 1$.
%
For any $S\in \PP$ and $A\in \extF$, we let $A_S:=\{v\in \ereals^m\,:\,(v,S)\in A\}$ be the \emph{section} of $A$ at $S$. Note that $A_S\in\F_m$, as sections of measurable sets are measurable, see \cite[Th.2.6.2, proof(1)]{Ash}. 

\begin{definition}[PPG Markov kernel]\label{def:ppgMK}
The function $\K:\Omega\times \extF\rightarrow \reals^+$ is defined as follows, for each $\omega\in \Omega$ and $A\in \extF$:
\begin{equation}\label{eq:PPGMK}
\K(\omega)(A):= \left\{\begin{array}{ll}
	\delta_\omega(A) & \text{if $\omega\notin \ereals^m\times \PP$}\\
	\sum_{(S,\varphi,\zeta,S')\in E_S} \,\varphi(v) \cdot \zeta(v)\left(A_{S'}\right)& \text{if $\omega=(v,S)\in \ereals^m\times \PP$.}
\end{array}\right.
\end{equation}
\end{definition}

\begin{lemma}\label{lemma:MKO} 
The function $\K$ is  a Markov kernel from $\Omega$ to $\Omega$.
\end{lemma}


\section{Trace semantics for PPGs}\label{sec:obs}
In what follows, we fix an arbitrary PPG, $\ppg=(\St,E,\nil,\psco)$  and let $\K$ denote the induced Markov kernel, as per Definition \ref{def:ppgMK}.
For any $t\geq 1$,  we call $\Omega^t$   the set of \emph{paths of length $t$}. Consider now the set of paths of infinite length,  $\Omega^\infty$, that is the set of infinite sequences $\iomega=(\omega_1,\omega_2,...)$ with $\omega_i\in\Omega$. For any   $\omega^t\in \Omega^t$ and $\iomega\in \Omega^\infty$, we identify     the pair $(\omega^t,\iomega)$ with the element of $\Omega^\infty$ in which the prefix $\omega^t$ is followed by $\iomega$.  For $t\geq 1$ and a measurable $B_t\subseteq \Omega^t$, we let $\cy{B_t}:=B_t\cdot \Omega^\infty\subseteq \Omega^\infty$ be the \emph{measurable cylinder} generated by $B_t$. We let $\Cyl$ be the minimal sigma-field over $\Omega^\infty$ generated by all measurable cylinders. Under the same assumptions of  Theorem \ref{th:prod} on the measure $\mu^1$ and on the kernels $K_2,K_3,...$ there exists a unique measure $\mu^\infty$ on $\Cyl$ such that for each $t\geq 1$ and each measurable cylinder $\cy{B_t}$, it holds that $\mu^\infty(\cy{B_t})=\mu^t(B_t)$: see \cite[Th.2.7.2]{Ash}, also   known as the \emph{Ionescu-Tulcea  theorem}. In the definition below,  we let $0=(0,...,0)$ ($m$ times) and consider $\delta_{(0,S)}$, the Dirac's measure on $\Omega$ that concentrates all the probability mass in $(0,S)$.

\begin{definition}[probability measure induced by $S$]\label{def:measPP}
Let  $S\in \PP$. For each integer $t\geq 1$, we let $\mu^t_S$ be the probability measure over $\Omega^t$ uniquely defined by Theorem \ref{th:prod}
 by letting $\mu^1=\delta_{(0,S)}$ and $K_2=\cdots =K_t=\K$.
We let $\mu^\infty_S$ be the unique probability measure on $\Cyl$ induced by $\mu_1$ and $K_2=\cdots =K_t=\cdots=\K$, as determined by the Ionescu-Tulcea theorem.
%
\end{definition}

In other words, $\mu^t_S=\delta_{(0,S)}\otimes \K\otimes \cdots \otimes \K$ ($t-1$ times $\K$). By convention,  if $t=1$, $\mu^t_S=\delta_{(0,S)}$. The measure $\mu^\infty_S$ can be informally interpreted as the limit of the measures $\mu^t_S$ and represents the semantics of $S$.

Recall that the \emph{support} of an (extended) real valued function $f$ is the set $\supp(f):=\{z\,:\,f(z)\neq 0\}$. In what follows, \emph{we shall concentrate on nonnegative measurable functions} $f$ to avoid unnecessary complications with the existence of integrals. General functions can be dealt with by the usual trick of decomposing $f$ as $f=f^+ - f^-$, where $f^+=\max(0,f)$ and $f^-=-\min(0,f)$, and then dealing separately with $f^+$ and $f^-$, \minorhighlight{provided that the two corresponding integrals are not both infinite (so that the indeterminate form $+\infty-\infty$ does not arise)}. Let us introduce a  \emph{combined score function}   $\scoz:\Omega\rightarrow [0,1]$ as follows, for each $\omega=(v,S)$:
\begin{equation}\label{eq:scoz}
\scoz(\omega):= \left\{\begin{array}{ll}\psco(S)(v) &\text{ if $\omega=(v, S)\in \ereals^m\times \PP$}\\
1 & \text{ otherwise.}
\end{array}\right.
\end{equation}
The function $\scoz(\cdot)$ is extended to a \emph{weight function} on infinite traces, $\wt:\Omega^\infty \rightarrow [0,1]$ by letting\footnote{Note that $\wt(\tilde\omega)$ is well-defined
because $0\leq \scoz(\omega_j)\leq 1$ for each $j\geq0$, so the sequence of partial products is  nonincreasing.}, for any $\tilde\omega=(\omega_1,\omega_2,...)\in \Omega^\infty$:
{\small
\begin{equation}\label{eq:wt}
\wt(\tilde\omega):=\Pi_{j\geq 1}\scoz(\omega_j)\,.
\end{equation}
}\noindent
For each $t\geq 1$, we define the weight function truncated at time $t$,    $\wt_t:\Omega^t\rightarrow [0,1]$, by $\wt_t(\omega^t):=\Pi_{j=1}^t\scoz(\omega_j)$.
Both $\wt$ and $\wt_t$ ($t\geq 1$) are measurable functions  on the respective domains\iffull (see Lemma \ref{lemma:aux})\fi. We arrive at the definition of the semantics of programs. We consider the ratio of the unnormalized semantics ($[S]f$) to the weight of all traces, terminated or not ($[S]\wt$). In the special case when the score functions represent conditioning, this choice corresponds to quotienting over the probability of \emph{non failed} traces.     In PPL, quotienting over non failed states is   somewhat standard: see e.g. the discussion  in \cite[Section 8.3.2]{KaminskiTh}.

\begin{definition}[trace semantics]\label{def:obsPP}
Let $f$ be a nonnegative measurable function defined on $\Omega^\infty$.  We let the \emph{unnormalized semantics of $S$ and $f$} be $[S] f := \expc_{\mu^\infty_S}[f]\,(=\int \mu^\infty_S(d \iomega)f(\iomega))$. 
We let
\begin{align}\label{eq:semPP}
\sem{S}f&:=\dfrac{[S] (f\cdot \wt)  }{[S] \wt} 
\end{align}
provided   the denominator above is $>0$; otherwise $\sem{S}f$ is undefined. 
\end{definition}
\vspace{-0.5cm}

\section{Finite approximation}\label{sec:FA}
The operational semantics of a probabilistic program is defined over infinite traces, due to the possibility of unbounded loops. Yet in practice,
we can  sample only \emph{finite} traces. The main result of this section, Theorem \ref{th:approx}, provides a rigorous way to approximate expectations over infinite traces by computing expectations over finite prefixes.
Intuitively, the theorem says that if we truncate all traces at a fixed length $t$, and restrict our attention
to those that have already terminated by that point, then we can compute lower and upper bounds for
the expectation of $f$. Moreover, the bounds converge to the true value as $t\rightarrow+\infty$, if the program is guaranteed to
terminate in \orangehighlight{finite time with probability $1$ (Theorem \ref{th:tight}).}  

\orangehighlight{The approximation result will apply to observables whose value is fixed once the execution has terminated. We first make this condition precise.}

\begin{definition}[termination]\label{def:term}
Let $\term:=\ereals^m\times \{\nil\}$ be the set of \emph{terminated} states and let $\term^{\mathrm{c}}$ denote its complement. We let  $T_j\subseteq\Omega^j$ ($j\geq 0$) be the set of finite sequences that \emph{terminate at time} $j$, that is: $T_0:=\es$ and   $T_j:=(\term^{\mathrm{c}})^{j-1}\cdot \term$, for $j\geq 1$. We let $T_{\fini}:=\cup_{t\geq 0}\cy{T_t}\subseteq \Omega^\infty$ denote the set of infinite sequences that \emph{terminate in finite time}.
\end{definition}

\orangehighlight{For $t\geq 1$, let $T^{\leq t}:=\cup_{j=1}^t T_j\cdot\Omega^{t-j}\subseteq\Omega^t$ be the set of paths of length $t$ that have terminated by time $t$. Thus $\cy{T_t}$ is the set of infinite paths that terminate exactly at time $t$, while $\cy{T^{\leq t}}$ is the set of infinite paths that terminate by time $t$.}

\begin{definition}[termination-determined function]\label{def:termf}
\orangehighlight{A nonnegative measurable function $f:\Omega^\infty\rightarrow\erealspl$ is \emph{termination-determined} if, for every $j\geq 1$, every $\omega^j\in T_j$, and every $\tilde\omega,\tilde\omega'\in\Omega^\infty$,}
\begin{align*}
\orangehighlight{f(\omega^j,\tilde\omega)=f(\omega^j,\tilde\omega').}
\end{align*}
\end{definition}

\orangehighlight{In other words, the above definition says that once in  trace  termination is reached, the value of $f$ is independent of what is appended after the terminating prefix.}
\revhighlight{For observables specified in terms of a finite computation, the condition can be checked directly: any value computed from the prefix up to first termination, and unaffected by later states, is termination-determined. We do not assume a general decision procedure for arbitrary measurable observables.}

\begin{example}\label{ex:f}
\orangehighlight{The indicator $1_{T_{\fini}}$ is termination-determined. For the PPG in Example \ref{ex:dm1}, consider also the functions $f_1$, which returns the time at which the process terminates (and $0$ on nonterminating traces), and $f_2$, which returns the value of $d$ at termination when $d\in(0,2]$ (and $0$ otherwise). Both are termination-determined: after the first terminated state has been reached, neither value can be affected by any continuation of the trace.}
\end{example}

We will now study how to consistently approximate infinite computations ($\mu^\infty_S$ semantics) with finite ones ($\mu^t_S$ semantics). This will lead to the main result of this section (Theorem \ref{th:approx}). As a first  step, let us   introduce an appropriate notion of finite approximation for functions $f$ defined on the infinite product space $\Omega^\infty$.   Fix an arbitrary element $\star\in \Omega$. For each $f:\Omega^\infty \rightarrow \ereals^+$ and $t\geq 1$, let us define the function $f_t:\Omega^t \rightarrow \ereals^+$ by $f_t(\omega^t):=f(\omega^t,\star^\infty)$.
\orangehighlight{For a termination-determined function $f$, the relevant point is that $f_t$ is already exact on paths that have terminated by time $t$: if $\omega^t\in T^{\leq t}$, then $f_t(\omega^t)$ is independent of the arbitrarily chosen $\star$ and agrees with the value of $f$ on every infinite continuation of the terminating prefix. No analogous claim is needed for paths that are still nonterminated at time $t$.}

It is not difficult to check that, for any $t$,  $f_t$ is measurable over $\Omega^t$ \iffull (Lemma \ref{lemma:aux} in the Appendix).\else.\fi The next result shows how to approximate   $\sem S f$ with quantities defined \emph{only in terms of $f_t, \wt_t$ and $\mu^t_S$}, which is the basis for the   sampling-based inference algorithm  in the next section.
Formally, for $t\geq 1$ and a   measurable function $h:\Omega^t\rightarrow \ereals^+$, we let
$$
\begin{array}{ll}
[S]^t h&:=\expc_{\mu^t_S}[h] \,\,(=\int \mu^t_S(d\omega^t)h(\omega^t)\,)\,.
\end{array}
$$
%
%
%
%
%
%
The proof of Theorem \ref{th:approx} (Appendix \ref{app:proofs}) will follow  by applying the following lemma to the numerators and denominators of the expressions involved in \eqref{eq:approx}.
\orangehighlight{The intuition is that the contribution of traces that have terminated by time $t$ is represented exactly by $[S]^t f_t\cdot 1_{T^{\leq t}}\cdot\wt_t$. The only missing contribution comes from traces that have not terminated by time $t$; if $f\leq M$, this remainder can be bounded uniformly by $M$.}

\begin{lemma}\label{lemma:basic}
	\orangehighlight{Let $t\geq 1$ and let $f\leq M$ be a termination-determined function:}
	{ %
		\begin{equation}\label{eq:basic}
			\begin{array}{rcl}
				\greenhighlight{[S]^t f_t\cdot 1_{T^{\leq t}}\cdot \wt_t} & \leq  [S]f\cdot \wt  \leq &  \greenhighlight{[S]^t f_t\cdot 1_{T^{\leq t}}\cdot \wt_t} + M\cdot[S]^t(  1- 1_{ T^{\leq t}})\cdot \wt_t\,.
			\end{array}
		\end{equation}
	}\noindent
\end{lemma}

In the formulation of the theorem's upper bound, we find it convenient to introduce a  `correction factor' $\alpha_t\geq 1$,  the ratio of the weight of \emph{all} traces  to \emph{terminated} traces at time $t$. 

\begin{theorem}[finite approximation]\label{th:approx} \minorhighlight{Consider $S\in\St$ and assume that for some $t_0\geq 1$, $[S]^{t_0} 1_{T^{\leq t_0}}\cdot \wt_{t_0}>0$.} 
\orangehighlight{Then for any termination-determined function $f$ we have that $\sem{S}f$ is well defined.} 
Moreover, \revhighlight{if $f\leq M$ for some finite $M<+\infty$,} \minorhighlight{for every $t\geq t_0$, letting $\alpha_t:=\frac{[S]^t \wt_t}{[S]^t 1_{T^{\leq t}}\cdot \wt_t}$,} we have:
\begin{equation}\label{eq:approx}
\begin{array}{rcccl}
	\dfrac{\greenhighlight{[S]^t f_t\cdot 1_{T^{\leq t}}\cdot \wt_t}}{[S]^t \wt_t} &\leq & \sem{S}f&
	\leq &
	\dfrac{\greenhighlight{[S]^t f_t\cdot 1_{T^{\leq t}}\cdot \wt_t}}{[S]^t \wt_t}\alpha_t+M\cdot   \left(\alpha_t   -1\right)\,.
\end{array}
\end{equation}
\end{theorem}

\revhighlight{For arbitrary nonnegative termination-determined $f$, even if unbounded, the lower inequality in \eqref{eq:approx} remains valid. Without a finite bound $M$, the theorem does not provide a finite upper bound.}

When $f$ is an indicator function, $f=1_A$,   we can of course take $M=1$ in the theorem above. We first illustrate the above result with a simple example.

\begin{example}\label{ex:simple1}
Consider the PPG of Example \ref{ex:simple0} (Fig. \ref{fig:drunk}, left).
We ask what is the expected value of $c$ upon termination of this program. Formally, we consider the program checkpoint $S=0$, and the function $f$  on traces that returns the value of $c$ on the first terminated state, if any, and 0 elsewhere. \orangehighlight{This $f$ is clearly termination-determined.} We apply Theorem \ref{th:approx} to $\sem S f$. Fixing the time $t = 4$, we can calculate   easily the quantities involved in the approximation of $\sem S f$ in \eqref{eq:approx}. \greenhighlight{In doing so, we must consider the finitely many paths of length $t$ of nonzero probability and weight (there are only two of them), their weights and the value of $c$ on their final state when terminated.}
{\small
	$$
	\begin{array}{ll}
		[S]^t f_t\cdot 1_{T^{\leq t}}\cdot \wt_t =0\cdot \frac 1 2 + 1\cdot \frac 1 2\cdot \frac 1 2=  \frac{1}{4}
		\hspace*{.5cm}&
		[S]^t  \wt_t = \frac 1 2   +  \frac 1 2\cdot \frac 1 2 = \frac{3}{4} \\  %
		{[}S{]}^t  1_{T^{\leq t}}\cdot\wt_t =\frac 1 2   +  \frac 1 2\cdot \frac 1 2 = \frac{3}{4} & \alpha_t   = 1\,.
	\end{array}
	$$
}\noindent
Then, with $M=1$,  the lower and upper bounds in \eqref{eq:approx} coincide and yield $ \sem S f = \frac 1 3$.
If we remove conditioning on \minorhighlight{node 3}, then all the paths of length $t$ have weight 1, and a similar calculation yields   $ \sem S f= \frac 1 2$.
\end{example}

In more complicated cases, we may not be able to calculate exactly the quantities involved in \eqref{eq:approx}, but only to estimate them via sampling. To this purpose, we will introduce Feynman-Kac models and the Particle Filtering algorithm in the next section. For now, we content ourselves with the following example.

\begin{example}{\label{ex:RW2}
	Consider the PPG of Example \ref{ex:dm1} and   $f=f_2$ from Example \ref{ex:f}.  Take $S=0$. Then $\sem{S}f$ is the posterior expectation of the value of $d$, the drunk man's standard deviation. We can compute upper and lower bounds on  $\sem{S}f$ using \eqref{eq:approx}. Let us fix $t=60$. By sampling from $\mu^t_S$, we can compute separately the following estimates for each of the expected values  involved in  \eqref{eq:approx}:
	\greenhighlight{$[S]^t f_t\cdot 1_{T^{\leq t}}\cdot \wt_t=  0.396 $}, $[S]^t \wt_t=0.987$ and $\alpha_t=1.497$.
	Combining these estimates as in \eqref{eq:approx}, with $M=2$, \highlight{we obtain the estimated values $0.402$ and $1.596$ of the theoretical lower and upper bounds, respectively}. This   relatively large interval can be narrowed down by considering higher values of $t$, hence $\alpha_t$ closer to 1.
	A more efficient and accurate  method to compute the bounds in \eqref{eq:approx} will be introduced in Section \ref{sec:MC}, the Particle Filtering algorithm.
}\end{example}

The theorem below confirms that the bounds established above are asymptotically tight, at least under the assumption that the program $S\in \PP$ terminates with probability 1. In this case, in fact, the probability mass outside $\term^{\leq t}$ tends to 0, which leads the lower and the upper bound in \eqref{eq:approx} to coincide. Moreover, we get a simpler formula in the special case when termination is guaranteed to happen within a fixed time limit; for instance, in the case of    acyclic\footnote{Or, more accurately, PPGs where the only loop is the self-loop on the \textsf{nil} state.} PPGs.


\begin{theorem}[tightness]\label{th:tight} Assume the same hypotheses as in Theorem \ref{th:approx}. Further
assume that $\mu^\infty_S(T_\fini)=1$ \highlight{and that $f$ has a finite upper bound $M<+\infty$}. 
Then both the lower and the upper bounds in \eqref{eq:approx} tend to $\sem S f$ as $t\rightarrow +\infty$. In particular, if for some $t\geq 1$ we have    $[S]^t 1_{T^{\leq t}}=1$, 
then
\begin{equation}\label{eq:exact}
	\begin{array}{rcl}
		\sem{S} f  &= & \dfrac{[S]^t f_t\cdot \wt_t}{[S]^t    \wt_t}\,.
	\end{array}
\end{equation}

\end{theorem}

\begin{remark}\revhighlight{The finite upper bound in Theorem~\ref{th:tight} is essential for the stated convergence of the upper estimate, not for almost-sure termination itself. For example, an almost-surely terminating geometric loop can return an unbounded counter while remaining active with positive probability at every finite horizon. No finite $M$ is then available, so Theorem~\ref{th:approx} supplies only the lower bound. 
}
\end{remark}

\begin{example}\label{ex:simple2} For the PPG of Example \ref{ex:simple0} one has $\mu^\infty_S(T_\fini)=1$. As already seen in \minorhighlight{Example \ref{ex:simple1}}, lower and upper bounds coincide for $t\geq 4$.
\end{example}


\begin{example}[A nonterminating case]
\orangehighlight{
\label{ex:nonterm-observable}
Consider the probabilistic program $P$ in Fig.~\ref{fig:nonterm-observable}: with probability $1/2$ it terminates, while with probability $1/2$ it enters a nonterminating loop at checkpoint $S$. All score functions are identically $1$, hence $\wt=1$.
Define $f:\Omega^\infty\rightarrow\{0,1\}$ by
\[
f(\tilde\omega):=
1_{\{\tilde\omega\text{ never terminates and visits }S\text{ infinitely often}\}}.
\]
The function $f$ is termination-determined, since it is identically $0$ on every trace that has already terminated. For $P$, with probability $1/2$ the execution reaches $\nil$, where $f=0$, and with probability $1/2$ it reaches $S$ and remains there forever, where $f=1$. Therefore
\[
\sem P f=\frac12.
\]
After    $t\geq3$  the initial sampling and branching transitions have been performed:
\[
[P]^t1_{T^{\leq t}}\wt_t=\frac12,
\qquad
[P]^t\wt_t=1,
\qquad
[P]^t f_t1_{T^{\leq t}}\wt_t=0.
\]
Thus $\alpha_t=2$, and Theorem~\ref{th:approx}, with $M=1$, yields only
\[
0\leq \sem P f\leq 1.
\]
The bounds are sound but do not tighten with $t$. Indeed,
$\mu^\infty_P(T_{\fini})=1/2<1$, so the almost-sure termination hypothesis of Theorem~\ref{th:tight} is not satisfied. This shows that this hypothesis is essential: positive mass carried by nonterminating executions can contribute to a termination-determined observable while remaining invisible to the terminated part of finite truncations.}
\end{example}

\begin{figure}[t!]
\centering
\begin{minipage}{.40\linewidth}
\vspace{-0.15cm}
{\scriptsize
\begin{lstlisting}
b$\sim$B(1/2);
if (b==0) skip
else {
    while (true) { }
}
\end{lstlisting}
}
\vspace{-0.3cm}
\end{minipage}
\hspace*{.8cm}
\begin{minipage}{.42\linewidth}
{\small
\begin{center}
\tikzset{every picture/.style={line width=0.75pt}}
\begin{tikzpicture}[>=stealth, node distance=1cm,
  every node/.style={font=\scriptsize},
  checkpoint/.style={circle,draw,minimum size=4.5mm,inner sep=0pt},
  terminated/.style={circle,draw,double,minimum size=4.5mm,inner sep=0pt}]
  \node[checkpoint] (q0) {\revhighlight{$P$}};
  \node[checkpoint, below of=q0] (q1) {$1$};
  \node[checkpoint, below right of=q1, xshift=4mm] (S) {$S$};
  \node[terminated, below left of=q1, xshift=-4mm] (nil) {$\nil$};

  \draw[->] (q0) -- node[right] {$b\sim B(1/2)$} (q1);
  \draw[->] (q1) -- node[above left] {$[b=0]$} (nil);
  \draw[->] (q1) -- node[above right] {$[b\neq0]$} (S);
  \draw[->] (S) edge[loop right] (S);
  \draw[->] (nil) edge[loop left] (nil);
\end{tikzpicture}
\end{center}}
\end{minipage}
\caption{\scriptsize \textbf{Left}. Pseudo-code for the program of Example~\ref{ex:nonterm-observable}. \textbf{Right}. The corresponding PPG.
}
\label{fig:nonterm-observable}
\end{figure}


\orangehighlight{A practically relevant subclass of termination-determined functions consists of observables obtained by applying a function $h$, defined on $\Omega$, to the first terminated state of the trace, if any.}  
This way  $h$ is    {\emph{lifted}} to $\Omega^\infty$. \revhighlight{This construction includes terminal-state observables, but not the nontermination observable of Example~\ref{ex:nonterm-observable}.}  We formally introduce lifting below. Recall that for $t\geq 1$, $T_t=(\term^{\mathrm{c}})^{t-1}\cdot\term$.

\begin{definition}[lifting]\label{def:simple} Let $h:\Omega \rightarrow \ereals^+$ a nonnegative measurable function such that $\supp(h)\subseteq \term$. The \emph{lifting of}  $h$   is the measurable function $\lift h:\Omega^\infty\rightarrow \ereals^+$ defined as follows for each $\iomega=(\omega_1,\omega_2,...)$: $ \lift  h(\iomega):=\sum_{t\geq 1}1_{\cy{T_t}}(\iomega)\cdot h(\omega_t)$.
\end{definition}

\orangehighlight{Clearly, any $\lift h$ is termination-determined. Moreover, $\supp(\lift h)\subseteq T_{\fini}$.}


\begin{example}{\label{ex:lift}
	The indicator function for the set of paths that eventually terminate, $\lift h=1_{\term_\mathrm{f}}$, is clearly the lifting of $h=1_{\term}$. For more interesting examples, consider Example \ref{ex:f}.    The function $f_1$ is not a lifting of any function on $\Omega$, as it must count how many states are traversed up to  the first one in $\term$. However, one can modify the program of Example \ref{ex:dm1} to insert a new counter variable $j$  that provides this information, and return the value of $j$ at termination. The function $f_2$ is the lifting of the function $h(d,r,x,y,z)=d\cdot [0\leq d\leq 2]\cdot [z=\nil]$.

}\end{example}


\section{Feynman-Kac models}\label{sec:MC}
In the field of Sequential Monte Carlo methods, Feynman-Kac (FK) models \cite[Ch.9]{SMC}  are  characterized by the use of \emph{potential} functions.
A potential
in a Feynman-Kac model is a function that assigns a weight $G_t(x)$ to a \emph{particle} (instance of a random process) in state
$x$ at time
$t$. This weight represents how plausible or fit
$x$ is at time
$t$ based on some observable or conditioning. In other words,   $G_t$ modifies the \emph{importance} of particles as the system evolves. For instance, in a model for tracking an object, the potential function could depend on the distance between the predicted particle position and the actual observed position. Particles closer to the observed position get higher weights.

\paragraph{FK models and probabilistic program semantics}\label{sub:genrel}
As seen, our semantics incorporates conditioning via score functions applied at program checkpoints, and
aggregates their effect into a global weight $w$ over traces. This makes it possible to interpret program  semantics
  as a reweighted expectation (Definition \ref{def:obsPP}). Here we will show that this expectation can be approximated reliably using the Feynman-Kac framework and particle filtering.
We first introduce FK models in a general context.  Our formulation follows closely \cite[Ch.9]{SMC}. Throughout this and the next section, we let $t\geq 1$  be an arbitrary fixed integer.

\begin{definition}[Feynman-Kac models]\label{def:FK}  A \emph{Feynman-Kac (FK) model} is a tuple $\FC=(\X,t,\mu^1,\{K_i\}_{i=2}^t,\{G_i\}_{i=1}^t)$, where $\X=\ereals^\ell$ for some $\ell\geq 1$, $\mu^1$ is a probability measure on $\X$, \minorhighlight{for $i=2,...,t$, $K_i$ is a Markov kernel from $\X$ to $\X$, and for $i=1,...,t$, $G_i:\X\rightarrow \ereals^+$ is a measurable function}.

Let  $\mu^t$ denote the unique product measure on $\X^t$ induced by $\mu_1,K_2,...,K_t$ as per Theorem \ref{th:prod}. Let $G:=\Pi_{i=1}^t G_i$. 
Provided $0<\expc_{\mu^t}[G]<+\infty$, the \emph{Feynman-Kac measure} induced by $\FC$ is defined by the following, for every measurable $A\subseteq \X^t$:
	\vspace{-0.2cm}
\begin{align}\label{eq:FCM}
	\mufk_{\FC}(A)&:=\frac{\expc_{\mu^t}[1_A\cdot  G]}{\expc_{\mu^t}[G]}\,.
\end{align}
\end{definition}
	\vspace{-0.1cm}
We will refer to $G$ in the above definition as the \emph{global potential}.
Equality \eqref{eq:FCM} easily generalizes to expectations   taken according to $\mufk_{\FC}$. That is, for any measurable nonnegative function $g$   on $\X^t$, we can easily show that:
	\vspace{-0.3cm}
\begin{align}\label{eq:FCMe}
\expc_{\mufk_{\FC}}[g]&=\frac{\expc_{\mu^t}[g\cdot  G]}{\expc_{\mu^t}[G]}\,.
	\vspace{-0.2cm}
\end{align}
In what follows, we will suppress  the subscript ${}_{\FC}$ from ${\mufk}_{\FC}$ in the notation,  when no confusion   arises. Comparing  \eqref{eq:FCMe} against the definition \eqref{eq:semPP} suggests that the global potential $G$ should play in FK models a role analogous to the weight function $\wt$ in probabilistic programs. Note however that there is a major technical difference between the two, because FK models are only defined for a finite time horizon model given by $t$. A reconciliation between the two is possible thanks to the finite approximation theorem seen in the last section; this will be elaborated further below (see Theorem \ref{th:filtlift}).

We will be  particularly interested in the \emph{$t$-th marginal} of $\mufk$, that is the probability measure on $\X$ defined as ($A\subseteq \X$ measurable):
\begin{equation}\label{eq:phit}
\mufk_t(A):=\mufk(\X^{t-1}\times A)  =  \expc_\mufk[1_{\X^{t-1}\times A}]\,.
\end{equation}
The measure $\phi_t$ is called \emph{filtering} distribution (at time $t$), and can be effectively computed via the Particle Filtering algorithm   described in the next subsection.
Now  let $\ppg=(\St,E,\nil,\psco)$ be an arbitrary fixed PPG.  Comparing  \eqref{eq:FCMe} against e.g. the lower bound in  \eqref{eq:approx} suggests   considering the following FK model associated with $\ppg$ and a checkpoint $S$.

\begin{definition}[$\FC_S$  model]\label{def:PPFK} Let $t\geq 1$ be an integer and $S$ a program checkpoint of $\ppg$. \minorhighlight{Assume that $[S]^t\wt_t>0$.} We define $\FC_S$ as the FK model where: $\X=\Omega$, $\mu^1=\delta_{(0,S)}$, $K_i=\K$ ($i=2,...,t$) and $G_i=\scoz$ ($i=1,...,t$).
We let $\phi_S$  denote the measure on $\Omega^t$ induced by $\FC_S$.
\end{definition}

We now restrict our attention to functions $f$ that are the lifting of a nonnegative $h$ defined on $\Omega$.
Let $\phi_{S,t}$ denote the filtering distribution of $\phi_{S}$ at time $t$ obtained by \eqref{eq:phit}. In the following theorem we express the bounds in \eqref{eq:approx} in terms of the measure  $\phi_{S,t}$. The whole point and interest of this result is that the bounds are expressed  directly as expectations; these are moreover taken w.r.t. a \minorhighlight{single-time marginal} filtering distribution ($\phi_{S,t}$), rather than a distribution on the full $t$-step path space ($\mu^t_S$). Importantly, there are well-known algorithms to estimate expectations under a filtering distribution, as we will see in the next subsection.

\begin{theorem}[filtering distributions and lifted functions]\label{th:filtlift} Under the same assumptions of Theorem \ref{th:approx}, further assume that $f$ is the lifting of $h$. \revhighlight{Given a finite upper bound $h\leq M$,} \minorhighlight{for every $t\geq t_0$}, $\alpha_t= \expc_{\phi_{S,t}}[1_{\term}]^{-1}$ and
\begin{equation}\label{eq:filtering}
	\begin{array}{rcl}
		\beta_L:=\;\expc_{\phi_{S,t}}[h] &\leq  \,\, \sem S f \,\,\leq &  \expc_{\phi_{S,t}}[h]\cdot \alpha_t \;+\;M\cdot (\alpha_t-1)\;=:\beta_U\,.
	\end{array}
\end{equation}
\end{theorem}

\revhighlight{For an unbounded nonnegative $h$, the lower bound $\beta_L\leq\sem S f$ still holds, but the displayed expression gives no finite upper bound.}

\begin{example}\label{ex:simple3}
Consider again the PPG of Example \ref{ex:simple0}. We can re-compute $\sem S f$ relying on Theorem \ref{th:filtlift}. Fix $t=4$. We first compute the filtering distribution $\phi_t$ on $\X=\Omega=\ereals^3$ relying on its definition \eqref{eq:phit}. Similarly to what we did in Example \ref{ex:simple1}, we   consider the nonzero-weight, nonzero-probability traces of length four. Then we project onto   the final (fourth) state, and compute the weights of  the resulting triples $(c,d,S)$, then normalize. There are only two triples $(c,d,S)$ of nonzero probability: 
\vspace*{-0.2cm}
{\small
	\begin{align*}
		\phi_t(0,0,2)&=\frac 2 3 &
		\phi_t(1,1,2)&=\frac 1 3\,.
		\vspace*{-0.2cm}
	\end{align*}
}\noindent
The function $f$ considered in Example \ref{ex:simple1} is the lifting of the function   $h(c,d,S)=c\cdot  {[S=2]}$ defined on $\X=\ereals^3$. We apply Theorem \ref{th:filtlift} and get $\beta_L= \expc_{\phi_t}[h]=\frac 1 3\leq \sem S f$. Moreover $\expc_{\phi_t}[1_\term] =1$, hence $\alpha_t=1$ according to Theorem \ref{th:filtlift}. Hence $\beta_L=\beta_U=\sem S f = \frac 1 3$. This agrees with   examples   \ref{ex:simple1} and   \ref{ex:simple2}.
\end{example}

We can apply the above theorem to the functions described in 
Example \ref{ex:f} and to
 other computationally challenging cases: we will do so in Section \ref{sec:experiments}, after introducing in the next section the Particle Filtering algorithm.

\paragraph{The Particle Filtering algorithm}\label{sub:algo}
From a computational point of view, our interest in FK models lies in the fact that they allow for a simple, unified presentation of a class of efficient inference algorithms,  known as \emph{Particle Filtering (PF)} \cite{SMC,DelMoral04,wood}.  For the sake of presentation,  we only introduce here the basic  version, \emph{Bootstrap} PF,  following closely\footnote{\iffull Additional details in Appendix  \ref{app:PF}.\else Additional details in \cite{BC25}.\fi} \cite[Ch.11]{SMC}.
%
Fix  a generic FK model, $\FC=(\X,t,\mu^1,\{K_i\}_{i=2}^t,\{G_i\}_{i=1}^t)$. Fix  $N\geq 1$, the number of \emph{particles}, that is instances of the random process represented by the $K_i$'s, we want to simulate.  Let  $W=W^{1:N}=(W^{(1)},...,W^{(N)})$   be a tuple of $N$ real nonnegative random variables, the  \emph{weights}.
Denote by $\widehat W$ the normalized version of $W$, that is
$
\widehat W^{(i)}=W^{(i)}/(\sum_{j=1}^N W^{(j)})$.
A \emph{resampling scheme} for  $(N,W)$ is a $N$-tuple of random variables $R=(R_1,...,R_N)$  taking values on $1..N$ and depending on $W$, such that, for each $1\leq i\leq N$, 
one has:
$\expc[\sum_{j=1}^N  1_{R^{(j)}=i}\,|W]=N\cdot \widehat W^{(i)}
$.
In other words,      each index $i\in 1..N$   on average   is selected in $R$ a number of times  proportional to its  weight in $W$. We shall write $R(W)$ to indicate that $R$ depends on a given weight vector $W$.  Various   resampling schemes have been proposed in the literature, among which the simplest is perhaps \emph{multinomial resampling}; see e.g.  \cite[Ch.9]{SMC} and references therein.  Algorithm \ref{alg:PF} presents a generic  PF algorithm. Resampling here takes place at step 4: its purpose is to give more importance to particles with higher weight, when extracting the next generation of $N$ particles, while discarding particles with lower weight.

The justification and usefulness of this algorithm is that, under mild assumptions  \cite{SMC},   \highlight{for any measurable function $h$ defined on $\X$ such that $G_t\cdot h$ is bounded}, expectation under $\phi_t$, the filtering distribution on $\X$ at time $t$, in the limit can be expressed as weighted sum with weights $\widehat W^{(j)}_t$:
\vspace*{-0.3cm}
{
\begin{equation}\label{eq:convPF}
	\sum_{j=1}^N \widehat W^{(j)}_t\cdot h(X^{(j)}_t)\, \longrightarrow\,\expc_{\phi_t}[h] \text{\ \ \ \  a.s. as } N\longrightarrow +\infty\,.
	\vspace*{-0.3cm}
\end{equation}
}\noindent
The practical implication here is that we can estimate quite effectively    the expectations   involved in \eqref{eq:filtering}, for $\phi_t=\phi_{S,t}$, as weighted sums like in \eqref{eq:convPF}. 
Note that in the above consistency   statement  $t$ is held fixed --- it is one of the parameter of the FK model --- while the number of particles $N$ tends to $+\infty$.
%

\begin{algorithm}[t]{\small
	\caption{A generic PF algorithm}
	\label{alg:PF}
	{\small
		\begin{algorithmic}[1]
			\Statex \textbf{Input}:  {$\FC=(\X,t,\mu^1,\{K_k\}_{k=2}^t,\{G_k\}_{k=1}^t)$, a FK model; $N\geq 1$, no. of particles.}
			\Statex \textbf{Output}:$X^{1:N}_{t}\in \X^N$, $W^{1:N}_t\in  (\mathbb{R}^+)^N$. 
			\State  $X^{(j)}_1 \sim \mu^1 $\Comment{state initialisation}\tikzmark{top}
			\State $ W^{(j)}_{1} := G_1(X^{(j)}_{1})$\ \ \tikzmark{bottom}\tikzmark{right}\Comment{weight initialisation}
			\For{ $k=2,...,t$ }
			\State $r_{1:N}\sim R(W^{1:N}_{k-1})$\Comment{resampling}
			\State  $X^{(j)}_k \sim K_k(X^{(r_j)}_{k-1})$\ \ \tikzmark{topT}\tikzmark{rightT}\Comment{state update}
			\State $W^{(j)}_{k}:=G_{k}(X^{(j)}_{k})$ \tikzmark{bottomT}\Comment{weight update}
			\EndFor
			\State\Return $(X_t^{1:N}, W_t^{1:N})$
		\end{algorithmic}
		\label{key}		\AddNote{top}{bottom}{right}{\ \    ($j=1,...,N$)}
		\AddNote{topT}{bottomT}{rightT}{\ \  ($j=1,...,N$)}
}}
%
\end{algorithm}


\begin{algorithm}[t]{\small
{\small
	\begin{algorithmic}[1]
		\Statex \textbf{Input}:  {$\ppg=(\St,E,\nil,\psco)$, a PPG; $S\in \St$, initial pr. checkpoint; $t\geq 1$, time horizon; $N\geq 1$, no. of particles.}
		\Statex \textbf{Output}:{ $V \in \ereals^{m\times N}$, $Z,W\in \ereals^{1\times N}$.}
		\State  \highlight{$V:=0$}\ ;\  $Z:=S $\Comment{state initialisation}
		\State \highlight{$ W := \gamma_S(V)$}\Comment{weight  initialisation}
		\For{ $t-1\text{ \textbf{times}}$}
		\State $(V,Z):=\mathrm{Resampling}((V,Z),W)$\Comment{resampling}
		\For{ $(s,\varphi,\zeta,s')\in E$ }
		\State $M_{s,\varphi}:=\varphi(V)\ast (Z=s)$\Comment{mask computation}
		\EndFor
		\State  $V \sim \sum_{(s,\varphi,\zeta,s')\in E}\, \zeta(V)\ast M_{s,\varphi}$\ \ ;\ \  $Z := \sum_{(s,\varphi,\zeta,s')\in E} \, \,s'\cdot M_{s,\varphi}$\Comment{state update}
		\State  $W := \sum_{ s \in \St}\, \gamma_s(V)\ast (Z=s)$\Comment{weight update}
		\EndFor
		\State\Return $(V,Z,W)$
	\end{algorithmic}
}	
\caption{VPF, a Vectorized PF algorithm  for PPGs.}\label{alg:VPF}}
\vspace*{-0.1cm}
\end{algorithm}
%

\section{Implementation and experimental validation}\label{sec:experiments}
\subsection{Implementation}
The PPG model is naturally amenable to a vectorized implementation  of PF that leverages the fine-grained, SIMD parallelism existing at the level of particles. At every iteration,  the state of the $N$ particles, $\omega^N=(\omega_1,...,\omega_N)$  with $\omega_i=(v_i,z_i)\in \ereals^{m+1}$, will be stored using a pair of arrays $(V,Z)$  of   shape ${N\times m}$  and   ${N\times 1}$, respectively. The weight vector is stored using another array $W$ of shape ${N\times 1}$. We rely on  {vectorization} of operations: for a  function $f:\ereals^k\rightarrow \ereals$ and a $N\times k$ array $U$,  $f(U)$ will denote the $N\times 1$  array obtained by applying $f$   to each row of $U$. In particular, we denote by $(Z=s)$ (for any $s\in \mathbb{N}$) the $N\times 1$ array obtained applying   element-wise the indicator function $1_{\{s\}}$  to $Z$ element-wise, and by $\varphi(V)$ the $N\times 1$ array obtained by applying the predicate $\varphi$ to $V$  to the row-wise. For $U$   a $N\times k$   array and $W$ a $N\times 1$ array, $U\ast W$ denotes the $N\times k$   array obtained by multiplying the the $j$th row of $U$ by the $j$th element of $W$, for $j=1,...,N$: when $W$ is a $0/1$ vector, this is an instance of \emph{boolean masking}. Abstracting the vectorization primitives of modern CPUs and programming languages, we model the assignments  of  a vector to an array variable    as a single instruction,
written $U:=Z$. The usual rules for broadcasting scalars to vectors apply, so e.g. $V:=S$ for $S\in \ereals$ means filling $V$ with $S$.  Likewise, for $\zeta$ a parametric distribution, $U\sim \zeta(V)$ means sampling $N$ times independently from   $\zeta(v_1),...,\zeta(v_N)$, and assigning the resulting matrix to $U$: this too counts as a single instruction.

Based on the above idealized model of vectorized computation,  we present VPF, a vectorized  version of the  PF algorithm for PPGs, as Algorithm \ref{alg:VPF}.   Here it is assumed that $\St\subseteq \mathbb{N}$, while $\psco(s)=\gamma_s$. On line 4, $\text{Resampling}(\cdot)$ denotes the result of applying a generic resampling algorithm based on weights $W$ to the current particles' state, represented by the   pair of vectors $(V,Z)$. With respect to the generic PF Algorithm \ref{alg:PF}, here in  the returned output, $(V,Z)$ corresponds to $X_t$  and $W$ to $W_t$. Note that there are no loops where the number of iterations depends on $N$;  the \textbf{for} loop in lines 5--7   only scans the transitions set $E$, whose size is independent of $N$. Line 8 is just a vectorized implementation of sampling from the Markov kernel function in \eqref{eq:PPGMK}.  Line 9 is a vectorized implementation of the combined score function \eqref{eq:scoz}.
In the actual TensorFlow implementation, the sums in lines 8 and 9  are encoded via boolean masking and vectorized operations.

\begin{example}
	Consider the PPG of Example \ref{ex:dm1} and   $f=f_2$ from Example \ref{ex:f}. Take $S = 0$.
	Then $[\![S]\!]_f$ is the posterior expectation of the value of $d$, the drunk
	man's standard deviation. \highlight{We estimate the theoretical lower and upper bounds on $[\![S]\!]_f$ defined in Theorem \ref{th:filtlift}, applying  Algorithm \ref{alg:VPF}.} Let us fix $t = 60$ and $N = 10^4$ particles. Under the filtering distribution $\phi_{S,t}$, we get $E_{\phi_{S,t}}[h] = 0.768, \alpha_t = 1.029$.
	Combining these estimates as in \eqref{eq:filtering}, with $M = 2$, \highlight{we obtain the estimated values $0.768$ and $0.847$ of the theoretical lower and upper bounds, respectively}.
	The
	value $\alpha_t = 1.029 \approx 1$ indicates that nearly all particles have 	terminated by time $t = 60$ under $\phi_{S,t}$.
\end{example}


\subsection{Experimental validation}
We   illustrate some   experimental results obtained with a proof-of-concept  TensorFlow-based \cite{TF} implementation of Algorithm \ref{alg:VPF}. We still refer to this implementation as \TSIpf.  We have considered a number of challenging probabilistic programs that feature conditioning inside loops.   For all these programs, we will estimate $\sem S f$, for given functions $f$,  relying on the bounds provided by Theorem \ref{th:filtlift}  in terms of expectations w.r.t. filtering distributions. Such expectations will be estimated via \TSIpf.
We also compare \TSIpf\ with two state-of-the-art PPLs, webPPL \cite{webppl} and CorePPL \cite{CorePPL}. webPPL is  a popular PPL supporting several inference algorithms, including SMC, where resampling is handled via continuation passing.   We have chosen to   consider CorePPL by Lunden et al. as it supports a very efficient implementation of PF: in \cite{CorePPL},  a comparison of CorePPL with webPPL, Pyro \cite{Pyro} and other PPLs in terms of performance shows the superiority of CorePPL SMC-based inference across a number of benchmarks.
CorePPL's implementation  is based on a compilation into an intermediate format, conceptually similar to our PPGs.

\paragraph{Models} For our experiments we have considered the following  programs:
 \emph{Aircraft tracking} (AT, \cite{WuEtAl}), \emph{Drunk man and mouse} (DMM, Example \ref{ex:dm1}), \emph{Hare and tortoise} (HT, e.g. \cite{Bagnall}), \emph{Bounded retransmission protocol} (BRP, \cite{5}), \emph{Non-i.i.d. loops} (NIID, e.g. \cite{5}), the \emph{ZeroConf} protocol (ZC, \cite{2}), and two variations of \textit{Random Walks},  RW1 (\cite{VMCAI24}, Example 2) and RW2 in the following. In particular, AT is a model where a single aircraft is tracked in a 2D space using noisy measurements from six radars.
  HT simulates a race between a hare and a tortoise on a discrete line. BRP models a scenario where multiple packets are transmitted over a lossy channel.    NIID   describes a process that keeps tossing two fair coins until both show tails. ZC is an idealized version of the network connection protocol by the same name.  RW1, RW2 are random walks with Gaussian steps. The pseudo-code of these models  is reported in  Appendix \ref{app:models}.
These programs feature conditioning/scoring inside loops.
In particular, DMM, HT and NIID feature unbounded loops: for these three programs, in the case of \TSIpf\ we have truncated the execution after   $k=1000,100,100$ iterations, respectively, and set the time parameter $t$ of Theorem \ref{th:filtlift} accordingly, \highlight{which allows us to estimate the theoretical bounds on the value of $\sem S f$}.  For the other tools, we just consider the   truncated estimate returned  at the end of   $k$ iterations. 
 \revhighlight{AT, BRP, RW1 and RW2 have bounded loops. In contrast, the probe counter in ZC is reset on a restart, so ZC may iterate indefinitely. The ZC experiments use a $100$-iteration simulation cutoff, not an intrinsic bound on the loop.}
In particular,  AT   features multiple conditioning inside a for-loop,  sampling from a mix of continuous and discrete distributions, and noisy observations.
Below, we discuss the obtained experimental results in terms of accuracy and performance.
A   description  of these programs, together with  further details on the experimental set up, can be found in \iffull Appendix \ref{app:air}; \else \cite{BC25,github}. \fi
 code available from \cite{github}. 
Table \ref{tab:table1}  summarizes the obtained experimental results, which we comment in the following.

{\small
\begin{sidewaystable}
	\captionsetup{width=\textwidth}
	\vspace{15cm}
	{ \renewcommand{\arraystretch}{0.9}		
		\centering
		\resizebox{\textwidth} {!}{
			\begin{tabular}{|c|c||c|c|c||c|c|c||c|c|c||c|c|c||c|c|c|}
				\hline
				\multicolumn{2}{|c||}{\multirow{2}{*}{}} & \multicolumn{3}{c||}{\small{\textbf{AT}}} & \multicolumn{3}{c||}{\small{\textbf{DMM}}} & \multicolumn{3}{c||}{\small{\textbf{HT}}} & \multicolumn{3}{c||}{\small{\textbf{BRP}}} & \multicolumn{3}{c|}{\small{\textbf{NIID}}} \\
				\cline{3-17}
				\multicolumn{2}{|c||}{}& \small{\textbf{VPF}} & \small{\textbf{CorePPL}} &  \small{\textbf{webPPL-smc}} &\small{\textbf{VPF}}& \small{\textbf{CorePPL}} &  \small{\textbf{webPPL-smc}} & \small{\textbf{VPF}} & \small{\textbf{CorePPL}} &  \small{\textbf{webPPL-smc}} &\small{\textbf{VPF}} & \small{\textbf{CorePPL}} & \small{\textbf{webPPL-smc}} & \small{\textbf{VPF}} & \small{\textbf{CorePPL}} & \small{\textbf{webPPL-smc}} \\
				\hline
				
				\multirow{3}{*}{$N=10^3$}
				&\textit{\scriptsize time}
				&\textbf{0.009}&0.014&0.190
				& 2.348  &\textbf{0.050}&0.474
				&0.274&\textbf{0.015}&0.152
				&0.676&\textbf{0.021}&0.155
				&0.240&\textbf{0.010}&0.061 \\
				\cline{2-17}
				&\textit{\scriptsize EV}
				&6.805&6.955&6.696
				&0.478$\pm 0.110$&0.501&0.427
				&32.834&33.683&32.368
				&0.018&0.016&0.023
				&3.594&2.694&3.473 \\
				\cline{2-17}
				&\textit{\scriptsize ESS}
				&\textbf{1000}&\textbf{1000}&{999}
				&\textbf{994.6}&726.9&9.73
				&\textbf{955.0}&758.9&951.1
				&\textbf{1000}&\textbf{1000}&974.5
				&\textbf{1000}&846.6&726.9 \\
				\hline
				\multirow{3}{*}{$N=10^4$}
				&\textit{\scriptsize time}
				&\textbf{0.131}&0.194&3.842
				&{1.947}&\textbf{0.988}&7.190
				&0.290&\textbf{0.180}&3.839
				&0.786&\textbf{0.309}&1.328
				&0.323&\textbf{0.058}&0.490 \\
				\cline{2-17}
				&\textit{\scriptsize EV}
				&6.817&6.967&6.760
				&0.539$\pm 0.115$&0.498&0.481
				&32.725&33.474&32.702
				&0.029&0.025&0.024
				&3.364&2.766&3.417 \\
				\cline{2-17}
				&\textit{\scriptsize ESS}
				&$\mathbf{10^4}$&$\mathbf{10^4}$&\textit{9975}
				&\textbf{9984.5}&7797.4&69.417
				&\textbf{9445.0}&7692.9&9476.2
				&$\mathbf{10^4}$&$\mathbf{10^4}$&9745.8
				&$\mathbf{10^4}$&8555.6&7560.5 \\
				\hline

				\multirow{3}{*}{$N=10^5$}
				&\textit{\scriptsize time}
				&\textbf{0.354}&{2.252}&-
				&\textbf{2.268}&29.936&-
				&\textbf{0.379}&4.225&361.792
				&\textbf{0.797}&5.010&15.038
				&\textbf{0.445}&1.083&92.419 \\
				\cline{2-17}
				&\textit{\scriptsize EV}
				&6.818&6.970&-
				&0.537$\pm 0.120$&0.507&-
				&33.128&33.545&32.560
				&0.024&0.025&0.026
				&3.467&2.772&3.430 \\
				\cline{2-17}
				&\textit{\scriptsize ESS}
				&$\mathbf{10^5}$&9.9\text{e}$\mathbf{10^5}$&-
				&$\approx\mathbf{10^5}$&$7.6\text{e}10^4$&-
				&$\approx\mathbf{10^5}$&$7.7\text{e}10^4$&$\approx\mathbf{10^5}$
				&$\mathbf{10^5}$&$\mathbf{10^5}$&$9.7\text{e}10^4$
				&$\mathbf{10^5}$&$8.5\text{e}10^4$&$7.6\text{e}10^4$\\
				\hline

				\multirow{3}{*}{$N=10^6$}
				&\textit{\scriptsize time}
				&\textbf{2.286}&26.481&-
				&\textbf{38.977}&-&-
				&\textbf{3.749}&49.493&-
				&\textbf{10.155}&58.448&-
				&\textbf{2.916}&14.323&- \\
				\cline{2-17}
				&\textit{\scriptsize EV}
				&6.834&6.980&-
				&0.541$\pm 0.111$&-&-
				&33.432&33.606&-
				&0.024&0.025&-
				&3.413&2.774&- \\
				\cline{2-17}
				&\textit{\scriptsize ESS}
				&$\mathbf{10^6}$&9.9\text{e}$10^5$&-
				&$\approx\mathbf{10^6}$&-&-
				&$\approx\mathbf{10^6}$&$7.7\text{e}{10^5}$&-
				&$\mathbf{10^6}$&$\mathbf{10^6}$&-
				&$\mathbf{10^6}$&$8.5\text{e}{10^5}$&- \\
				\hline
				\small{\textbf{webPPL-rej}}&\textit{\scriptsize EV}&\multicolumn{3}{c||}{-}&\multicolumn{3}{c||}{0.494}&\multicolumn{3}{c||}{32.683}&\multicolumn{3}{c||}{0.023}&\multicolumn{3}{c|}{3.414}\\
				\hline
			\end{tabular}
		}
	}
	$ $\\
	$ $\\
	{ \renewcommand{\arraystretch}{0.6}		
		\centering
		\resizebox{\textwidth} {!}{
			\begin{tabular}{|c|c||c|c|c||c|c|c||c|c|c||c|c|c||c|c|c|}
				\hline
				\multicolumn{2}{|c||}{\multirow{2}{*}{}} & \multicolumn{3}{c||}{\textbf{RW1}} &  \multicolumn{3}{c||}{\textbf{ZC.1}}  &  \multicolumn{3}{c||}{\textbf{ZC.2}} &  \multicolumn{3}{c||}{\textbf{RW2.1}}&  \multicolumn{3}{c|}{\textbf{RW2.2}} \\
				\cline{3-17}
				\multicolumn{2}{|c||}{}& \small{\textbf{VPF}} & \small{\textbf{CorePPL}} & \small{\textbf{webPPL}} &\small{\textbf{VPF}} & \small{\textbf{CorePPL}} & \small{\textbf{webPPL}}&\small{\textbf{VPF}} & \small{\textbf{CorePPL}} & \small{\textbf{webPPL}}&\small{\textbf{VPF}} & \small{\textbf{CorePPL}} & \small{\textbf{webPPL}}&\small{\textbf{VPF}} & \small{\textbf{CorePPL}} & \small{\textbf{webPPL}}\\
				\hline
				
				\multirow{3}{*}{$N=10^3$}
				&\textit{\scriptsize time}
				&0.192&\textbf{0.009}&0.045
				&0.206&\textbf{0.018}&0.083
				&0.262&\textbf{0.016}&0.049
				&0.231&\textbf{0.024}&0.232
				&0.232&\textbf{0.021}&0.187\\
				
				\cline{2-17}
				&\textit{\scriptsize EV}
				&$0.323$&0.324&0.343
				&$0.212$&0.142&0.250
				&$0.514$&0.477&0.483
				&$1.046$&0.642&0.912
				&$0.677 $&0.750&1.092\\
				\cline{2-17}
				&\textit{\scriptsize ESS}
				&537.0&\textbf{1000}&46.739
				&\textbf{1000}&\textbf{1000}&392.2
				&\textbf{1000}&\textbf{1000}&245.2
				&\textbf{992.0}& 780.0&479.9
				&\textbf{999.0}&997.0&133.9\\
				\hline
				\multirow{3}{*}{$N=10^4$}
				&\textit{\scriptsize time}
				&0.238&\textbf{0.031}&0.269
				&0.349&\textbf{0.068}&0.610
				&0.325&\textbf{0.029}&0.207
				&0.285&\textbf{0.186}&3.043
				&\textbf{0.271}&0.275&2.081\\
				\cline{2-17}
				&\textit{\scriptsize EV}
				&$0.334$&0.328&0.336
				&$0.242$&0.129&0.232
				&$0.483$&0.474&0.478
				&$1.367$&0.856&1.066
				&$0.982$&0.929&1.083\\
				\cline{2-17}
				&\textit{\scriptsize ESS}
				&5163.0&\textbf{10000}&562.795
				& \textbf{10000}&\textbf{10000}&4263.0
				&\textbf{10000}&\textbf{10000}&2446.8
				&\textbf{9967.0}&7529.0&4026.6
				&\textbf{9701.0}&9350.9&582.5\\
				\hline
				\multirow{3}{*}{$N=10^5$}
				&\textit{\scriptsize time}
				&0.436&\textbf{0.260}&7.956
				& \textbf{0.479}&0.558&5.778
				&0.378&\textbf{0.243}&1.673
				&\textbf{0.405}&3.003&181.802
				&\textbf{0.290}&3.887&149.831\\
				\cline{2-17}
				&\textit{\scriptsize EV}
				&$0.337$&0.328&0.332
				&$0.174$&0.131&0.233
				&$0.493$&0.479&0.479
				&$1.009$&0.998&0.982
				&$1.040$&0.982&1.037\\
				\cline{2-17}
				&\textit{\scriptsize ESS}
				&$5.1\text{e}{10^4}$&$\mathbf{10^5}$&$5.7\text{e}{10^3}$
				&$\mathbf{10^5}$&$\mathbf{10^5}$&$4.2\text{e}{10^4}$
				&$\mathbf{10^5}$&$\mathbf{10^5}$&$2.4\text{e}{10^4}$
				&$\approx\mathbf{10^5}$&$7.3\text{e}{10^4}$&$4.7\text{e}{10^4}$
				 &$\approx\mathbf{10^5}$&$9.5\text{e}{10^4}$&$2.0\text{e}{10^3}$\\
				\hline

				\multirow{3}{*}{$N=10^6$}
				&\textit{\scriptsize time}
				&3.422&\textbf{2.569}&-
				&\textbf{3.829}&5.790&-
				&3.928&\textbf{2.485}&-
				&\textbf{3.742}&35.271&-
				&\textbf{3.595}&42.134&-\\
				\cline{2-17}
				&\textit{\scriptsize EV}
				&$0.329$&0.329&-
				&$0.245$&0.130&-
				&$0.479$&0.480&-
				&$1.011$&1.001&-
				&$1.023$&1.002&-\\
				\cline{2-17}
				&\textit{\scriptsize ESS}
				&$5.2\text{e}{10^5}$&$\mathbf{10^6}$&-
				&$\mathbf{10^6}$&$\mathbf{10^6}$&-
				&$\mathbf{10^6}$&$\mathbf{10^6}$&-
				&$\approx\mathbf{10^6}$&$7.3\text{e}{10^5}$&-
				&$\approx\mathbf{10^6}$&$9.4\text{e}10^5$&-\\
				\hline
				\small{\textbf{webPPL-rej}}&\textit{\scriptsize EV}&\multicolumn{3}{c||}{0.332}&\multicolumn{3}{c||}{0.235}&\multicolumn{3}{c||}{0.479}&\multicolumn{3}{c||}{1.022}&\multicolumn{3}{c|}{1.061}\\
				\hline
			\end{tabular}
		}
	}
	\caption{Execution time ($time$) in seconds, estimated expected value ($EV$) and effective sample size ($ESS:=(\sum_{i=1}^N W_i)^2/(\sum_{i=1}^N W^2_i)$; the higher the better, see e.g. \cite{RobertESS}) as the number of particles ($N$) increases, for \TSIpf, CorePPL and webPPL, when applied on Aircraft tracking (AT), Drunk man and mouse (DMM), Hare and tortoise (HT), Bounded retransmission protocol (BRP), Non-i.i.d. loops (NIID), ZeroConf (ZC.1, ZC.2) and  Random Walks (RW1  and RW2.1, RW2.2). For \TSIpf,  with reference to Theorem \ref{th:filtlift}:  \revhighlight{for the bounded loops AT, BRP, RW1, RW2.1 and RW2.2, $\beta_L=\beta_U=\sem S f$ (since $\alpha_t=1$) and $EV$ estimates that common value; for ZC.1 and ZC.2, $EV$ is a finite-horizon estimate, not asserted to equal the unrestricted infinite-trace expectation;} for HT and NIID, we only provide $\beta_L$, as $\beta_U$ is vacuous $(M=+\infty)$. For DMM we give the midpoint of the interval $[\beta_L,\beta_U]$ $\pm$ its half-width. Best results for $time$ and $ESS$ for each example and value of $N$ are marked in \textbf{boldface}. Everywhere,  '$-$' means  no result due to out-of-memory   or  timeout ($500$s). The results for DMM, especially for smaller values of $N$, exhibit a significant empirical variance:  those reported in the table are obtained by averaging over 10 runs of each algorithm. Generally, there is an  agreement  across the tools about the estimates  $EV$: an exception is  NIID, where CorePPL returns values significantly different from the other tools' and from the exact value $\frac{24}7=3.428\cdots$, cf. \cite{5}. Also, for DMM the  EV estimates returned by CorePPL and webPPL  appear  to be consistently lower than the midpoint of the interval returned by \TSIpf.}
	\label{tab:table1}
\end{sidewaystable}	
}

\paragraph{Accuracy}
We have compared \TSIpf, CorePPL and webPPL across the above mentioned examples for different values of $N$, the number of particles
(details  in Table \ref{tab:table1}).
At least for $N\geq 10^5$,   the tools tend generally to  return very  similar estimates of the expected value, which we take    as an empirical evidence of accuracy. Additional insight into accuracy is obtained by directly comparing the results of VPF with those of webPPL-rejection (when available), which is an exact inference algorithm: the expected values estimated by webPPL-rejection  are consistently in line to those of VPF.  We have also considered Effective Sample Size (ESS), a measure of   diversity of the sample,  the higher the better \cite{RobertESS}. In terms of ESS, the difference across the tools is significant: with one exception (program RW1),  VPF  yields ESS that are higher  or comparable to those of the other tools. We refer the reader to Appendix \ref{app:air} for additional explanation, in particular as to the significance of the mentioned exception.
%
%
\paragraph{Execution time} For larger values of $N$   \TSIpf\  generally outperforms   the other   considered tools   in terms of execution time. The difference is especially noticeable  for $N=10^6$.
Figure \ref{fig:scatterplot} provides a graphical comparison,  with scatterplots showing the ratio of execution times $(time_{\mathrm{other-tool}} / time_{\mathrm{VPF}})$ on a log scale, across the different examples (actual data points in   Table \ref{tab:table1}).
In the case of WebPPL, nearly all data points lie above the x-axis, indicating superiority of VPF. In the case of CorePPL, for $N=10^5$ the data points are  quite uniformly distributed  across the x-axis, indicating basically a tie. For $N=10^6$, we have a majority of points above  the x-axis, indicating again superiority of VPF, sometimes by orders of magnitude.
\newcommand{\scalefactor}{8.0}
\begin{figure}[t]
\includegraphics[width=4.2cm,height=3.35cm]{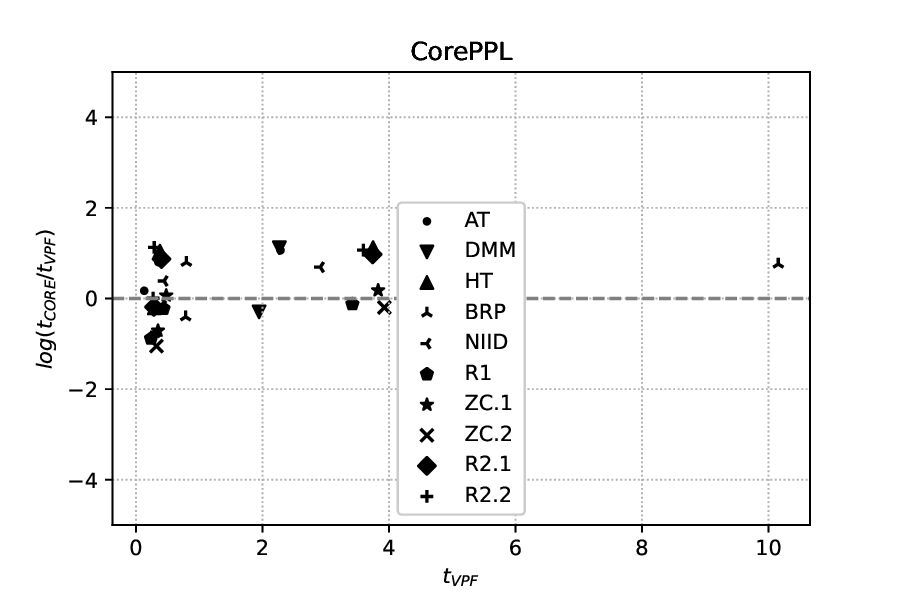}\hspace*{-0.5cm}
\includegraphics[width=4.2cm,height=3.35cm]{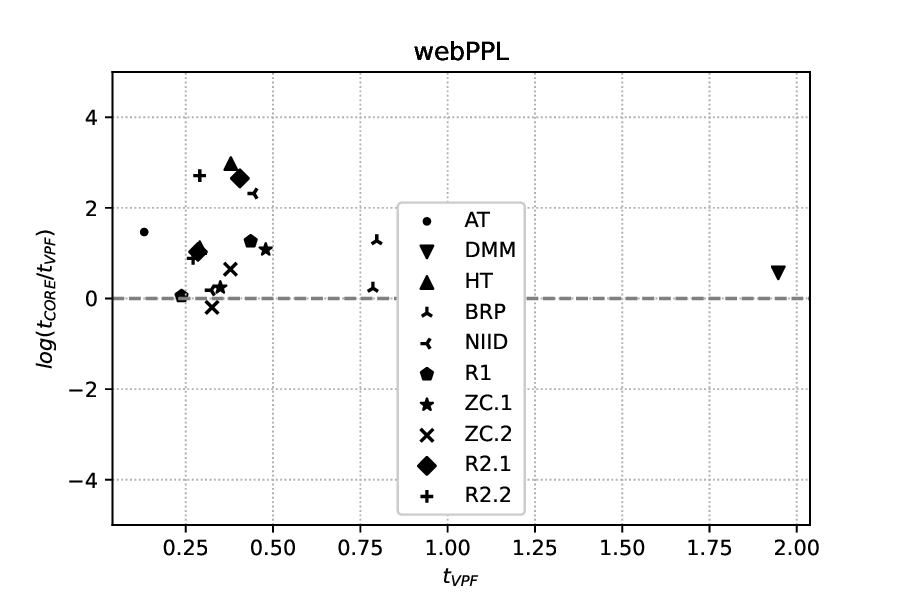}\hspace*{-0.4cm}
\includegraphics[width=4.2cm,height=3.35cm]{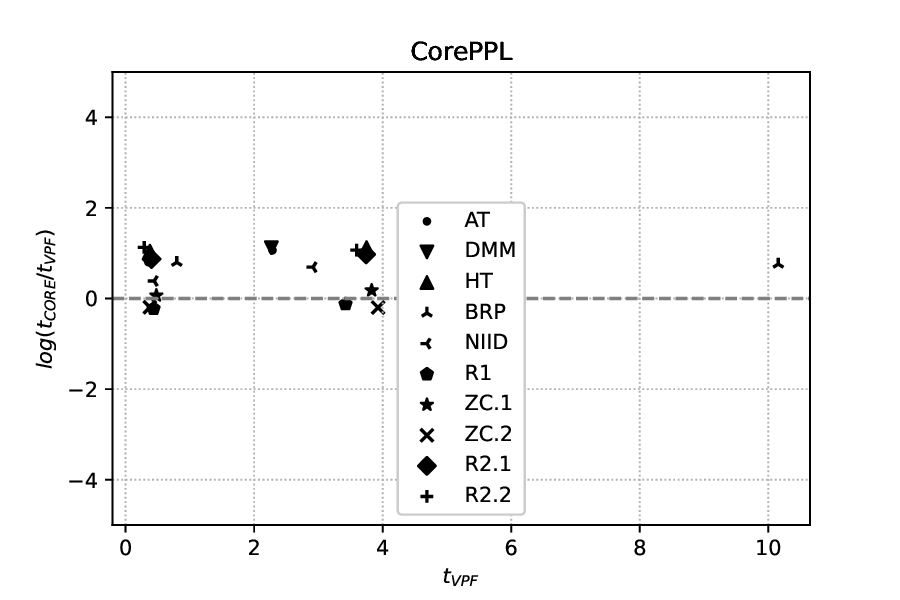}\hspace*{-0.5cm}
\includegraphics[width=4.2cm,height=3.35cm]{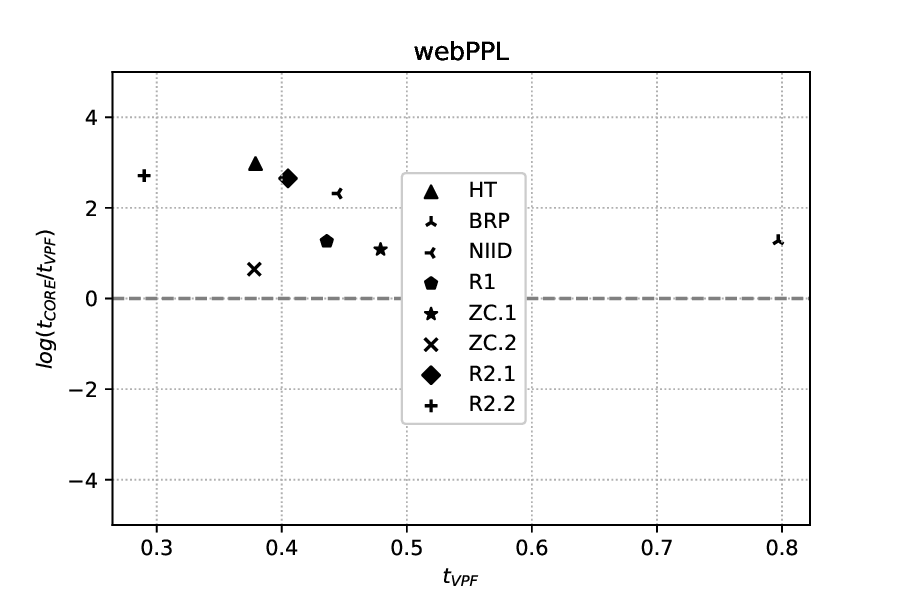}\hspace{+0.3cm}
	\vspace{-0.7cm}
\caption{\footnotesize
	For $N=10^5,10^6$, scatterplots of the log-ratios of execution times, $\log_{10}(\mathrm{time}_{{tool}}/\mathrm{time}_{\mathrm{VPF}})$, based on   data points in  Table \ref{tab:table1}. From left to right: $N=10^5$, $tool=$ CorePPL;  $N=10^5$, $tool=$ WebPPL-smc;  $N=10^6$, $tool=$ CorePPL;  $N=10^6$, $tool=$ WebPPL-smc.
	For $N=10^6$,  the vast majority of the data points lie above the x-axis, indicating superior performance of VPF across different examples.	 An enlarged version of these plots is reported in 	 Appendix \ref{app:air}.
}
\label{fig:scatterplot}
\vspace{-0.5cm}
\end{figure}

A closer look in the $N=10^6$  case reveals that the only programs where CorePPL beats VPF
are RW1 and ZC. This is most likely due to the low probability of conditioning  in these programs; for instance in  RW1  just a single final conditioning is performed. As in CorePPL   resampling  is only performed following a conditioning, this may   explain its lower execution times in these cases.  To further investigate this issue, we consider   a version of RW2  where the probability of conditioning is governed by a parameter $\lambda\in [0,1]$, and run it for different values of $\lambda$.
The obtained results are showed in Figure \ref{fig:timeev}. We   observe that for both CorePPL and WebPPL execution time tends to increase  as the probability  $\lambda$ of conditioning  increases; on the contrary,  the execution time of VPF appears to be insensitive to  $\lambda$. This suggests that VPF has a definite advantage over tools with explicit resample,  on models with heavy conditioning.
%
The difference between CorePPL and VPF is  further analyzed  in the next paragraph   in terms of scalability.
\begin{figure}[t]
	\centering
	\includegraphics[width=4.8cm,height=3.7cm]{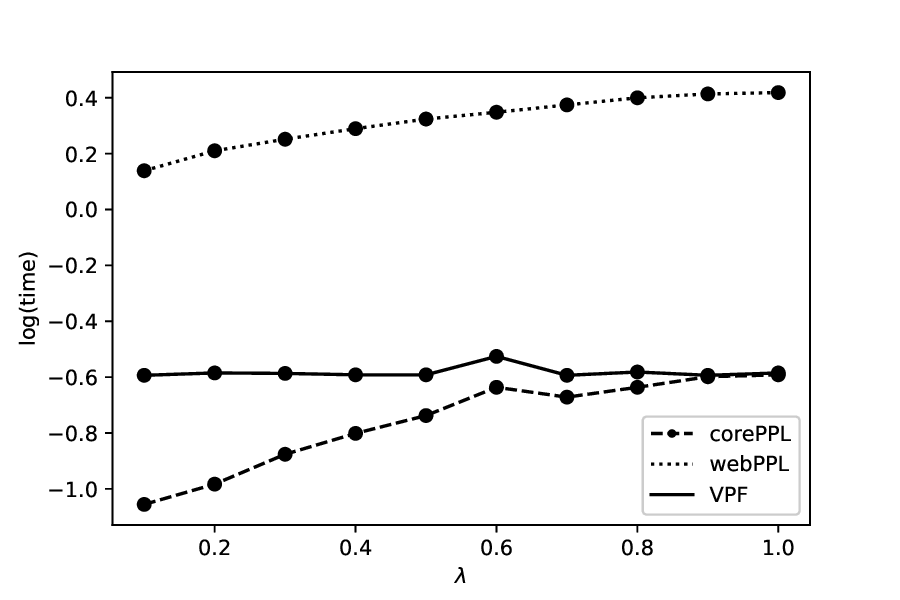}
	\includegraphics[width=4.8cm,height=3.7cm]{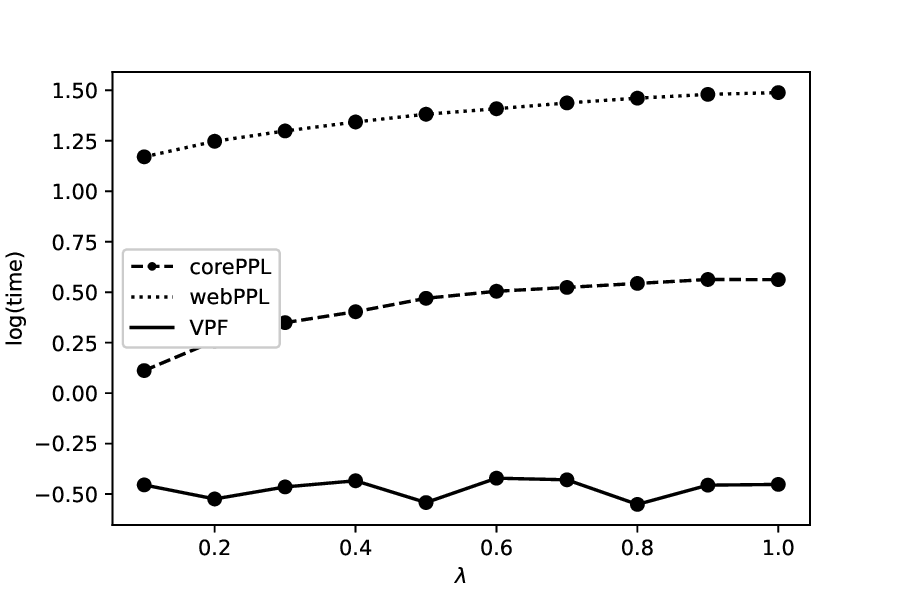}
	\includegraphics[width=4.8cm,height=3.7cm]{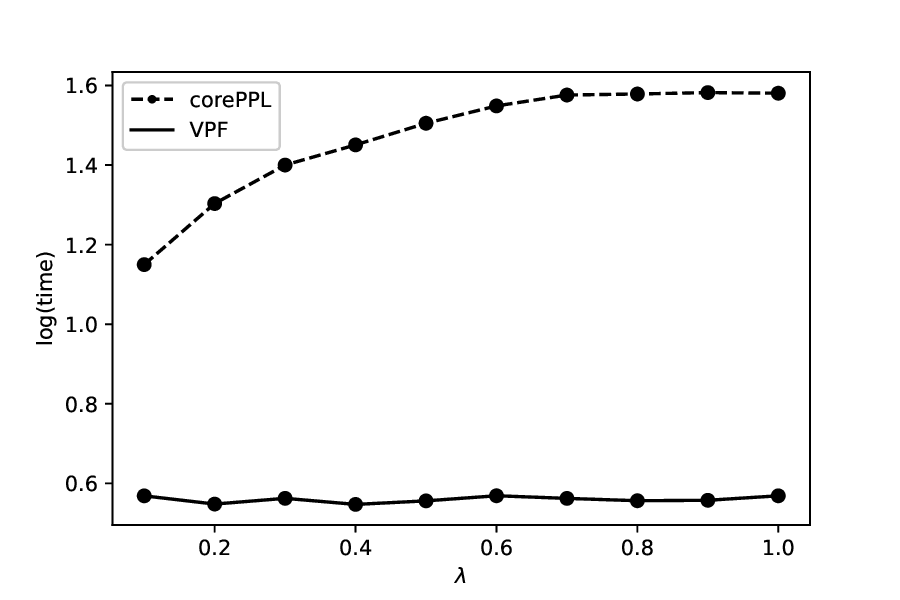}
	\caption{\scriptsize Execution times (in seconds) for the RW2 program, as a function of the 	 probability $\lambda$  of conditioning on external data for $N=10^4$ (left), $N=10^5$ (center) and $N=10^6$ (right). webPPL   missing from the right-most plot due to time-out. Execution times of VPF are basically insensitive to  $\lambda$.}
	\label{fig:timeev}
\vspace*{-.5cm}
\end{figure}

%

\paragraph{Scalability}
\begin{wrapfigure}{r}{4cm}
	\vspace*{-.81cm}
	\includegraphics[width=4.0cm,height=3.6cm]{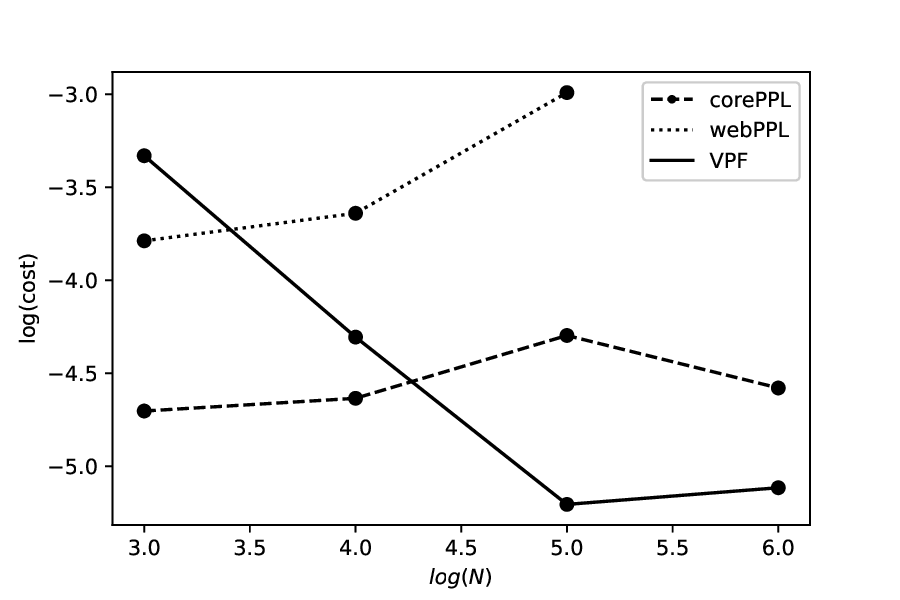}
\end{wrapfigure}
The plot on the right shows the behaviour the \emph{average unit cost (per particle)} of VPF, CorePPL and WebPPL across all the models we analyzed for $N=10^3,...,10^6$ on a log-scale. Here, for each $N$ the average unit cost (in seconds) across $k$ models is $(t_1+t_2+..+t_k)/(N\cdot k)$, with $t_i$ the execution time of the $i$-th model. Consistently with Figure \ref{fig:timeev}, we   observe that the cost of VPF decreases as the number of samples increases, whereas the cost of the other tools remains constant or increases (webPPL).


\section{Conclusion}\label{sec:concl}
We have proposed a framework for inference in universal probabilistic programs that combines a rigorous  trace-based operational semantics with practical algorithmic realizability.
The approach builds on a structured operational model and connects to classical tools in stochastic process theory (FK models), leading to an inference method that is both theoretically sound and efficiently implementable.
We summarize the main insights of our approach and possible future directions below.

\paragraph{Structured semantics via PPGs (Sections \ref{sec:PP} and \ref{sec:obs})}
Our framework is centered around   an automata-theoretic description format for programs, {Probabilistic Program Graphs} (PPGs). 
In PPGs,   transitions encode  sampling behavior, while nodes represent conditioning checkpoints via score functions. This structure supports a rigorous infinite-trace semantics  and facilitates the alignment of computations, a feature that becomes crucial for vectorized implementations.

\paragraph{Approximation via Feynman-Kac models (Sections \ref{sec:FA} and \ref{sec:MC})}
\orangehighlight{The main theoretical contribution is a connection between the infinite-trace semantics of PPGs and {Feynman-Kac} (FK) models, a standard tool in the analysis of state-based stochastic processes. The expected values of termination-determined observables over infinite traces can be bounded using finite, truncated computations. In particular, Theorem~\ref{th:approx} provides such finite-time bounds, while Theorem~\ref{th:tight} shows that they converge to the exact value under almost-sure termination and boundedness of the observable. We then establish a connection with FK models: in particular Theorem~\ref{th:filtlift} reformulates the approximation bounds for lifted terminal-state observables in terms of the \emph{filtering} distributions of a PPG-induced FK model---distributions that can be consistently and efficiently estimated via particle filtering (PF).}

\paragraph{Efficient and parallel inference (Section \ref{sec:experiments})}
A central insight of our approach is that the operational structure of PPGs enables a naturally parallelizable implementation of inference. Since all particles evolve synchronously (in lock-step) through the same control-flow graph (PPG), and conditioning  is applied via score functions in a uniform, aligned fashion, our particle filtering algorithm maps directly to SIMD-style vectorized execution. This design avoids the particle misalignment issues that affect continuation-based or functional semantics for PPLs. As a result, our operational model is easily mapped into  modern hardware architectures supporting data-level parallelism.
%
Our vectorized implementation of a PPG-based particle filter, VPF, practically demonstrates the effectiveness of our approach. On challenging examples involving nested conditioning and unbounded loops, VPF matches or outperforms state-of-the-art probabilistic programming systems in both accuracy and runtime. 

\paragraph{Future directions}
On the practical side, developing   compilers from high-level PPLs to PPGs, extending the framework to richer type systems and data structures is a natural next-step. Combining a sampling-based approach with symbolic or constraint-based reasoning techniques is a challenging theoretical direction.

\newpage

\bibliographystyle{plain}

\bibliography{generic2}

\iffull
\newpage
\appendix
\section{Proofs}\label{app:proofs}
The following result, which subsumes Theorem \ref{th:prod}, is   well known from measure theory. The formulation below  is a specialization of \cite[Th.2.6.7]{Ash} to Markov kernels and nonnegative functions.
Part (a) gives a way to construct a measure on the product space $\Omega^t$, starting from an initial measure $\mu^1$   and $t-1$ Markov kernels. The  product space is, intuitively,      the sample space of the paths of length $t$ of a Markov chain. In particular, a path of length $t=1$  consists of just an initial state --- no transition has been fired.  Part    (b) is a generalization  of   Fubini theorem, which allows one to express  an integral over the product space w.r.t. the measure of part (a) in terms of iterated integrals over the component spaces. Below, we will let $\omega^t$   range over $\Omega^t$.
\begin{theorem}[product of measures]\label{th:prode}
Let  $t\geq 1$ be an integer. Let $\mu^1$ be a probability measure on $\Omega$ and $K_2,...,K_t$ be $t-1$   (not necessarily distinct) Markov kernels from $\Omega$ to $\Omega$.
\begin{itemize}
	\item[(a)] There is a unique probability measure $\mu^t$ defined on $(\Omega^t,\F^t)$ such that for every $A_1\times \cdots\times A_t\in \F^t$ we have:
	{\small\begin{equation}\label{eq:prodmeas}
			\begin{array}{rcl}
				\mu^t(A_1\times \cdots\times A_t)&=&\int_{A_1} \mu^1(d\omega_1)\int_{A_2}K_2(\omega_1)(d\omega_2)\cdots \int_{A_t}K_t(\omega_{t-1})(d\omega_t)\,.
			\end{array}
	\end{equation}}\noindent
	\item[(b)] (Fubini) Let $f$ be a nonnegative measurable function defined on $\Omega^t$. Then, letting $\omega^t=(\omega_1,...,\omega_t)$, we have
	{\small\begin{equation}\label{eq:prodInt}
			\begin{array}{rcl}
				\int \mu^t(\omega^t)f(\omega^t) &=&\int  \mu^1(d\omega_1)\int K_2(\omega_1)(d\omega_2)\cdots \int K_t(\omega_{t-1})(d\omega_t)f(\omega^t)\,.
			\end{array}
	\end{equation}}\noindent
	In particular, on the right-hand side, for each $j=1,...,t-1$ and $(\omega_1,...,\omega_{j-1})$, the function $\omega_j\mapsto   \int K_{j+1}(\omega_{j})(d\omega_{j+1})\cdots \int K_t(\omega_{t-1})(d\omega_t)f(\omega^t)$ is measurable over $\Omega$.
\end{itemize}
\end{theorem}

\vsp
We now proceed to the proof of the results stated in the paper.
\vsp

\begin{proofof}{Lemma \ref{lemma:MKO}}
As a function $\ereals^{m+1}\times \F\rightarrow \ereals^+$, $\K$ can be written as follows:
\begin{align}\label{eq:kvz}
	\K(v,z)(A)&=[z\notin \St]\cdot \delta_{  v   }(A_z)\,+\,  \sum_{(S,\varphi,\zeta,S')\in E}  [z=S]\cdot\varphi(v)\cdot \zeta(v)(A_{S'})\,.
\end{align}
We now check the two properties required by Definition \ref{def:MK}.
\begin{itemize}
	\item For each $(v,z)\in\Omega$, the function $A\mapsto \K(v,z)(A)$ is a probability measure.  Consider  the right-hand side of \eqref{eq:kvz}: since for each $S$, $\sum_{(S,\varphi,\zeta,S')\in E}   \varphi(v)=1$, for the given $(v,z)$  exactly one of the summands is different from the constant 0 function. In particular,  there is   a probability measure   $\nu$ on $\F_m$ such that for every $A\in \F$, we have $\K(v,z)(A)=\nu(A_z)$.  Next, we note the following general property of sections of measurable sets, which can be shown by elementary set-theoretic reasoning: for any measurable set $A\in \F$ s.t. $A=\cup_{j\geq 0}A_j$  (disjoint union   of measurable sets) and $z\in \ereals$, we have $A_z=\cup_{j\geq 0}(A_j)_z$ (disjoint union   of measurable sets). Applying the  two facts just established and the additivity of the measure $\nu$, we have:  $  \K(v,z)(A)=  \nu((\cup_{j\geq 0}A_j)_z)= \nu(\cup_{j\geq 0}(A_j)_z)=\sum_{j\geq 0} \nu((A_j)_z)= \sum_{j\geq 0} \K(v,z)(A_j)$. This shows that $\K(v,z)$ is a measure. Moreover, $\K(v,z)(\Omega)=\nu(\Omega_z)=\nu(\ereals^m)=1$, which completes the proof that $\K(v,z)$ is a probability measure.
	\item  For each $A\in \F$, the function $(v,z)\mapsto \K(v,z)(A)$ is nonnegative and measurable. Consider  again the right-hand side of \eqref{eq:kvz}, but write $ \delta_{  v   }(A_z)$ as the indicator function  $1_{A_z}(v)$: as a function of $v$, this is measurable (as $A_z$ is a measurable set).  Moreover, for any $A$ and $S'$,  $\zeta(v)(A_{S'})$ is a measurable function of $v$ (as $\zeta$ is a Markov kernel);  hence for each $\varphi$, also  $\varphi(v)\cdot\zeta(v)(A_{S'})$ is a measurable function of $v$. But any measurable function of $v$ alone, say $h(v)$, is also a measurable function of $(v,z)$ (i.e. the function obtained by composing the projection $(v,z)\mapsto v$   with $v\mapsto h(v)$). Likewise, the predicates $[z\notin \St]$ and $[z=S]$ (for any fixed $S\in \St$) are measurable functions of $z$, hence of $(v,z)$. As $\K(\cdot)(A)$ is obtained by   products and sums of  nonnegative measurable functions of $(v,z)$, it is a measurable function of $(v,z)$ \cite[Ch.1,Th.1.5.6]{Ash}.
\end{itemize}
\end{proofof}

\vsp

The following is a  general lemma  useful to connect measure over sets of infinite and finite traces. In its statement, we let $\mu^\infty$ denote a generic measure on the cylindrical sigma-field, obtained as   an infinite product of kernels in the sense of the Ionescu-Tulcea theorem, and by $\mu^t$ the corresponding finite product measures. We shall make use of the following notation.  For $\tilde\omega=(\omega_1,\omega_2,...)\in \Omega^\infty$, we let $\tilde\omega_{1:t}:=(\omega_1,...,\omega_t)\in \Omega^t$. For
$h:\Omega^t\rightarrow \erealspl$ a nonnegative   function, we let   $\tilde h:\Omega^\infty \rightarrow \erealspl$ be defined as  follows for each $\tilde\omega\in \Omega^\infty$:
\begin{align}\label{eq:htilde}
	\tilde h(\iomega):=h(\iomega_{1:t})\,.
\end{align}

In what follows, we shall make use of the following properties of measurable cylinders: for measurable $A,B\subseteq \Omega^t$ ($t\geq 1$), we have $\cy{A\cup B}=\cy{A}\cup\cy{B}$ and  $\cy{A\cap B}=\cy{A}\cap\cy{B}$.

\begin{lemma}\label{lemma:aux0}  Let
	$h:\Omega^t\rightarrow \erealspl$ a nonnegative measurable function. Then $\tilde h$ as defined in \eqref{eq:htilde} is measurable.   Moreover, for each measurable cylinder $\cy{B_t}\subseteq \Omega^\infty$ ($B_t\subseteq \Omega^t$),  we have $\int_{\cy{B_t}} \mu^\infty(d\iomega) \tilde h(\iomega)=\int_{{B_t}} \mu^t(d\omega^t)   h(\omega^t)$.
\end{lemma}
\begin{proof}
First, consider the case of indicator functions $h=1_{A_t}$, for a measurable $A_t\subseteq \Omega^t$.  Then $\tilde h=1_{\cy{A_t}}$, the indicator function of the measurable cylinder generated by $A_t$, and the statement is obvious, because $  h$ is measurable, and $\int_{\cy{B_t}} \mu^\infty(d\iomega) \tilde h(\iomega)=\int  \mu^\infty(d\iomega)   \tilde h(\iomega)1_{\cy{B_t}}(\iomega)=\mu^\infty(\cy{B_t}\cap \cy{A_t})=\mu^\infty(\cy{B_t \cap  A_t})=\mu^t(B_t\cap A_t)=\int_{B_t} \mu^t(d\omega^t)   h(\omega^t)$. The statement for the general case of $h$ follows then by \minorhighlight{approximating $h$ by an increasing sequence of nonnegative simple functions and applying the Monotone Convergence Theorem}.
\end{proof}

\vsp

We can now readily establish measurability of various functions used throughout the paper.

\begin{lemma}[measurability of  functions]\label{lemma:aux}
Let $t\geq 1$. The following functions are measurable:
(1) $\wt_t:\Omega^t\rightarrow [0,1]$;
(2) $\wt: \Omega^\infty\rightarrow [0,1]$;
(3) 
$f_t:\Omega^t\rightarrow \ereals^+$, provided   $f:\Omega^\infty\rightarrow \ereals^+$  is measurable.
%
\end{lemma}
\begin{proof}
%
%
Concerning parts 1 and 2, first  note one can write the score function on $\Omega$ (Definition \eqref{eq:scoz}) as:
$\scoz(v,z)=[z\notin \St]+\sum_{S\in \St}[z=S]\cdot \psco(S)(v)$, where $\gamma=\psco(S)$ is a measurable score function on $\ereals^m$. This easily implies that $\scoz(\cdot)$ is measurable on $\Omega$ (cf. also the proof of Lemma  \ref{lemma:MKO}, second item). As $\wt_t(\omega^t)=\scoz(\omega_1)\cdots\scoz(\omega_t)$ is the product of measurable functions on $\Omega$, it is a measurable function on $\Omega^t$.
Now consider $\widetilde{\wt}_t:\Omega^\infty\rightarrow [0,1]$: applying Lemma  \ref{lemma:aux0} with $h=\wt_t$, we deduce that $\widetilde{\wt}_t$ is measurable.
Finally, as $t\rightarrow +\infty$, it is seen that $\widetilde{\wt}_t\rightarrow \wt$ pointwise: then $\wt$ is measurable as well, because it is the pointwise limit of a sequence of measurable functions  functions, cf. \cite[Th.1.5.4]{Ash}.

Concerning part 3, define the \emph{$t$-section of $C\in \Cyl$ at} $\iomega\in \Omega^\infty$ as $C^t_{\iomega}:=\{\omega^t\in \Omega^t\,:\,(\omega^t,\iomega)\in C\}$.  A proof very similar to that given  in \cite[Th.2.6.2(1)]{Ash} for finite products shows that $t$-sections of elements in $\Cyl$  are measurable. Now let $A$ be any measurable set in $\ereals$. By definition, $f^{-1}(A)=\{(\omega^t,\iomega)\,:\,f(\omega^t,\iomega)\in A\}$ is measurable. The $t$-section at $\iomega=\star^\infty$ of this set is precisely $f_t^{-1}(A)$, hence it is measurable. This implies that $f_t$ is measurable.

\end{proof}

\vsp
We now turn to the proof of Theorem \ref{th:approx}. We first prove a few  results about  to the measures $\mu^\infty_S$ and  $\mu^t_S$.  In what follows, we shall  consider the   the notation  $\int f\,d\mu$ for $\int  \mu (d\omega) f(\omega)$. We shall use the two notations interchangeably; the second one is more convenient for expressing iterated integrals. We need the following definition.

\begin{definition}[consistent paths]\label{def:consistent}
We define the following measurable subsets of $\Omega^t$ ($t\geq 1$) and $\Omega^\infty$:
\begin{align}\label{eq:T}
	\T^{\leq t}  := \cup_{j= 0}^{t-1}\,(\term^{\mathrm{c}})^{j} \cdot   \term^{t-j}\quad\quad
	\T_{t }   := \T^{\leq t} \cup  (\term^{\mathrm{c}})^{t} \quad\quad
	\T&:= \left(\cup_{j\geq 0}\,(\term^{\mathrm{c}})^{j} \cdot   \term^\infty\right) \,\cup\,(\term^{\mathrm{c}})^\infty\,.
\end{align}
\end{definition}

\vsp
\begin{lemma}\label{lemma:zero}
(a) $\mu^\infty_S(\T)=1$. Hence for any measurable set $A$ and nonnegative measurable $f$ defined on $\Omega^\infty$, $\int_A f \,d\mu^\infty_S=\int_{A\cap \T}f\,d\mu^\infty_S$.

(b)  Let $t\geq 1$. Then $\mu^t_S(\T_t)=1$. Hence for any measurable set $A$ and nonnegative measurable $f$ defined on $\Omega^t$, $\int_A f \,d\mu^t_S=\int_{A\cap \T_t}f\,d\mu^t_S$.
\end{lemma}
\begin{proof}  Let us consider part (a).
Consider $\T^{\mathrm{c}}=\cup_{j\geq 0}\cy{\Omega^j\cdot \term\cdot\term^{\mathrm{c}}}$;  note that this union is in general not disjoint, but this is not relevant for the rest of the proof. For any $j\geq 0$, we will show that $\mu^\infty_S(\cy{\Omega^j\cdot \term\cdot\term^{\mathrm{c}}})=0$, which implies the thesis. Indeed, $\mu^\infty_S(\cy{\Omega^j\cdot \term\cdot\term^{\mathrm{c}}})= \mu^{j+2}_S({\Omega^j\cdot \term\cdot\term^{\mathrm{c}}})$, by definition of the product measure $\mu^\infty_S$.     Theorem \ref{th:prode}(a) gives us
{\small
	\begin{align}\label{eq:cons}
		\mu^{j+2}_S({\Omega^j \term \term^{\mathrm{c}}})& =\int \delta_{(0,S)}(d\omega_1)\int \K(\omega_1)(d\omega_2)\int\cdots
		\int_{\term} \K(\omega_{j})(d\omega_{j+1})\int_{\term^{\mathrm{c}}} \K(\omega_{j+1})(d\omega_{j+2})1\,.
	\end{align}
}\noindent
Considering  the innermost two integrals in the above expression,   let $J(\omega_j):= \int_{\term} \K(\omega_{j})(d\omega_{j+1})\int_{\term^{\mathrm{c}}} \K(\omega_{j+1})(d\omega_{j+2})1= \int \K(\omega_{j})(d\omega_{j+1})\int  \K(\omega_{j+1})(d\omega_{j+2})1_{\term}(\omega_{j+1})\cdot 1_{\term^{\mathrm{c}}}(\omega_{j+2})$.
Suppose $\omega_{j+1}=(v,\nil)\in \term$: then considering the innermost integral in $J(\omega_j)$, by definition of $\K$ we have $\int \K(\omega_{j+1})(d\omega_{j+2})1_{\term}(\omega_{j+1})\cdot 1_{\term^{\mathrm{c}}}(\omega_{j+2})= \int  \K(v,\nil)(d\omega_{j+2})  1_{\term^{\mathrm{c}}}(\omega_{j+2})=\int  \delta_{(v,\nil)}(d\omega_{j+2})1_{\term^{\mathrm{c}}}(\omega_{j+2})= 1_{\term^{\mathrm{c}}}(v,\nil)=0$. Similarly, we have that the innermost integral is 0 if  $\omega_{j+1} \in \term^{\mathrm{c}}$. This implies that $J(\omega_j)=0$, hence the integral in \eqref{eq:cons} is 0.

\orangehighlight{For part (b), observe that, with the convention that the union is empty when $t=1$,
\[
\T_t^{\mathrm c}=\bigcup_{r=1}^{t-1}\Omega^{r-1}\cdot\term\cdot\term^{\mathrm c}\cdot\Omega^{t-r-1}.
\]
Fix $1\leq r\leq t-1$. By Theorem \ref{th:prode}(b), the probability of the $r$th set in this union can be written as an iterated integral. After the integrations up to $\omega_r$, the remaining expression contains, for $\omega_r\in\term$, the factor
\[
\int_{\term^{\mathrm c}}\K(\omega_r)(d\omega_{r+1})1.
\]
If $\omega_r=(v,\nil)$, then $\K(\omega_r)=\delta_{\omega_r}$, hence this factor is
$\delta_{\omega_r}(\term^{\mathrm c})=0$. Thus every set in the finite union has $\mu^t_S$-measure $0$, and therefore $\mu^t_S(\T_t)=1$. The integral identities in (a) and (b) then follow from the fact that, if $B$ has full measure and $q$ is nonnegative measurable, then $\int_Aq\,d\mu=\int_{A\cap B}q\,d\mu$.}
\end{proof}


\vsp

\begin{lemma}\label{lemma:conv} Let $S\in\St$. As $t\rightarrow +\infty$, we have $[S]^t 1_{\term^{\leq t}}\cdot\wt_t  \longrightarrow [S]  1_{\term_{\mathrm{f}}}\cdot \wt $. Moreover, the sequence  $[S]^t 1_{\term^{\leq t}}\cdot\wt_t$ ($t\geq 1$) is monotonically nondecreasing.
\end{lemma}
\begin{proof}
\orangehighlight{For $t\geq1$, let
\[
h_t:=1_\Theta\cdot \widetilde{(1_{\term^{\leq t}}\wt_t)}.
\]
The function $h_t$ is measurable by Lemma \ref{lemma:aux0} and the measurability of $\Theta$ from Definition \ref{def:consistent}. We first identify it pointwise on $\Theta$. Let $\tilde\omega=(\omega_1,\omega_2,\ldots)\in\Theta$. If $\tilde\omega\notin\cy{T^{\leq t}}$, then $h_t(\tilde\omega)=0$. If instead $\tilde\omega\in\cy{T^{\leq t}}$, let $j\leq t$ be its first termination time, so $\omega^j\in T_j$. By the definition of $\Theta$, $\omega_i\in\term$ for every $i\geq j$. Hence $\scoz(\omega_i)=1$ for every $i>j$, because the checkpoint component of a terminated state is $\nil\in\St$. Therefore
\[
\widetilde{\wt_t}(\tilde\omega)=\prod_{i=1}^t\scoz(\omega_i)
=\prod_{i=1}^j\scoz(\omega_i)
=\prod_{i=1}^{\infty}\scoz(\omega_i)=\wt(\tilde\omega).
\]
Consequently, for every $\tilde\omega\in\Omega^\infty$,
\begin{align}\label{eq:conv-pointwise}
h_t(\tilde\omega)=1_\Theta(\tilde\omega)\,1_{\cy{T^{\leq t}}}(\tilde\omega)\,\wt(\tilde\omega).
\end{align}
Since $\cy{T^{\leq t}}\uparrow T_\fini$, identity \eqref{eq:conv-pointwise} shows directly that $(h_t)_{t\geq1}$ is nondecreasing and that $h_t\uparrow1_\Theta1_{T_\fini}\wt$ pointwise. The Monotone Convergence Theorem \cite[Th.1.6.2]{Ash} therefore gives
\[
\int h_t\,d\mu_S^\infty\longrightarrow\int1_\Theta1_{T_\fini}\wt\,d\mu_S^\infty.
\]
Finally, Lemma \ref{lemma:zero}(a) and Lemma \ref{lemma:aux0}, with $B_t=\Omega^t$, yield
\[
\int h_t\,d\mu_S^\infty=\int\widetilde{(1_{T^{\leq t}}\wt_t)}\,d\mu_S^\infty
=\int1_{T^{\leq t}}\wt_t\,d\mu_S^t=[S]^t1_{T^{\leq t}}\wt_t,
\]
and similarly
\[
\int1_\Theta1_{T_\fini}\wt\,d\mu_S^\infty
=\int1_{T_\fini}\wt\,d\mu_S^\infty=[S]1_{T_\fini}\wt.
\]
This proves both assertions.}
%
\end{proof}

\vsp
\begin{lemma}[terminated-prefix identity]\label{lemma:pfsupp}
\orangehighlight{Let $f$ be termination-determined and let $t\geq 1$. Then}
\begin{align}\label{eq:ftterm}
\orangehighlight{[S]^t f_t\cdot 1_{T^{\leq t}}\cdot\wt_t
= [S]\bigl(f\cdot\wt\cdot 1_{\cy{T^{\leq t}}}\bigr).}
\end{align}
\end{lemma}
\begin{proof}
\orangehighlight{Set $g_t:=f_t1_{T^{\leq t}}\wt_t$. By Lemma \ref{lemma:aux0}, with $B_t=\Omega^t$,
\begin{align}
[S]^t g_t=\int g_t\,d\mu_S^t=\int\widetilde{g_t}\,d\mu_S^\infty.\label{eq:ftterm-lhs}
\end{align}
On the other hand, by definition,
\begin{align}
[S]\bigl(f\wt1_{\cy{T^{\leq t}}}\bigr)=\int f\wt1_{\cy{T^{\leq t}}}\,d\mu_S^\infty.\label{eq:ftterm-rhs}
\end{align}
By Lemma \ref{lemma:zero}(a), both integrals are unchanged if restricted to $\Theta$. We therefore prove the pointwise identity
\begin{align}
\widetilde{g_t}(\tilde\omega)=f(\tilde\omega)\wt(\tilde\omega)1_{\cy{T^{\leq t}}}(\tilde\omega)
\qquad(\tilde\omega\in\Theta).\label{eq:ftterm-pointwise}
\end{align}
Fix $\tilde\omega=(\omega_1,\omega_2,\ldots)\in\Theta$ and write $\omega^t=(\omega_1,\ldots,\omega_t)$. If $\omega^t\notin T^{\leq t}$, then $1_{T^{\leq t}}(\omega^t)=0$ and $1_{\cy{T^{\leq t}}}(\tilde\omega)=0$, so both sides of \eqref{eq:ftterm-pointwise} are $0$.

Suppose now that $\omega^t\in T^{\leq t}$. Since
\[
T^{\leq t}=\bigcup_{j=1}^t T_j\cdot\Omega^{t-j}
\]
is a disjoint union, there is a unique $j\leq t$ such that $\omega^j\in T_j$. By termination-determinedness, the two infinite traces $\tilde\omega$ and $(\omega^t,\star^\infty)$, which have the same prefix $\omega^j\in T_j$, have the same $f$-value. Hence
\begin{align}
f_t(\omega^t)=f(\omega^t,\star^\infty)=f(\tilde\omega).\label{eq:ftterm-f}
\end{align}
Moreover, because $\tilde\omega\in\Theta$ and its first terminated state occurs at time $j$, we have $\omega_i\in\term$ for every $i\geq j$. For every terminated state $\omega_i=(v_i,\nil)$, Definition \eqref{eq:scoz} gives $\scoz(\omega_i)=1$. Therefore
\begin{align}
\wt_t(\omega^t)&=\prod_{i=1}^t\scoz(\omega_i)
=\prod_{i=1}^j\scoz(\omega_i)
=\prod_{i=1}^{\infty}\scoz(\omega_i)
=\wt(\tilde\omega).\label{eq:ftterm-w}
\end{align}
Since both indicators in \eqref{eq:ftterm-pointwise} are now equal to $1$, equations \eqref{eq:ftterm-f} and \eqref{eq:ftterm-w} prove \eqref{eq:ftterm-pointwise}. Integrating that identity over $\Theta$, and using Lemma \ref{lemma:zero}(a) together with \eqref{eq:ftterm-lhs}--\eqref{eq:ftterm-rhs}, yields \eqref{eq:ftterm}.}
\end{proof}

\vsp
\begin{proofof}{Lemma \ref{lemma:basic}}
\orangehighlight{Using Lemma \ref{lemma:pfsupp},}
\begin{align*}
[S]^t f_t1_{T^{\leq t}}\wt_t
&=[S]\bigl(f\wt 1_{\cy{T^{\leq t}}}\bigr)
\leq [S]f\wt,
\end{align*}
\orangehighlight{which is the lower bound. For the upper bound, split the integral according to whether termination has occurred by time $t$:}
\begin{align*}
[S]f\wt
&=[S]\bigl(f\wt 1_{\cy{T^{\leq t}}}\bigr)+[S]\bigl(f\wt 1_{\cy{T^{\leq t}}^{\mathrm c}}\bigr)\\
&\leq [S]^t f_t1_{T^{\leq t}}\wt_t
   +M\,[S]\bigl(\wt 1_{\cy{T^{\leq t}}^{\mathrm c}}\bigr)\\
&\leq [S]^t f_t1_{T^{\leq t}}\wt_t
   +M\,[S]\bigl(\widetilde{\wt_t}\,1_{\cy{T^{\leq t}}^{\mathrm c}}\bigr)\\
&=[S]^t f_t1_{T^{\leq t}}\wt_t
   +M\,[S]^t(1-1_{T^{\leq t}})\wt_t.
\end{align*}
\orangehighlight{Here $f\leq M$ gives the first inequality. The second one uses the pointwise inequality $\wt\leq\widetilde{\wt_t}$, which follows because all score factors lie in $[0,1]$. Finally, $\cy{T^{\leq t}}^{\mathrm c}=\cy{\Omega^t\setminus T^{\leq t}}$ and therefore
\[
\widetilde{\wt_t}1_{\cy{T^{\leq t}}^{\mathrm c}}
=\widetilde{((1-1_{T^{\leq t}})\wt_t)};
\]
Lemma \ref{lemma:aux0}, with $B_t=\Omega^t$, gives the last equality. This proves \eqref{eq:basic}.}
\end{proofof}


\ifmai
\begin{theorem}[finite approximation]\label{th:approx} Consider $S\in\St$ and $t\geq 1$ such that $[S]^t 1_{T^{\leq t}}\cdot \wt_t>0$. 
	Then for any  prefix-closed   function $f$ with branches $L_0,L_1,...$ we have that $\sem{S}f$ is well defined. 
	Moreover,  given  an upper bound  $f\leq M$ ($M\in \erealspl$),    for each $t$ large enough and $\alpha_t:=\frac{[S]^t \wt_t}{[S]^t 1_{T^{\leq t}}\cdot \wt_t}$ we have:
	\begin{equation}\label{eq:approx}
		\begin{array}{rcccl}
			\dfrac{[S]^t f_t\cdot 1_{L^{\leq t}\cap T^{\leq t}}\cdot \wt_t}{[S]^t \wt_t} &\leq & \sem{S}f&
			\leq &
			\dfrac{[S]^t f_t\cdot 1_{L^{\leq t}\cap T^{\leq t}}\cdot \wt_t}{[S]^t \wt_t}\alpha_t+M\cdot   \left(\alpha_t   -1\right)\,.
		\end{array}
	\end{equation}
\end{theorem}
\fi
\vsp
\vsp
\begin{proofof}{Theorem \ref{th:approx}}
	We first show that $\sem S f$ is well defined, that is that $[S]\wt >0$.
	Indeed, from Lemma \ref{lemma:conv}, and from $ [S]^t 1_{\term^{\leq t}}\cdot\wt_t > 0$  for at least one   $t$,   we get $[S]  1_{\term_{\mathrm{f}}}\cdot \wt >0$; since $1_{\term_{\mathrm{f}}}\cdot \wt \leq \wt$, we get $[S]  1_{\term_{\mathrm{f}}}\cdot \wt\leq [S]   \wt$, hence the wanted statement. \minorhighlight{Moreover, by the monotonicity statement in Lemma \ref{lemma:conv}, for every $t\geq t_0$ we have $[S]^t1_{T^{\leq t}}\cdot\wt_t>0$, and therefore also $[S]^t\wt_t>0$.}
	\ifmai
	Now, using the notation introduced in \eqref{eq:htilde},  for any $t\geq 1$  consider the function $\tilde\wt_t
	$:
	applying Lemma \ref{lemma:aux0} with $h=\wt_t$ and $B_t=\Omega^t$, we obtain that $\int\mu^{\infty}_S(d\iomega)\tilde\wt_t(\iomega)= \int\mu^{t}_S(d\omega^t) \wt_t(\omega^t)=[S]^t \wt_t$. Moreover,  $\{\tilde \wt_t\}_{t=1}^\infty$ is a monotonically nondecreasing sequence of nonnegative functions over $\Omega^\infty$, that converges pointwise to $\wt$; by the Monotone Convergence Theorem \cite[Th.1.6.2]{Ash}, $[S]^t   \wt_t = \int\mu^{\infty}_S(d\iomega)\tilde\wt_t(\iomega) \rightarrow \int\mu^{\infty}_S(d\iomega) \wt (\iomega) =[S]\wt$; moreover the sequence of integrals is in turn nondecreasing, which implies that $[S]\wt >0$, since $[S]^t \wt_t >0$ for at least one $t$.
	\fi
	
	Now consider $\sem S f =\dfrac{[S]f\cdot \wt}{[S]\wt}$, for $f$ like in the hypothesis, and the inequalities in \eqref{eq:approx}. Consider the following bounds for the numerator and denominator of this fraction.
	{
		\begin{align}
			\greenhighlight{[S]^t f_t\cdot 1_{T^{\leq t}}\cdot \wt_t} & \leq  [S]f\cdot \wt  \leq    \greenhighlight{[S]^t f_t\cdot 1_{T^{\leq t}}\cdot \wt_t} + M([S]^t  \wt_t- [S]^t 1_{T^{\leq t}}\cdot \wt_t)\label{eq:B1}\\
			[S]^t 1_{T^{\leq t}}\cdot\wt_t  &\leq  [S] \wt \leq  [S]^t  \wt_t\,.\label{eq:B2}
		\end{align}
	}\noindent
	The bounds in \eqref{eq:B1} are just those in Lemma \ref{lemma:basic}, with the   term $M\cdot(\cdots)$  written in an equivalent form. 
	\orangehighlight{As to \eqref{eq:B2}, apply Lemma \ref{lemma:basic} to the constant function $f=1$, which is termination-determined, with $f_t=M=1$ identically. The bounds \eqref{eq:basic} then specialize immediately to \eqref{eq:B2}.}
	%
	%
	From the above established bounds \eqref{eq:B1} and \eqref{eq:B2} for the numerator and denominator of $\sem S f =\frac{[S]f\cdot \wt}{[S]\wt}$, it follows that
	{\small
		\begin{equation}\label{eq:approx2}
			\begin{array}{rcccl}
				\dfrac{\greenhighlight{[S]^t f_t\cdot 1_{T^{\leq t}}\cdot \wt_t}}{[S]^t \wt_t} &\leq & \sem{S}f&
				\leq &
				\dfrac{\greenhighlight{[S]^t   f_t\cdot  1_{T^{\leq t}}\cdot \wt_t}}{[S]^t 1_{T^{\leq t}}\cdot \wt_t}+M\cdot   \left(\dfrac{[S]^t \wt_t }{[S]^t 1_{T^{\leq t}}\cdot \wt_t}   -1\right)  \,.
			\end{array}
		\end{equation}
	}\noindent
	Now, multiplying and dividing the first term of the above upper bound by $[S]^t \wt_t$, \minorhighlight{positive for every $t\geq t_0$ by the preceding observation}, and recalling the definition of $\alpha_t$, the wanted \eqref{eq:approx} follows.
\end{proofof}

\vsp
\vsp
\begin{proofof}{Theorem \ref{th:tight}}
Write  \eqref{eq:approx}  in the form \eqref{eq:approx2}.
We discuss the limit as $t\rightarrow+\infty$ of each of the  three distinct involved terms .
\begin{enumerate}
	\item \orangehighlight{$[S]^t f_t1_{T^{\leq t}}\wt_t\rightarrow [S]f\wt$.} By Lemma \ref{lemma:pfsupp},
	\begin{align}\label{eq:int0}
	\orangehighlight{[S]^t f_t1_{T^{\leq t}}\wt_t
	=\int f\wt\,1_{\cy{T^{\leq t}}}\,d\mu^\infty_S.}
	\end{align}
	\orangehighlight{The sets $\cy{T^{\leq t}}$ increase to $T_\fini$. Hence, by the Monotone Convergence Theorem \cite[Th.1.6.2]{Ash}, the right-hand side of \eqref{eq:int0} converges to $\int f\wt 1_{T_\fini}\,d\mu^\infty_S$. Since $\mu^\infty_S(T_\fini)=1$, this is $[S]f\wt$.}
		\item \orangehighlight{$[S]^t \wt_t\rightarrow [S]\wt$. By taking $B_t=\Omega^t$ in Lemma \ref{lemma:aux0}, we have $[S]^t \wt_t= \int \wt_t\,d\mu_S^t=\int \widetilde{\wt_t}\,d\mu_S^\infty$. Moreover, $\widetilde{\wt_t}\rightarrow\wt$ pointwise by definition of the infinite weight, and $0\leq\widetilde{\wt_t}\leq1$. Hence the Dominated Convergence Theorem \cite[Th.1.6.9]{Ash} yields
	\[
	\int \widetilde{\wt_t}\,d\mu_S^\infty\longrightarrow\int\wt\,d\mu_S^\infty=[S]\wt.
	\]}
	
	\item $[S]^t 1_{T^{\leq t}}\cdot \wt_t \rightarrow [S]\wt$. \orangehighlight{Apply the first item above to the constant termination-determined function $f=1$.} 
\end{enumerate}

As to \eqref{eq:exact}, the hypothesis $[S]^t1_{T^{\leq t}}=1$ is exactly $\mu^t_S(T^{\leq t})=1$. Hence $\mu^t_S((T^{\leq t})^{\mathrm c})=0$, and therefore, for every nonnegative measurable $q$ on $\Omega^t$,
\begin{align}\label{eq:fullmeasure-exact}
\int q\,d\mu^t_S=\int q\,1_{T^{\leq t}}\,d\mu^t_S.
\end{align}
Applying \eqref{eq:fullmeasure-exact} first to $q=\wt_t$ gives
\[
[S]^t1_{T^{\leq t}}\wt_t=[S]^t\wt_t,
\]
so $\alpha_t=1$. Hence the lower and upper bounds in \eqref{eq:approx} coincide with
\orangehighlight{$\frac{[S]^t f_t\cdot 1_{T^{\leq t}}\cdot\wt_t}{[S]^t\wt_t}$}.
Applying \eqref{eq:fullmeasure-exact} again, now to $q=f_t\wt_t$, gives
\orangehighlight{$[S]^t f_t\wt_t=[S]^t f_t1_{T^{\leq t}}\wt_t$}.
Substitution yields \eqref{eq:exact}.
\end{proofof}

\vsp
\vspace*{.3cm}
The following   result is   useful to relate our semantics to the filtering distribution of a Feynman-Kac model. Here we let   $\pr_j:\Omega^t\rightarrow \Omega$  ($1\leq j\leq t$) denote  the projection on the $j$th component; this is a measurable function. The result basically says that taking the expectation of $f_t$   on paths of length $t$   is the same as taking the expectation of   $h\circ \pr_t$, that only looks at the last state of a path.

\begin{lemma}\label{lemma:simple}
Let $f=\lift h$ for a nonnegative measurable $h$ defined on $\Omega$.
\orangehighlight{For each $S$ and $t\geq 1$, we have $[S]^t f_t\cdot 1_{T^{\leq t}}\cdot\wt_t= [S]^t( h \circ \pr_t) \cdot \wt_t$.}
\end{lemma}
\begin{proof}
\orangehighlight{We first establish a property of the finite-trace measure after termination. Fix $1\leq j\leq t$ and define
\[
\eta_j(\omega^t):=\prod_{k=j+1}^t[\omega_k=\omega_j],
\]
with the empty product equal to $1$ when $j=t$. Each equality predicate $[\omega_k=\omega_j]$ is measurable (the diagonal of $\Omega\times\Omega$ is Borel), hence $\eta_j$ is measurable. Moreover, $\eta_j(\omega^t)=1$ exactly when $\omega_j=\omega_{j+1}=\cdots=\omega_t$. We claim that, for every nonnegative measurable $g:\Omega^t\to\ereals^+$,
\begin{align}
\int_{T_j\cdot\Omega^{t-j}}g\,d\mu^t_S
=
\int_{T_j\cdot\Omega^{t-j}}g\,\eta_j\,d\mu^t_S.
\label{eq:absorb-id}
\end{align}
To prove this, write both sides as iterated integrals using Theorem \ref{th:prode}(b). Fix the outer variables $\omega_1,\ldots,\omega_j$ with $\omega^j\in T_j$; in particular $\omega_j=(v,\nil)\in\term$ and therefore $\K(\omega_j)=\delta_{\omega_j}$. For the left-hand side, the $t-j$ innermost integrations are
\begin{align*}
I_g(\omega^j):={}&\int \K(\omega_j)(d\omega_{j+1})\cdots
\int \K(\omega_{t-1})(d\omega_t)\,g(\omega^t)\\
={}&g(\omega_1,\ldots,\omega_j,\omega_j,\ldots,\omega_j),
\end{align*}
where the second equality follows by repeated use of
$\int\delta_\omega(d\omega')q(\omega')=q(\omega)$; after the first evaluation, every subsequent kernel is again $\K(\omega_j)=\delta_{\omega_j}$. Replacing $g$ by $g\eta_j$ gives the same value, because
$\eta_j(\omega_1,\ldots,\omega_j,\omega_j,\ldots,\omega_j)=1$. Hence the inner blocks of the two iterated integrals agree for every $\omega^j\in T_j$; integrating the common value with respect to the outer variables proves \eqref{eq:absorb-id}.}

\orangehighlight{Now take $g=f_t\wt_t$. If $\omega^t\in T_j\cdot\Omega^{t-j}$, then the first terminated state of the infinite trace $(\omega^t,\star^\infty)$ occurs at time $j$; by the definition of lifting,
\[
f_t(\omega^t)=h(\omega_j).
\]
Using \eqref{eq:absorb-id}, followed by the identity
$\eta_j(\omega^t)h(\omega_j)=\eta_j(\omega^t)h(\omega_t)$, we obtain
\begin{align}
\int_{T_j\cdot\Omega^{t-j}}f_t\wt_t\,d\mu^t_S
&=\int_{T_j\cdot\Omega^{t-j}}f_t\wt_t\eta_j\,d\mu^t_S\\
&=\int_{T_j\cdot\Omega^{t-j}}\wt_t\eta_j h(\omega_j)\,d\mu^t_S\\
&=\int_{T_j\cdot\Omega^{t-j}}\wt_t\eta_j h(\omega_t)\,d\mu^t_S\\
&=\int_{T_j\cdot\Omega^{t-j}}(h\circ\pr_t)\wt_t\,d\mu^t_S,
\label{eq:lift-piece}
\end{align}
where the last equality is another application of \eqref{eq:absorb-id}, now to $g=(h\circ\pr_t)\wt_t$. Summing \eqref{eq:lift-piece} over $j=1,\ldots,t$ and using the disjoint union
$T^{\leq t}=\bigcup_{j=1}^t T_j\cdot\Omega^{t-j}$ gives
\begin{align}
[S]^t f_t1_{T^{\leq t}}\wt_t
=
\int_{T^{\leq t}}(h\circ\pr_t)\wt_t\,d\mu^t_S.
\label{eq:lift-term}
\end{align}
It remains to remove the restriction to $T^{\leq t}$. By Lemma \ref{lemma:zero}(b), $\mu^t_S(\Theta_t)=1$, hence
\[
\int_{T^{\leq t}}(h\circ\pr_t)\wt_t\,d\mu^t_S
=
\int_{T^{\leq t}\cap\Theta_t}(h\circ\pr_t)\wt_t\,d\mu^t_S.
\]
By Definition \ref{def:consistent}, $\Theta_t\setminus T^{\leq t}=(\term^{\mathrm c})^t$. On this set $\omega_t\notin\term$, and therefore $h(\omega_t)=0$ because $\supp(h)\subseteq\term$. Consequently
\[
\int_{T^{\leq t}\cap\Theta_t}(h\circ\pr_t)\wt_t\,d\mu^t_S
=
\int_{\Theta_t}(h\circ\pr_t)\wt_t\,d\mu^t_S
=
\int(h\circ\pr_t)\wt_t\,d\mu^t_S,
\]
where the last equality again uses Lemma \ref{lemma:zero}(b). Combining this with \eqref{eq:lift-term} proves the claim.}
\end{proof}

\vsp

\begin{proofof}{Theorem \ref{th:filtlift}}
\orangehighlight{By the definition of the $t$-th marginal in \eqref{eq:phit}, for every nonnegative measurable $h$ we have
\[
\expc_{\mufk_t}[h]=\expc_{\mufk}[h\circ\pr_t].
\]
Applying \eqref{eq:FCMe} to the nonnegative measurable function $h\circ\pr_t$ yields the identity}
\begin{align}\label{eq:fcexp}
	\expc_{\mufk_t}[h]=\frac{\expc_{\mu^t}[(h\circ\pr_t)\cdot   G]}{\expc_{\mu^t}[G]}\,.
\end{align}
\orangehighlight{Concerning $\beta_L$, from the lower bound in \eqref{eq:approx} and Lemma \ref{lemma:simple} we have
$$ \sem S f\geq \frac{[S]^t f_t\cdot 1_{T^{\leq t}}\cdot \wt_t}{[S]^t \wt_t}= \frac{[S]^t (h\circ \pr_t)\cdot \wt_t}{[S]^t \wt_t}=\expc_{\phi_{S,t}}[h],$$
where in the last step we have applied \eqref{eq:fcexp} to the model $\FC_S$, with $G=\wt_t$ as a global potential.}

Concerning $\beta_U$, consider the function $f=\lift{1}_\term$ (the lifting of the function $h=1_\term$): \orangehighlight{it is immediate from the definition of lifting that, for each $t\geq 1$, $1_{ T^{\leq t}}\cdot \wt_t=f_t\cdot 1_{ T^{\leq t}}\cdot \wt_t$.} Then we can
repeat the reasoning  used  above for $\beta_L$  with this   function $f$  to prove  that $\frac{[S]^t    1_{ T^{\leq t}}\cdot \wt_t}{[S]^t   \wt_t}=\expc_{\phi_{S,t}}[1_\term]$, that is $\alpha_t=\expc_{\phi_{S,t}}[1_\term]^{-1}$. Then the upper bound in \eqref{eq:approx}
allows us to complete  the proof.
\end{proofof}

\section{The Particle Filtering algorithm}\label{app:PF}
From a computational point of view, our interest in FK models lies in the fact that they allow for a simple, unified presentation of a class of efficient inference algorithms,  known as \emph{Particle Filtering (PF)} \cite{DelMoral04,SMC}. In what follows we present an algorithm to compute the filtering distribution $\phi_t$. We will introduce below a general PF algorithm scheme  following closely \cite[Ch.11]{SMC}.

Fix a generic FK model, $\FC=(\X,t,\mu^1,\{K_i\}_{i=2}^t,\{G_i\}_{i=1}^t)$. Fix  $N\geq 1$, the number of \emph{particles}, that is instances of the random process represented by the $K_i$'s, we want to simulate.  For any tuple $W=W^{1:N}=(W^{(1)},...,W^{(N)})$    of real nonnegative random variables, the  \emph{weights},   denote by $\widehat W$ the normalized version of $W$, that is\footnote{With the proviso that e.g. $\widehat W^{(i)}:=1/N$ in the event all the $W^{(i)}$'s are 0.  In the   execution of the PF algorithm this event will occur with   probability $\rightarrow 0$ as $N\rightarrow +\infty$.} $\widehat W^{(i)}=W^{(i)}/(\sum_{j=1}^N W^{(j)})$.  A \emph{resampling scheme} for  $(N,W)$ is $N$-tuple of random variables taking values on $1..N$, say $R=(R_1,...,R_N)$, such that:  for each $i\in 1..N$, letting $F_i$ denote the number of occurrences of $i$ in  $R$, one has
$$\expc[F_i|W]=N\cdot \widehat W_i\,.
$$
In other words,    $R$ is a randomized selection process of $N$ indices out of $1..N$, with repetitions, such that, on average,  each index $i\in 1..N$    is selected a number of times  proportional to its  weight in $W$. We shall write $R(W)$ to indicate that $R$ depends on a given weight vector $W$.  Various   resampling schemes have been proposed in the literature. Perhaps the simplest is letting $R$ be $N$ i.i.d. random variables each distributed according to $\widehat W$: this is known as  \emph{multinomial resampling}. We refer the reader to the specialized literature on PF for details and efficient implementation methods, see e.g.  \cite[Ch.9]{SMC} and references therein.

Algorithm \ref{alg:PF} is a generic PF algorithm. At the $k$-th iteration, for $k=1,...,t$, two $N$-tuples are extracted:
\begin{itemize}
\item a tuple of states $X_k=X^{1:N}_k=(X^{(1)}_k,...,X^{(N)}_k)\in \X^N$;
\item a tuple of (unnormalized) weights $W_k=W^{1:N}_k=(W^{(1)}_k,...,W^{(N)}_k)\in ({\reals^+})^N$.
\end{itemize}
The elements of $X^{1:N}_k$   depend   on the tuples   $X^{1:N}_{k-1}, W_{k-1}$ of the previous
iteration. The purpose of the resampling step 4 is to give more importance to particles with higher weight, when extracting the next tuple of particles, while discarding particles with lower weight. In  case $R$ is multinomial resampling,
steps 4-5 amount  to drawing each $X^{(j)}_k$   from the (empirical) distribution $\sum_{j=1}^N \widehat W^{(j)}_{k-1}\delta_{K_k(X^{(j)}_{k-1})}$. The weights $W^{1:N}_k$   are computed via   the potential function  $G_k$, and will be used in the resampling step at     iteration $k+1$, if $k<t$, or returned as part of the algorithm's output.

The following theorem states consistency, in an asymptotic sense, of the PF algorithm with respect to the filtering distribution $\phi_t$ on $\X$.  Its practical implication is that we can estimate expectations with respect to the filtering distribution as weighted sums. Note that in its   statement  $t$ is held fixed --- it is one of the parameter of the FK model --- while the number of particles $N$ tends to $+\infty$.

\ifmai
--------------TO PLACE---------------

Note that there are no loops where the number of iterations depends on $N$;  the \textbf{for} loop in lines 5--7   only scans the transitions set $E$, whose size is independent of $N$. Line 8 is just a vectorized implementation of sampling from the Markov kernel function in \eqref{eq:FCM}. Line 9 is a vectorized implementation of the combined score function \eqref{eq:scoz}.
In the actual TensorFlow implementation, the sums in lines 8 and 9  are encoded via boolean masking and vectorized operations.

--------------------------------------
\fi

\begin{theorem}[convergence of PF, \cite{SMC}]\label{th:PFconv} Consider the random variables  $(X_t ,W_t )$ ($t\geq 1$) as returned by   Algorithm  \ref{alg:PF}. Suppose that the FK measure $\mufk $ is well defined on $\X^t$. Further assume  that $R(\cdot)$ is multinomial resampling and that the potential functions $G_k$ all have a finite upper bound. \highlight{For each nonnegative measurable function $h$ defined on $\X$ such that $G_t\cdot h$ is bounded}, we have:
$\sum_{j=1}^N \widehat W^{(j)}_t\cdot h(X^{(j)}_t)\, \longrightarrow\,\expc_{\phi_t}[h]$   almost surely\footnote{See \cite[Ch.11]{SMC} for the precise definition of the probability space where this assertion makes sense.} as $N\rightarrow+\infty$.
\end{theorem}

\section{Additional details on experiments}\label{app:air}
We give a more detailed textual description of the considered examples.
\begin{enumerate}
\item \emph{Aircraft tracking} (AT) \cite{WuEtAl}. An aircraft is modeled as a point moving on a 2D plan according to a Gaussian process, The aircraft   is tracked by six radar: at each discrete time step, each radar  noisily  measures the distance of the aircraft from its own position;  specific distances   are being   observed.
We target the posterior expected value of the final horizontal position of the aircraft. This is by far the most complicated example among those considered here; we provide a detailed description of its coding in terms of a PPG at the end of this section.

\item \emph{Drunk man and mouse} (DMM), Example \ref{ex:dm1}.
We target the posterior expected value of the drunk man variance.

\item \emph{Hare and tortoise} (HT), see e.g. \cite{Bagnall}. This model  simulates a race between a hare and a tortoise along a one-dimensional line: the tortoise takes a step of length 1 every time step, while the hare occasionally takes a step whose length is Gaussian-distributed.   Additionally, at each time step it is observed that the hare and the tortoise are never at a distance more than 10   from each other. The race is terminated as soon as the hare overtakes the turtle. We target the posterior expected value of the final   position of the hare.

\item \emph{Bounded retransmission protocol}, \cite{5}.   A number of packets must be transmitted over a lossy channel, and each packet can be lost with a probability of $0.02$. Losses can be observed only during the transmission of the last $80$ packets.
The transmission is considered successful if none of the packets needs more than $4$ retransmissions.  We target the posterior expected value of failure probability.

\item \emph{Non-i.i.d. loops}, \cite{5}.  This   model describes the behaviour of a discrete sampler that keeps tossing two fair coins, until they both turn tails. Additionally, it is observed that  at each iteration at least one of the coins yields the same outcome as in the previous iteration. This observation  induces  data
dependencies across consecutive loop iterations.
We target the posterior expected value of  the  number of iterations until termination.
\end{enumerate}



For each model, we draw   samples of $N$  particles ($N\in R=\{10^3,10^4,10^5,10^6\}$), and the corresponding weights, with each of the considered tools/algorithms. With the drawn samples and weights, the tools compute the (posterior) expected value of the quantities of interest. For each tool, we are interested in assessing:
\begin{itemize}
\item the accuracy of the computed expected values;
\item the quality of the drawn samples;
\item the performance in terms of execution time.
\end{itemize}
Concerning accuracy,    we do not know the exact value of the targeted expected values (but in one case, see below), so a direct comparison is not possible. \revhighlight{The PF consistency result recalled in subsection \ref{sub:algo} concerns a fixed finite horizon $t$ and test functions satisfying its boundedness hypothesis; as $N\to\infty$, the weighted estimates converge to the corresponding finite-horizon FK expectations, not automatically to the infinite-trace semantics. For the unbounded observables of HT and NIID, those hypotheses require additional justification, so we treat their reported results as empirical estimates.} For a specific value of $N$, there is no obvious way to judge how close we are to convergence. Pragmatically, we will take  a small difference of estimation between consecutive values of $N$ in the considered range $R$, and the fact that  different tools yield  estimates very close to one another, as  an empirical evidence of convergence and accuracy.
As remarked above,  differently from the other considered tools that return   truncated point estimates, \highlight{\TSIpf\ provides finite-$N$ estimates of the theoretical lower and upper bounds of $\sem S f$ given by Theorem \ref{th:filtlift}; these are not certified finite-sample bounds}.
The upper bound will be vacuous ($+\infty$) whenever the target  $f$ is unbounded, which is the case for    HT and NIID   here. \revhighlight{For ZC, the $100$-iteration truncation does not in itself establish $\alpha_t=1$ for the original model.} Besides, examples AT and BRP are bounded loops, for which $\alpha_t=1$, hence $\beta_L=\beta_U=\sem S f$.

We  measure  empirically the quality of the drawn samples   in terms of    \emph{effective sample size (ESS)} \cite{RobertESS} of the corresponding weights $W_1,...,W_N$:
\vspace{-0.2cm}
$$ESS:=\frac{(\sum_{i=1}^N W_i)^2}{\sum_{i=1}^N W^2_i}\,.
\vspace{-0.2cm}$$
ESS is an empirical measure of efficiency of the sampled particles, the higher the better.
Specifically, ESS quantifies the number of i.i.d. samples from the target distribution that would be required to achieve the same variance in the estimator as that obtained from the weighted samples. So a ESS close to $N$ indicates that the $N$ particles appear  to be drawn i.i.d. from the target distribution.

The experiments  have been run on a  2.8 GHz Intel Core i7 PC, with 16GB RAM and  Nvidia T500 GPU. \TSIpf\ and webPPL have been run under Windows 10 OS, with CUDA Toolkit v. 11.8, driver v. 522.06. CorePPL has been run under Ubuntu 22, with CUDA Toolkit v. 12.2, driver v. 535.86.

In Table \ref{tab:table1} (Section \ref{sec:experiments}), we report the execution time, the estimated expected value and the effective sample size for \TSIpf, CorePPL and webPPL, as the number $N$ of particles increases.
In the case of   \TSIpf,    a single value estimate is reported  for all examples but   DMM:   for programs AT and BRP this is an estimate of $\sem S f$; for   examples HT and NIID (unbounded $f$), this is an estimate of the lower bound $\beta_L$, being the upper bound   vacuous as discussed above.
We also remark  that for NIID, it is known that $\sem S f = \frac{24}7 = 3.428\cdots$ \cite{5}. We   note that, at least for $N\geq 10^5$,   the tools tend to yield very similar estimates of the expected value\footnote{An exception is represented by the result returned by CorePPL for the NIID example.  The results of the other three tools agree with each other and with the exact value $\frac{24}7=3.428\cdots$, though.}, as a consequence of the asymptotic consistency of   PF and other SMC algorithms:  we take this  as an empirical evidence of accuracy.
%




We end this section with an explicit description of the PPG for AT.
An aircraft is modeled as a point moving on a 2D plan   according to a Gaussian process, for $t=1,...,8$ discrete time instants. Throughout these time points, the airplane  is tracked by six radar. 
Each radar is characterized by a   radius:  at each time,  if the aircraft is within the radar's radius, the radar returns the noisily measured   distance from the aircraft, otherwise
the radar just returns     a noisy version of its own radius.
We aim to infer the final horizontal position of the aircraft, i.e. the   value of $x$ at time $t=8$, conditioned  on     actual  observed data  obtained from  the six radars at all eight time instants. 

In the PPG below, $o_{ij}$ is the   observed distance at time $i$ from radar $j$, for $1\leq i\leq 8$ and $1\leq j\leq 6$, while $(rx_j,ry_j)$ and $r_j$  are the coordinates and radius of   radar $j$, respectively.  The actual numerical data can be found in \cite{WuEtAl}. Moreover, $B(p)$ is the Bernoulli distribution of parameter $p$, while $N_T(a,b,c,d)$ represents the Normal density of mean $c$ and standard deviation $d$ truncated at $[a,b]$; $N_T(a,b,c,d)(z)$ is the value at $z$ of this density.

\begin{center}
\tikzset{every picture/.style={line width=0.75pt}} 

\begin{tikzpicture}[x=0.75pt,y=0.75pt,yscale=-1,xscale=1]
	
	\draw   (108.39,155.28) .. controls (108.49,148.04) and (114.23,142.26) .. (121.2,142.37) .. controls (128.17,142.48) and (133.73,148.44) .. (133.62,155.69) .. controls (133.52,162.93) and (127.78,168.71) .. (120.81,168.6) .. controls (113.84,168.49) and (108.28,162.53) .. (108.39,155.28) -- cycle ;
	\draw    (133.62,155.69) -- (188.63,156.26) -- (214.5,156.26) ;
	\draw [shift={(217.5,156.26)}, rotate = 180] [fill={rgb, 255:red, 0; green, 0; blue, 0 }  ][line width=0.08]  [draw opacity=0] (8.93,-4.29) -- (0,0) -- (8.93,4.29) -- cycle    ;
	\draw   (217.5,156.26) .. controls (217.5,148.25) and (223.74,141.76) .. (231.45,141.76) .. controls (239.15,141.76) and (245.4,148.25) .. (245.4,156.26) .. controls (245.4,164.26) and (239.15,170.76) .. (231.45,170.76) .. controls (223.74,170.76) and (217.5,164.26) .. (217.5,156.26) -- cycle ;
	\draw    (245.4,156.26) .. controls (283.51,142.16) and (228.53,105.77) .. (231.13,139.06) ;
	\draw [shift={(231.45,141.76)}, rotate = 261.3] [fill={rgb, 255:red, 0; green, 0; blue, 0 }  ][line width=0.08]  [draw opacity=0] (8.93,-4.29) -- (0,0) -- (8.93,4.29) -- cycle    ;
	\draw    (231.45,170.76) -- (231.91,218.87) ;
	\draw [shift={(231.93,220.87)}, rotate = 269.45] [color={rgb, 255:red, 0; green, 0; blue, 0 }  ][line width=0.75]    (10.93,-3.29) .. controls (6.95,-1.4) and (3.31,-0.3) .. (0,0) .. controls (3.31,0.3) and (6.95,1.4) .. (10.93,3.29)   ;
	\draw   (217.98,235.37) .. controls (217.98,227.36) and (224.22,220.87) .. (231.93,220.87) .. controls (239.63,220.87) and (245.88,227.36) .. (245.88,235.37) .. controls (245.88,243.37) and (239.63,249.87) .. (231.93,249.87) .. controls (224.22,249.87) and (217.98,243.37) .. (217.98,235.37) -- cycle ;
	\draw   (220.18,235.37) .. controls (220.18,228.92) and (225.44,223.7) .. (231.93,223.7) .. controls (238.42,223.7) and (243.68,228.92) .. (243.68,235.37) .. controls (243.68,241.81) and (238.42,247.04) .. (231.93,247.04) .. controls (225.44,247.04) and (220.18,241.81) .. (220.18,235.37) -- cycle ;
	\draw    (245.88,235.37) .. controls (277.13,234.49) and (255.8,198.14) .. (242.82,223.51) ;
	\draw [shift={(241.64,226.04)}, rotate = 292.96] [fill={rgb, 255:red, 0; green, 0; blue, 0 }  ][line width=0.08]  [draw opacity=0] (8.93,-4.29) -- (0,0) -- (8.93,4.29) -- cycle    ;
	\draw  [xshift=28,line width=0.75]  (436.23,186.42) .. controls (440.9,186.3) and (443.17,183.91) .. (443.04,179.24) -- (442.39,154.08) .. controls (442.22,147.41) and (444.46,144.02) .. (449.13,143.9) .. controls (444.46,144.02) and (442.04,140.75) .. (441.87,134.09)(441.95,137.09) -- (441.41,116.4) .. controls (441.29,111.73) and (438.9,109.46) .. (434.23,109.58) ;
	\draw  [dash pattern={on 0.84pt off 2.51pt}]  (238.23,162.42) -- (288.23,198.58) ;
	
	\draw (143.34,105.82) node [anchor=north west][inner sep=0.75pt]  [font=\scriptsize,rotate=-359.58,xslant=0]  {$ \begin{array}{l}
			x\sim N( 2,1) ;\\
			y\sim N( -1.5,1) ;\\
			t:=1
		\end{array}$};
	\draw (253.69,55.58) node [anchor=north west][inner sep=0.75pt]  [font=\scriptsize,rotate=-359.84,xslant=0]  {$ \begin{array}{l}
			[ 1\leqslant t\leqslant 8] ,\\
			x\sim N( x,2) ;\\
			y\sim N( y,2) ;\\
			t:=t+1\\
			let\ d=d(( x,y) ,( rx_{j} ,ry_{j})) \ in\ \\
			\ \ \ \ \ d_{j} \sim \ if\ d > r_{j} \ then\ \\
			\ \ \ \ \ \ \ \ \ \ \ \ \ \ \ \ \ \ \ if\ B( .999) \ then\ r_{j} \ \\
			\ \ \ \ \ \ \ \ \ \ \ \ \ \ \ \ \ \ \ \ \ \ \ \ \ \ else\ r_{j} +N_{T}( 0,1,0,r_{j}) \ \\
			\ \ \ \ \ \ \ \ \ \ \ \ \ \ \ \ \ \ \ \ else\ d+N_{T}( 0,1,0,r_{j}) \ ;
		\end{array}$};
	\draw (175.62,184.52) node [anchor=north west][inner sep=0.75pt]  [font=\scriptsize,rotate=-359.84,xslant=0]  {$[ t\notin [ 1,8]]$};
	\draw (116.23,147.42) node [anchor=north west][inner sep=0.75pt]    {$0$};
	\draw (226.55,148.29) node [anchor=north west][inner sep=0.75pt]    {$1$};
	\draw (226.55,227.62) node [anchor=north west][inner sep=0.75pt]    {$2$};
	\draw  [dash pattern={on 0.84pt off 2.51pt}]  (290.23,199.58) -- (503.23,199.58) -- (503.23,248.58) -- (290.23,248.58) -- cycle  ;
	\draw (293.23,211.98) node [anchor=north west][inner sep=0.75pt]  [font=\scriptsize]  {$\gamma =\left[ \ \sum {_{i=1}^{8}}[ t=i] \cdot \prod _{j=1}^{6} N( o_{ij} ,.01)( d_{j})\right] \ \ \ \ \ $};
	\draw (491,138.4) node [anchor=north west][inner sep=0.75pt]  [font=\scriptsize]  {$( j=1,...6)$};

\end{tikzpicture}
\end{center}
	
Finally, we report an enlarged version of the plots of Fig. \ref{fig:scatterplot}.

\begin{figure}[ht]
	\centering
	
	\includegraphics[width=0.48\textwidth,height=5.9cm]{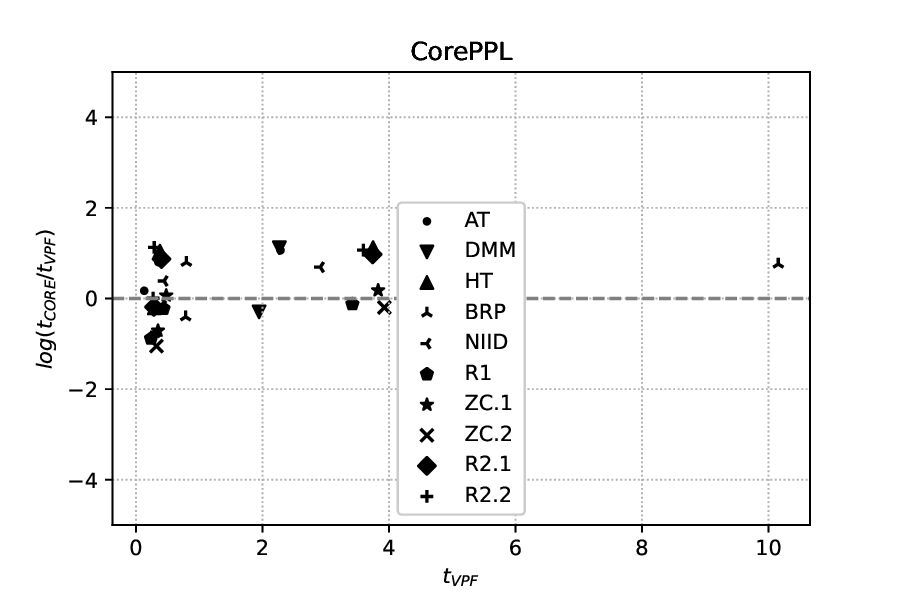}
	\hfill
	\includegraphics[width=0.48\textwidth,height=5.9cm]{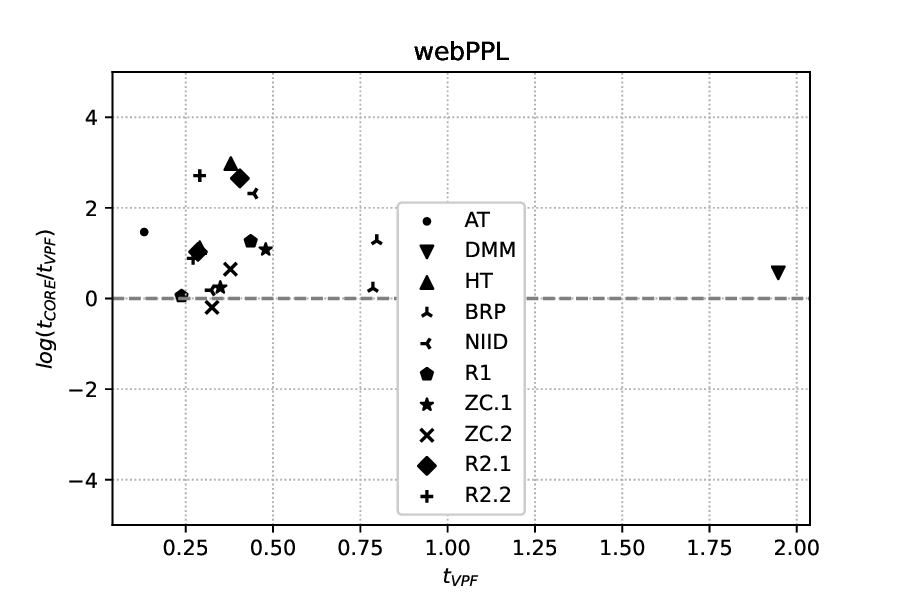}
	
	\vspace{0.3cm}
	
	\includegraphics[width=0.48\textwidth,height=5.9cm]{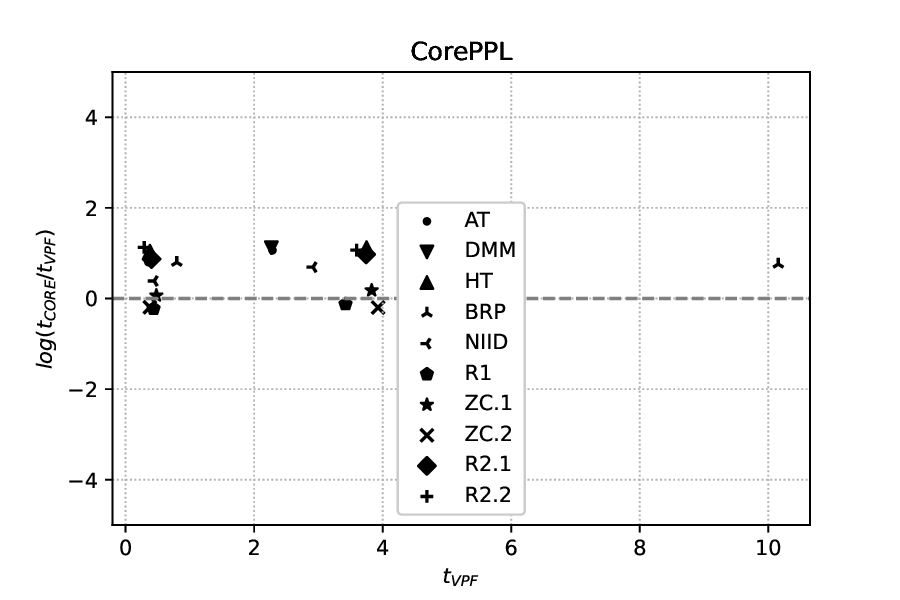}
	\hfill
	\includegraphics[width=0.48\textwidth,height=5.9cm]{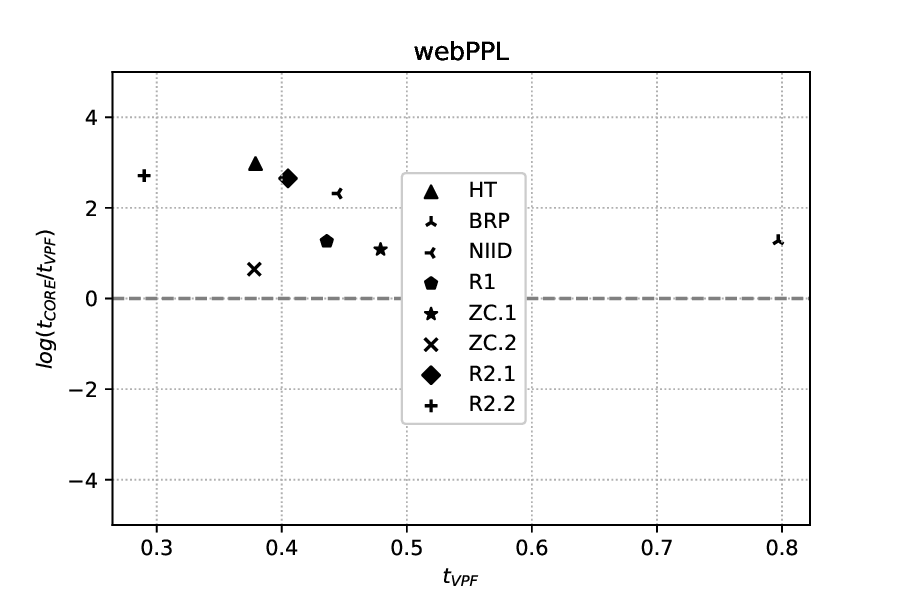}
	
	\caption{
		Enlarged version of plots in Fig. \ref{fig:scatterplot}.
	}
	\label{fig:scatterplot2}
	
\end{figure}

\ifmai
\begin{center}{
	\begin{sidewaystable}[]
		{ \renewcommand{\arraystretch}{1.7}		
			\centering
			
			\resizebox{\textwidth}{!}{ \large
				 \begin{tabular}{|c|c||c|c|c|c||c|c|c|c||c|c|c|c||c|c|c|c||c|c|c|c|}
					\hline
					\multicolumn{2}{|c||}{\multirow{2}{*}{}} & \multicolumn{4}{c||}{\textbf{AT}} & \multicolumn{4}{c||}{\textbf{DMM}} & \multicolumn{4}{c||}{\textbf{HT}} & \multicolumn{4}{c||}{\textbf{BRP}} & \multicolumn{4}{c|}{\textbf{NIID}} \\
					\cline{3-22}
					\multicolumn{2}{|c||}{}& \small{\textbf{VPF}} & \small{\textbf{CorePPL}} & \small{\textbf{RootPPL}} & \small{\textbf{webPPL}} &\small{\textbf{VPF}}& \small{\textbf{CorePPL}} & \small{\textbf{RootPPL}} & \small{\textbf{webPPL}} & \small{\textbf{VPF}} & \small{\textbf{CorePPL}} & \small{\textbf{RootPPL}} & \small{\textbf{webPPL}} &\small{\textbf{VPF}} & \small{\textbf{CorePPL}} & \small{\textbf{RootPPL}} & \small{\textbf{webPP}L} & \small{\textbf{VPF}} & \small{\textbf{CorePPL}} & \small{\textbf{RootPPL}} & \small{\textbf{webPPL}} \\
					\hline
					
					\multirow{3}{*}{$N=10^3$}
					&\textit{\scriptsize time}
					&\textbf{0.009}&0.014&0.034&0.190
					&0.283&\textbf{0.016}&0.184&0.140
					&0.274&\textbf{0.015}&0.159&0.152
					&0.676&\textbf{0.021}&1.014&0.155
					&0.240&\textbf{0.010}&0.150&0.061 \\
					\cline{2-22}
					&\textit{\scriptsize EV}
					&6.805&6.955&8.653&6.696
					&0.436$\pm 0.057$&0.502&0.432&0.431
					&32.834&33.683&33.988&32.368
					&0.018&0.016&0.024&0.023
					&3.594&2.694&3.500&3.473 \\
					\cline{2-22}
					&\textit{\scriptsize ESS}
					&\textbf{1000}&\textbf{1000}&\textbf{1000}&{999}
					&\textbf{991.0}&817.63&900.0&25.81
					&\textbf{955.0}&758.9&821.0&951.1
					&\textbf{1000}&\textbf{1000}&\textbf{1000}&974.5
					&\textbf{1000}&846.6&891.9&726.9 \\
					\hline

					\multirow{3}{*}{$N=10^4$}
					&\textit{\scriptsize time}
					&\textbf{0.131}&0.194&18.806&3.842
					&0.275&\textbf{0.178}&34.325&2.693
					&0.290&\textbf{0.180}&1.352&3.839
					&0.786&\textbf{0.309}&-&1.328
					&0.323&\textbf{0.058}&9.777&0.490 \\
					\cline{2-22}
					&\textit{\scriptsize EV}
					&6.817&6.967&6.168&6.760
					&0.565$\pm 0.055$&0.507&0.435&0.448
					&32.725&33.474&33.581&32.702
					&0.029&0.025&-&0.024
					&3.364&2.766&3.395&3.417 \\
					\cline{2-22}
					&\textit{\scriptsize ESS}
					&$\mathbf{10^4}$&$\mathbf{10^4}$&{9999}&\textit{9975}
					&\textbf{9908.0}&7798.8&8234.0&9963.o
					&\textbf{9445.0}&7692.9&7795.0&9476.2
					&$\mathbf{10^4}$&$\mathbf{10^4}$&-&9745.8
					&$\mathbf{10^4}$&8555.6&9283.0&7560.5 \\
					\hline

					\multirow{3}{*}{$N=10^5$}
					&\textit{\scriptsize time}
					&\textbf{0.354}&{2.252}&-&-
					&\textbf{0.529}&3.833&-&154.878
					&\textbf{0.379}&4.225&-&361.792
					&\textbf{0.797}&5.010&-&15.038
					&\textbf{0.445}&1.083&29.455&92.419 \\
					\cline{2-22}
					&\textit{\scriptsize EV}
					&6.818&6.970&-&-
					&0.506$\pm 0.066$&0.506&-&0.450
					&33.128&33.545&-&32.560
					&0.024&0.025&-&0.026
					&3.467&2.772&3.411&3.430 \\
					\cline{2-22}
					&\textit{\scriptsize ESS}
					&$\mathbf{10^5}$&9.9\text{e}$\mathbf{10^5}$&-&-
					&\textbf{99380.0}&78119.9&-&1541.07
					&94856.0&77243.0&-&\textbf{94881.29}
					&$\mathbf{10^5}$&$\mathbf{10^5}$&-&97484.0
					&$\mathbf{10^5}$&85609.7&92708.9&75809.5 \\
					\hline

					\multirow{3}{*}{$N=10^6$}
					&\textit{\scriptsize time}
					&\textbf{2.286}&26.481&-&-
					&\textbf{3.879}&46.420&-&-
					&\textbf{3.749}&49.493&-&-
					&\textbf{10.155}&58.448&-&-
					&\textbf{2.916}&14.323&-&- \\
					\cline{2-22}
					&\textit{\scriptsize EV}
					&6.834&6.980&-&-
					&0.481$\pm 0.068$&0.499&-&-
					&33.432&33.606&-&-
					&0.024&0.025&-&-
					&3.413&2.774&-&- \\
					\cline{2-22}
					&\textit{\scriptsize ESS}
					&$\mathbf{10^6}$&9.9\text{e}$10^5$&-&-
					&\textbf{993045.9}&778664.70&-&-
					&\textbf{947641.9}&771932.9&-&-
					&$\mathbf{10^6}$&$\mathbf{10^6}$&-&-
					&$\mathbf{10^6}$&855989.0&-&- \\
					\hline
				\end{tabular}
			}
		}
		\caption{Execution time ($time$) in seconds, estimated expected value ($EV$) and effective sample size ($ESS=(\sum_{i=1}^N W_i)^2/(\sum_{i=1}^N W^2_i)$ as the number of particles ($N$) increases, for \TSIpf, CorePPL, RootPPL and webPPL, when applied on Aircraft tracking (AT), Drunk man and mouse (DMM), Hare and tortoise (HT), Bounded retransmission protocol (BRP) and Non-i.i.d. loops (NIID). For \TSIpf,  with reference to Theorem \ref{th:filtlift}  we have $EV=\beta_L=\beta_U$ (as $\alpha_t=1$).  RootPPL and webPPL did not produce any result  with some input configurations due to insufficient memory,  or because reaching the timeout. We set a timeout of  $500$s. Best results for $time$ and $ESS$ marked in boldface.}
		\label{tab:table1}
	\end{sidewaystable}	
}\end{center}
\fi
\else
\fi

\section{Probabilistic programs pseudo-code}\label{app:models}
\vspace{-0.2cm}
We consider the following probabilistic models described for example in \cite{introGA}. For convenience, the programs are described in the language of \cite{introGA}, which is based on sequential composition, but they are easy to translate into our language. The \verb"return" statement at the end of each program describes the function $f$ considered in the estimation of $\sem{S}f$; e.g. the DMM example, \verb"return r" means $f$  is the lifting of the function $(d,r,x,y,S)\mapsto r$.
\revhighlight{In the ZeroConf pseudo-code below, the condition $\mathtt{curprobe}<100$ bounds only one probing attempt: $\mathtt{curprobe}$ is reset after a restart. Therefore it does not bound the total loop length. The experimental time cutoff, discussed above, is separate from this guard.}
\newcommand\numbered{\arabic{VerbboxLineNo}:\hspace{1ex}}
\begin{ssmall}
\begin{myverbbox}[\numbered]
{\vtheta}
while(time<=8){
    float[] radius;
    float obs-dist;
    if(time==0){
        x = Gaussian(2,1);
        y = Gaussian(-1.5,1);
    }else{
        x = Gaussian(x,2);
        y = Gaussian(y,2);}
    for i in (0,6]{
        d= compute-distance(i,x,y);
        if(d>radius[i]){
            flag=Bernoulli(0.999);
            if (flag==true){
                obs-dist=radius[i];
            }else{
                obs-dist=radius[i]+0.001*
                trunc-gauss(0,radius[i]));	
           }
        }else{
            obs-dist=d+0.1*trunc-gauss(0,radius[i]);
        }
        obs-dist1 = Gaussian(obs-dist,0.01);
        observe(obs-dist1==...); //evidence numbers
        omitted
    }
    time=time+1;
}
return x
\end{myverbbox}
\end{ssmall}

\begin{ssmall}
\begin{myverbbox}[\numbered]
{\vbeta}
d=uniform(0,2);
r=uniform(0,1);
x=-1;
y=1;
while(|x-y|<1/10){
    x=Gaussian(x,d);
    y=Gaussian(y,r);
    observe(|x-y|<=3);
}
return r;

\end{myverbbox}
\end{ssmall}

\begin{ssmall}
\begin{myverbbox}[\numbered]
{\vgamma}
initialPos=uniform(0,10);
tortoise=initialPos;
hare=0;
n=0;
while(hare<tortoise){
    n=n+1;
    tortoise=tortoise+1
    flag=Bernoulli (2/5);
    if (flag==true){
        hare=hare+Gaussian(4,2);
    }
    observe(|hare-tortoise|<=10);
}
observe((n>=20));
return hare;

\end{myverbbox}
\end{ssmall}

\begin{ssmall}
\begin{myverbbox}[\numbered]
{\valpha}
initialPos=uniform(0,10);
s=100;
f=0;
t=0;
n=0;
while(s>=0 && f<=4 && t<=280){
    t=t+1;
    flag=Bernoulli(0.2);
    if (flag==1){
        f=f+1;	
        n=n+1;
        observe((s<=80));
    }else{
        f=0;
        s=s-1;
}}
	return s>0;		
\end{myverbbox}
\end{ssmall}


\begin{ssmall}
\begin{myverbbox}[\numbered]
{\vt}
a0=1;
b0=1;
c0=1;
d0=1;
n=0;	
while((a0==1||b0==1)){
    a1=Bernoulli(0.5);
    b1=Bernoulli(0.5);
    observe(c0==a1 || d0==b1);
    c1=a1;
    d1=b1;
    n=n+1;
}	
return n
\end{myverbbox}
\end{ssmall}

\begin{ssmall}
\begin{myverbbox}[\numbered]	
{\vb}
r=uniform(0,1);
y=0;
n=0;	
while(|y|<1 && (n<=100)){
    y=Gaussian(y,2r);
    n=n+1;
}	
observe(n>=3);
return r;

\end{myverbbox}
\end{ssmall}

\begin{ssmall}
\begin{myverbbox}[\numbered]
{\vg}
p=uniform(0,1);
start=1;
curprobe,established=0;
while(curprobe < 100 && established <=0 && start <= 1){
    if(start = 1){
        flag=Bernoulli(p)
        if (flag==false){
        	established=1;    		
        }
        start=0;
    }else{
        flag=Bernoulli(lambda)
        if (flag==true){
            curprobe := curprobe + 1;
        }else{
            observe(curprobe>=20);
            start=1;
            curprobe=0;
 }}}
return p;
\end{myverbbox}
\end{ssmall}

\begin{ssmall}
\begin{myverbbox}[\numbered]
{\vh}
var=uniform(0,7);
y=1;
prob=0.5;
i=0;
while(i <= 100){
    oldy=y;
    y=Gaussian(oldy,2*var);
    flag=Bernoulli(lambda);
    if(flag){
       observe(|y-oldy|<2)};
    i=i+1;}
return y;

\end{myverbbox}
\end{ssmall}

\begin{figure}[!h]
	\centering
	
	\begin{subfigure}[t]{0.48\textwidth}
		\vtheta
		\caption{Aircraft Tracking (AT).}
	\end{subfigure}
	\hfill
	\begin{subfigure}[t]{0.48\textwidth}
		\vbeta
		\caption{Drunk Man and Mouse (DMM).}
	\end{subfigure}
	
	\vspace{0.5cm}
	
	\begin{subfigure}[t]{0.48\textwidth}
		\vgamma
		\caption{Hare and Tortoise (HT).}
	\end{subfigure}
	\hfill
	\begin{subfigure}[t]{0.48\textwidth}
		\valpha
		\caption{Bounded Retransmission Protocol (BRP).}
	\end{subfigure}
	
	\vspace{0.5cm}
	
	\begin{subfigure}[t]{0.48\textwidth}
		\vt
		\caption{Non-i.i.d. loops (NIID).}
	\end{subfigure}
	\hfill
	\begin{subfigure}[t]{0.48\textwidth}
		\vb
		\caption{Random Walk (RW1).}
	\end{subfigure}
	
	\label{fig:models}
\end{figure}

$$ $$
\begin{figure}[H]
	\ContinuedFloat
	\centering
	
	\begin{subfigure}[t]{0.48\textwidth}
		\vg
		\caption{ZeroConf (ZC.1: $\lambda=0.99$, ZC.2: $\lambda=0.5$).}
	\end{subfigure}
	\hfill
	\begin{subfigure}[t]{0.48\textwidth}
		\vh
		\caption{Random Walk (RW2.1: $\lambda=0.5$, RW2.2: $\lambda=0.9999$).}
	\end{subfigure}
	
	\label{fig:models2}
\end{figure}

\end{document}